\documentclass[abstract=true]{scrartcl}
\usepackage{epic,latexsym,amsbsy,amsmath,amscd,amssymb,textcomp,gensymb}
\usepackage{authblk}
\usepackage{amsthm}
\usepackage{mathtools}
\usepackage{bbold}
\usepackage[labelformat=simple]{subcaption}

\usepackage{graphicx}
\usepackage{tikz}
\usetikzlibrary{patterns, shapes, knots}
\usetikzlibrary{positioning,arrows}
\usetikzlibrary{decorations.markings}
\usetikzlibrary{decorations.pathmorphing}
\usetikzlibrary{arrows.meta}
\usetikzlibrary{calc,math}
\tikzstyle{every node}=[font=\small]
\usepackage{catchfile}
\newcommand\loaddata[1]{\CatchFileDef\loadeddata{#1}{\endlinechar=-1}}
\usepackage{ifthen}
\usepackage[font=small]{caption}
\usepackage{hyperref}
\usepackage{enumitem}
\setlist[enumerate]{itemsep=0pt}
\setlist[itemize]{itemsep=0pt}

\usepackage{todonotes}

\pgfmathdeclareoperator{!=}{notequal}{2}{infix} {250}

\usepackage[utf8]{inputenc}

\newcommand{\RNum}[1]{\scriptsize \uppercase\expandafter{\romannumeral #1\relax}}

\newtheorem{thm}{Theorem}[section]
\newtheorem{pro}[thm]{Proposition}
\newtheorem{lem}[thm]{Lemma}
\newtheorem{cor}[thm]{Corollary}
\newtheorem{opr}[thm]{Open Problem}

\theoremstyle{definition}
\newtheorem{dfn}[thm]{Definition}
\newtheorem{exa}[thm]{Example}
\newtheorem{rmk}[thm]{Remark}
\newtheorem{quest}[thm]{Question}
\newtheorem{conj}[thm]{Conjecture}

\newcounter{claimcount}
\newenvironment{clm}{\medskip\par\noindent\refstepcounter{claimcount}Claim~\arabic{claimcount}.}{\par\medskip}

\usepackage{etoolbox}
\AtBeginEnvironment{proof}{\setcounter{claimcount}{0}}

\newcommand\Ccal{\mathcal{C}}

\newcommand\one{\mathbf{1}}
\newcommand\two{\mathbf{2}}
\newcommand\three{\mathbf{3}}
\newcommand\four{\mathbf{4}}
\newcommand\five{\mathbf{5}}
\newcommand\six{\mathbf{6}}

\newcommand\good{good}

\newcommand{\initcube}{
\tikzmath{
\posx =0; \posy = 0;
\x1 = 1; \y1 = 0; \z1 = 0; 
\x2 = 0; \y2 = 1; \z2 = 0; 
\x3 = 0; \y3 = 0; \z3 = 1;}
}

\newcommand{\rolldown}{
\tikzmath{
\posy = \posy - 1;
\t = \y1; \y1 = -\z1; \z1 = \t; 
\t = \y2; \y2 = -\z2; \z2 = \t; 
\t = \y3; \y3 = -\z3; \z3 = \t;}
}
\newcommand{\rollup}{
\tikzmath{\posy = \posy + 4;}\rolldown \rolldown \rolldown}

\newcommand{\flipdown}{
\tikzmath{
\posy = \posy - 1;
\t = \y1; \y1 = -\t; 
\t = \y2; \y2 = -\t; 
\t = \y3; \y3 = -\t;}
}
\newcommand{\flipup}{
\tikzmath{\posy = \posy + 2;} \flipdown} 

\newcommand{\rollright}{
\tikzmath{
\posx = \posx + 1;
\t = \x1; \x1 = -\z1; \z1 = \t; 
\t = \x2; \x2 = -\z2; \z2 = \t; 
\t = \x3; \x3 = -\z3; \z3 = \t; }
}
\newcommand{\rollleft}{\tikzmath{\posx = \posx -4;}\rollright \rollright \rollright}

\newcommand{\flipright}{
\tikzmath{
\posx = \posx + 1;
\t = \x1; \x1 = -\t; 
\t = \x2; \x2 = -\t; 
\t = \x3; \x3 = -\t;}
}
\newcommand{\flipleft}{\tikzmath{\posx = \posx - 2;}\flipright}

\newcommand{\printcube}{

\pgfmathtruncatemacro\a{\z3!=0?\x1:(\z1!=0?\x2:\x3)}
\pgfmathtruncatemacro\c{\z3!=0?\y1:(\z1!=0?\y2:\y3)}
\pgfmathtruncatemacro\b{\z3!=0?\x2:(\z1!=0?\x3:\x1)}
\pgfmathtruncatemacro\d{\z3!=0?\y2:(\z1!=0?\y3:\y1)}

\pgfmathtruncatemacro\rot{\a==0?90:0}
\pgfmathtruncatemacro\xsc{\a+\b}
\pgfmathtruncatemacro\ysc{\d-\c}

\pgfmathtruncatemacro\lett{\z1==1?3:(\z2==1?2:(\z3==1?1:(\z1==-1?4: (\z2==-1?5:(\z3==-1?6:0)))))}
\pgfmathtruncatemacro\presc{\z1+\z2+\z3}
\pgfmathtruncatemacro\prerot{\z1==1?90:(\z3==0?-90:0)}
\pgfmathtruncatemacro\col{80+20*\xsc*\ysc*\presc}
  
\path[shift = {(\posx,\posy)},fill = gray!10] (-.5,-.5)--(-.5,.5)--(.5,.5)--(.5,-.5)--cycle;
\node[text height=.5cm,text width=.5cm,color = black!\col,xscale = \xsc, yscale = \ysc,rotate=\rot,xscale = \presc, rotate = \prerot] at (\posx,\posy) {\bfseries \Large \lett};
}

\title{Folding polyominoes into cubes}

\author{Oswin Aichholzer\thanks{O. Aichholzer was partially supported by FWF (Austrian Science Fund) project DK W1230.}}
\affil{\normalsize Institute of Software Technology, Graz University of Technology, Austria}

\author{Florian Lehner\thanks{F. Lehner was partially supported by FWF (Austrian Science Fund) projects P31889.}}
\affil{\normalsize Department of Mathematics, 
University of Auckland,
New Zealand}

\author{Christian Lindorfer\thanks{C. Lindorfer was partially supported by FWF (Austrian Science Fund) projects P31889 and DK W1230.}}
\affil{\normalsize Institut f\"ur Diskrete Mathematik, 
Technische Universit\"at Graz,
Austria}

\date{\today} 

\begin{document}

\maketitle

\begin{abstract}
Which polyominoes can be folded into a cube, using only creases along edges of the square lattice underlying the polyomino, with fold angles of $\pm 90\degree$ and $\pm 180\degree$, and allowing faces of the cube to be covered multiple times? Prior results studied tree-shaped polyominoes and polyominoes with holes and gave partial classifications for these cases. 

We show that there is an algorithm deciding whether a given polyomino can be folded into a cube. This algorithm essentially amounts to trying all possible ways of mapping faces of the polyomino to faces of the cube, but (perhaps surprisingly) checking whether such a mapping corresponds to a valid folding is equivalent to the unlink recognition problem from topology.

We also give further results on classes of polyominoes which can or cannot be folded into cubes. 
Our results include
(1) a full characterisation of all tree-shaped polyominoes that can be folded into the cube 
(2) that any rectangular polyomino which contains only one simple hole (out of five different types)  does not fold into a cube,
(3) a complete characterisation when a rectangular polyomino with two or more unit square holes (but no other holes) can be folded into a cube, and
(4) a sufficient condition when a simply-connected polyomino can be folded to a cube.

These results answer several open problems of previous work and close the cases of tree-shaped  polyominoes and rectangular polyominoes with just one simple hole.
\end{abstract}

\section{Introduction}

Even in its seemingly simplest form folding a piece of paper is a fruitful source for interesting research problems. Citing from Chapter 7 (The Combinatorics of paper folding) in the book of Martin Gardner~\cite{gardner1983}:
``One of the most unusual and frustrating unsolved problems in
modern combinatorial theory, proposed many years ago by
Stanislaw M. Ulam, is the problem of determining the number
of different ways to fold a rectangular map. The map is precreased
along vertical and horizontal lines to form a matrix of
identical rectangles. The folds are confined to the creases, and
the final result must be a packet with any rectangle on top and
all the others under it.''

Surprisingly, even the version of an $1 \times n$ map is open. This is known as the stamp folding problem, as it is similar to count in how many ways a single strip of stamps can be folded along their perforated edges such that all stamps end up on top of each other.
See again~\cite{gardner1983} for a history of this problem and also Example~\ref{exa:stampfolding} below.

In a similar spirit, the following innocent looking problem was proposed in~\cite{aich18}: Which polyominoes can be folded into a cube?
More precisely, given a piece of paper in the shape of a polyomino it is asked whether it has a folded state in the shape of a unit cube. Here a polyomino is a polygon in the plane build by a collection of unit squares on the integer square lattice that are connected edge-to-edge (see the next section for more formal definitions).

In the same paper different models for folding were used to tackle the problem. In addition to foldings along the creases of the lattice of the polyomino (similar to what is used for the above mentioned map problem) it can be also allowed to fold along a $\pm 45\degree$ diagonal of a square of the polyomino, or also so-called half-grid foldings are possible. These are foldings where orthogonal and diagonal folds between half-integral points can be made. Other differences in the folding model come from the way the faces of the cube are covered. For diagonal and half-grid models it is possible that a face of the cube is only partially covered by (parts of) a square of the polyomino, and the rest is covered by other (parts of) squares. Note, however, that in all these cases it is allowed that a face is covered by multiple layers of paper (like in origami, but unlike the general polyhedron unfolding), where some of these layers might only partially cover a face. 
Another restriction to the folding model could be the requirement that the surface of the cube is made entirely from just one side of the paper. That is, if you have a paper which is on one side green and on the other side red, then only the green side is visible to the exterior of the cube.
For further variates of folding models we refer to~\cite{aich18}.

For the half-grid model it was shown in~\cite[Thm. 3]{aich18} that all polyominoes of at least ten unit squares can be folded into a unit cube. Rather recently some ot the remaining cases were covered in~\cite{czajkowski2020folding}. For example there it is shown that any tree-shaped polyomino of at least 9 squares can be folded to the cube in the half-grid model.  Table 1 of this paper contains a nice overview of previous results in the different folding models.\\

In our work we focus on the grid-folding model. There foldings are allowed only along edges of the square lattice of the polyomino and the folding angles can be $\pm 90\degree$ or $\pm 180\degree$. Thus, every face of the cube is covered by a seamless square in the folding and we allow faces to be covered multiple times, that is, by several layers of different squares of the polyomino.
Readers who are not familiar with this folding model might find it instructive to solve the three cube folding puzzles given by Nikolai Beluhov~\cite{puzzle} to see how powerful the model is.

\subsection*{Previous and new results in the grid-folding model}

Our first few results concern computational aspects of the problem. 
In \cite{aich19}, the authors observe that if a polyomino folds into a cube, then there is an adjacency-preserving map from the faces of the polyomino to the faces of the cube, which they call a \emph{consistent mapping}. Clearly, if a polyomino does not admit a consistent mapping, then it cannot be folded into a cube, giving an easy brute-force algorithm to certify non-foldability. They ask whether this algorithm can also be used to certify foldability, and whether every consistent mapping can be turned into a valid folding. While the answer to the second question turns out to be negative even in the special case of trees (see Example \ref{exa:self-intersection-fan} below), we provide a positive answer to the first question: we show that there is an algorithm which determines whether a consistent mapping can be turned into a valid folding.

The second half of this paper focuses on foldability conditions for various classes of polyominoes, namely, rectangular polyominoes with (simple) holes, tree-shaped polyominoes, and simply connected polyominoes.

Our starting point for our investigation of polyominoes with holes is a result from \cite[Thm. 1]{aich19} which states that any polyomino that contains a non-simple hole can be folded into a cube. There are five different simple (also called basic) holes that can occur inside a polyomino, see Figure~\ref{fig:simpleholes}.
In the same paper several combinations of two simple holes were given that allow the polyomino to fold into a cube. The authors also show that rectangular polyominoes with only a single square hole (\cite[Thm. 11]{aich19}) or a slit of size 1 (\cite[Thm. 15]{aich19}) and no other hole cannot fold onto the unit cube. Extending their result and answering one of their open problems we show in Section~\ref{sec:one_hole} that no rectangular polyomino with only one simple hole (one of the five simple holes shown in Figure~\ref{fig:simpleholes}) folds into a cube.

Further extending this result and generalizing Theorem~12 of~\cite{aich19} we provide a complete characterisation when a rectangular polyomino with two or more unit square holes (but no other holes) can be folded into a cube. Roughly speaking this can be done if and only if at least two of these holes lie in the same or in adjacent rows (resp. columns) and the number of columns (resp. rows) between them is odd. See Theorem~\ref{thm:two_holes} in Section~\ref{sec:two_holes} for details.

Our next group of results concerns tree-shaped polyominoes. 
A polyomino is called \emph{tree-shaped} if its dual is a tree. As usual, here the dual of a polyomino consists of a vertex for every unit square of the polyomino and an edge connects two vertices whenever the corresponding unit squares are joined edge to edge (see the next section for a proper definition).
In Section~\ref{sec:tree-shaped} we fully characterise all tree-shaped polyominoes which can be folded onto the cube. This is done by a combination of considering well-defined classes of small polyominoes, bounding boxes of different sizes, and some exceptional cases of constant size. 
See e.g. Figure~\ref{fig:4times4-unfoldable} which shows the only tree-shaped polyomino with bounding box $4 \times 4$ that can not be folded into a cube. Besides this example we show that any tree-shaped polyomino with a bounding box of size $4 \times 4$ or larger can be folded to a cube (Theorems~\ref{thm:4times4} and~\ref{thm:4times5}), and exhaustively characterise all tree-shaped polyominoes of smaller bounding size (Theorems~\ref{thm:2_times_n} and~\ref{thm:3_times_n}). A similar result for bounding boxes of sizes $2 \times n$ and $3 \times n$ was presented already in~\cite[Thm. 4]{aich18}, but it seems their characterisation was incomplete, see Section~\ref{sec:tree-shaped} for details.

Finally, in Section~\ref{sec:simply_connected} we consider simply-connected polyominoes.
A polyomino is simply connected if it does not contain any holes. Clearly, all tree-shaped polyominoes are simply-connected, and simply-connected polyominoes which are not tree-shaped have to contain at least one rectangular sub-polyomino of size at least $2 \times 2$. This makes it harder to classify foldability for simply-connected polyominoes. Using the maximal distance of a point on the boundary of the polyomino to a corner of its bounding box, we provide a sufficient condition when a simply-connected polyomino can be folded to a cube, see Theorem~\ref{thm:simply-connected}.\\

Before providing the details of these results we start by stating the underlying notation in the next section. In Section~\ref{sec:folding_model} we specify our folding model, and give some more rigorous definitions and general results on folding polyominoes. We then provide details on consistent mappings to the unit cube in Section~\ref{sec:mapping}, which will be used in the consecutive sections to derive the results mentioned above.

\section{Notation}
A polyomino is a plane geometric figure formed by joining one or more unit squares edge to edge. It may be regarded as a finite, connected subset of the regular square tiling.  We do not require a connection between every pair of adjacent squares, i.e., we allow slits along the edges of the square tiling. As usual, the dual $D(P)$ of a polyomino $P$ is constructed by taking a vertex for every unit square of the polyomino and connecting two vertices by an edge whenever the corresponding unit squares are joined edge to edge. We call a maximal set of missing squares and slits a hole in $P$, if $D(P)$ has a cycle containing $P$ in its interior. A hole is called simple, if it is one of the following: a unit square, a slit of size 1, an I-slit of size 2, an L-slit of size 2 or a U-slit of size 3.

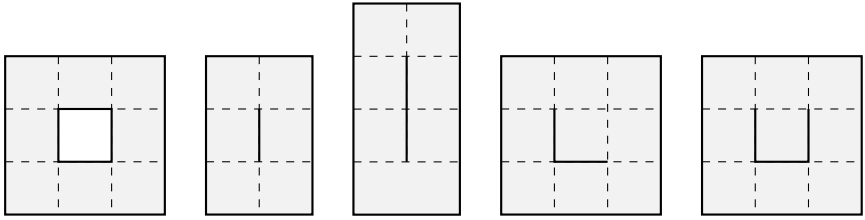
\begin{figure}[ht]
\centering
\begin{tikzpicture}[xscale=0.7,yscale=0.7]
\draw[thick,fill=gray!10] (0,1)--(0,-2)--(3,-2)--(3,1)--cycle;
\fill[white] (1,0)--(1,-1)--(2,-1)--(2,0)--cycle;
\foreach \i in {0,-1}{
    \draw[dashed] (0,\i)--(3,\i);}
\foreach \i in {1,2}{
    \draw[dashed] (\i,1)--(\i,-2);}
\draw[thick] (1,0)--(1,-1)--(2,-1)--(2,0)--cycle;
\end{tikzpicture}
\quad
\begin{tikzpicture}[xscale=0.7,yscale=0.7]
\draw[thick,fill=gray!10] (0,1)--(0,-2)--(2,-2)--(2,1)--cycle;
\foreach \i in {0,-1}{
    \draw[dashed] (0,\i)--(2,\i);}
\foreach \i in {1}{
    \draw[dashed] (\i,1)--(\i,-2);}
\draw[thick] (1,0)--(1,-1);
\end{tikzpicture}
\quad
\begin{tikzpicture}[xscale=0.7,yscale=0.7]
\draw[thick,fill=gray!10] (0,1)--(0,-3)--(2,-3)--(2,1)--cycle;
\foreach \i in {0,-1,-2}{
    \draw[dashed] (0,\i)--(2,\i);}
\foreach \i in {1}{
    \draw[dashed] (\i,1)--(\i,-2);}
\draw[thick] (1,0)--(1,-2);
\end{tikzpicture}
\quad
\begin{tikzpicture}[xscale=0.7,yscale=0.7]
\draw[thick,fill=gray!10] (0,1)--(0,-2)--(3,-2)--(3,1)--cycle;
\foreach \i in {0,-1}{
    \draw[dashed] (0,\i)--(3,\i);}
\foreach \i in {1,...,2}{
    \draw[dashed] (\i,1)--(\i,-2);}
\draw[thick] (1,0)--(1,-1)--(2,-1);
\end{tikzpicture}
\quad
\begin{tikzpicture}[xscale=0.7,yscale=0.7]
\draw[thick,fill=gray!10] (0,1)--(0,-2)--(3,-2)--(3,1)--cycle;
\foreach \i in {0,-1}{
    \draw[dashed] (0,\i)--(3,\i);}
\foreach \i in {1,...,2}{
    \draw[dashed] (\i,1)--(\i,-2);}
\draw[thick] (1,0)--(1,-1)--(2,-1)--(2,0);
\end{tikzpicture}
\caption{The five simple holes: unit square hole, slit of size 1, I-slit of size 2, L-slit of size 2, U-slit of size 3. This Figure is similar to~\cite[Fig. 3]{aich19} where these holes are called basic.}
\label{fig:simpleholes}
\end{figure}

\section{Folding model}
\label{sec:folding_model}

In order to define our folding model rigorously, we first rigorously define what a polyomino is. An \emph{abstract polyomino} is a square complex in which every edge and every vertex is contained in at least one face, and every edge is contained in at most two faces. In other words, an abstract polyomino is obtained from a disjoint collection of unit squares by identifying some of their edges in such a way that each edge gets identified with at most one other edge. Note that defining the faces to be unit squares allows us to view an abstract polyomino as a metric space, and therefore also as a topological space.


A \emph{folding blueprint} of an abstract polyomino $P$ is a map $\beta \colon P \to \mathbb R^3$ such that the restriction of $\beta$ to each face $A$ of $P$ is an isometry.
A \emph{folded state} of an abstract polyomino $P$ is a piecewise linear (topological) embedding $F \colon P \to \mathbb R^3$.
A folded state $F$ is said to \emph{$\epsilon$-correspond} to a folding blueprint $\beta$ if $\|F(x) - \beta(x)\|< \epsilon$ for every $x \in P$; in this case we call $F$ and \emph{$\epsilon$-realisation} of $\beta$.

We note that if $F$ is a folded state, then the boundary of $F(P)$ in $\mathbb R^3$ is homeomorphic to a disjoint union of circles.
To lighten notation, we fix some small $\epsilon$ for the rest of this section (for our purpose any $\epsilon < \frac 12$ will do), and write `correspond' for `$\epsilon$-correspond' and `realisation' for `$\epsilon$-realisation'.

Let us call a folding blueprint \emph{flat} if its image is contained in a plane and the images of different faces only overlap in edges. Note that not every abstract polyomino has a flat folding blueprint, but if it has one, then the flat folding blueprint is unique up to isometries of $\mathbb R^3$.

\begin{dfn}
    A \emph{polyomino} is an abstract polyomino $P$ which admits a flat folding blueprint. A realisation of a flat folding blueprint is called a \emph{trivial folding}.
\end{dfn}

Finally, we are ready to define our folding model.

\begin{dfn} \label{def:folding}
    Two folded states are called \emph{equivalent} if they are ambient isotopic. A folded state $F$ of a polyomino $P$ is called \emph{valid} if it is equivalent to a trivial folding; in this case we also call $F$ a \emph{folding} of $P$.
\end{dfn}

Intuitively speaking, the above definition states that two folded states $F_1$ and $F_2$ are equivalent, if we can `move' the image of $F_1$ around in $\mathbb R^3$ to obtain the image of $F_2$. In this process we allow `bending and stretching' the image, but we do not allow `cutting it and glueing it back together'. It is not hard to see that any two trivial foldings of a polyomino $P$ are equivalent, and therefore the same is true for any two foldings.

\begin{dfn} \label{def:foldinginto}
    Let $P$ be a polyomino, and let $Q \subseteq \mathbb R^3$. A \emph{folding of $P$ into $Q$} is a folding $F$ corresponding to a folding blueprint $\beta \colon P \to Q$. If $\beta$ is surjective, we say that $F$ is a \emph{folding of $P$ onto $Q$}.
\end{dfn}

Usually we will consider the case where $Q$ is the surface of the unit cube $\Ccal$. We note that any folding blueprint $\beta \colon P \to \Ccal$ must map faces of $P$ to faces of $\Ccal$.
Hence, if a folding $F$ corresponds to such a folding blueprint, then each face $A$ of $P$ is embedded somewhere near a face $B$ of $\Ccal$, and for every face $B$ of $\Ccal$ there is a face $A$ of $P$ embedded near it. In particular, a folding of $P$ into $\Ccal$ gives maps from the vertices, edges, and faces of $P$ to the vertices, edges, and faces of $\Ccal$. If for a folding into $\Ccal$, all faces of $\Ccal$ are covered we call this a folding onto $\Ccal$.

We finish this section by giving some necessary and sufficient conditions for valid foldings. Let us define a \emph{polyhedral sphere} as a polyhedron $Q \subseteq \mathbb R^3$ which is homeomorphic to a sphere. An \emph{unlink} is a collection of circles embedded in $\mathbb R^3$ which is ambient isotopic to a set of circles embedded side-by-side in a plane. The following theorem tells us that polyhedral spheres and unlinks can be used as certificates for valid foldings.

\begin{thm}
    \label{thm:foldedstatevalid}
    Let $F$ be a folded state of a polyomino $P$. The following are equivalent.
    \begin{enumerate}
        \item \label{itm:valid}The folded state $F$ is valid, that is, $F$ is a folding.
        \item \label{itm:sphere}There is a polyhedral sphere $Q$ such that $F(P) \subseteq Q$.
        \item \label{itm:unlink}The boundary of $F(P)$ is an unlink.
    \end{enumerate}
\end{thm}

\begin{proof}
    To see that condition \ref{itm:valid} implies the other two, we note that only allowing piecewise linear movement of the polyomino leads to an equivalent definition of valid foldings, see \cite[Theorem 6.2]{Hudson69}. Since the image of any trivial folding is contained in a polyhedral sphere and has an unlink as its boundary, the same must be true for any valid folding.

    The implication \ref{itm:sphere} to \ref{itm:valid} follows from the fact that any two polyhedral spheres are ambient isotopic, for instance by the piecewise linear Schönflies theorem \cite[Chapter 17, Theorem 12]{Moise77}.

    It remains to show the implication from \ref{itm:unlink} to \ref{itm:sphere}. For this purpose first note that every polyomino contains a (topological) disk $D$ such that $P \setminus \overline D$ is homeomorphic to a disjoint union of annuli $A_1, \dots, A_{k-1}$ where each $A_i$ contains one connected component $C_i$ of the boundary of $P$, and the last connected component $C_k$ of the boundary meets every $A_i$, see Figure \ref{fig:cutpolyomino}.

    \begin{figure}
    \centering
    \begin{tikzpicture}[xscale=0.7,yscale=0.7]
        \path[fill=gray!10] (0,1)--(3,1)--(3,0)--(7,0)--(7,3)--(5,3)--(5,5)--(2,5)--(2,4)--(0,4)--cycle;       
        \draw[dashed]
        (3,1)--(7,1)
        (0,2)--(7,2)
        (0,3)--(5,3)
        (2,4)--(5,4)
        (1,4)--(1,1)
        (2,4)--(2,1)
        (3,5)--(3,1)
        (4,5)--(4,0)
        (5,3)--(5,0)
        (6,3)--(6,0);
        \path[draw=red] (1.5,4) to [out=-50,in=90] (2.5,1) -- (3,1) to[out=80, in=-140]  (5.5,3) to (5,3) to[out=-140, in=-80] (2,4)--cycle;  
        \path[fill=white] (1,2) rectangle (2,3);
        \path[fill=white] (3,3) rectangle (4,4);
        \path[fill=white] (6,1) rectangle (4,2);
        \draw[thick]
        (0,1)--(3,1)--(3,0)--(7,0)--(7,3)--(5,3)--(5,5)--(2,5)--(2,4)--(0,4)--cycle
        (1,2) rectangle (2,3)
        (3,3) rectangle (4,4)
        (6,1) rectangle (4,2);  
        
    \end{tikzpicture}
    \caption{Cutting the polyomino along the red lines gives one connected component which is homeomorphic to a disk, and three connected components each of which is homeomorphic to an annulus.}
    \label{fig:cutpolyomino}
    \end{figure}
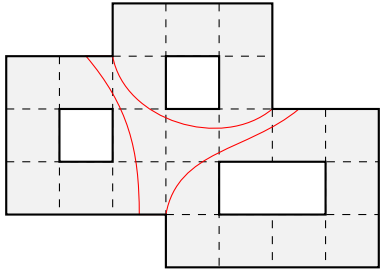
    
    Since the cycles $C_1, \dots, C_{k-1}$ form an unlink, for each $i$ there is $S_i \subseteq \mathbb R^3$ such that $S_i$ is homeomorphic to an open sphere, $F(A_i)\subseteq S_i$, and $S_i \cap S_j = \emptyset$ for $i \neq j$. We note that $F$ is ambient isotopic to a piecewise linear embedding $F'$ such that $F'(A_i) \subseteq S_i$; the corresponding ambient isotopy can be chosen to be the identity outside a small neighbourhood of $F(A_i)$. If $F'(A_i)$ for some $i$ is not ambient isotopic to an annulus embedded in the plane, then $F'(C_i) \cup F'(C_k)$ cannot be an unlink. 
    Hence we may assume that each $F'(A_i)$ is ambient isotopic to an annulus embedded in the plane, and thus we can find disks $D_i \subseteq S_i$ such that $\partial D_i = D_i \cap F'(P) = C_i$. Finally, $F'(P) \cup \bigcup_{i=1}^{k-1} D_i$ is homeomorphic to a closed disk and therefore is contained in some polygonal sphere.
\end{proof}

The last result in this section deals with the computational complexity of recognising whether a folded state is valid. Let \textsc{Folding} denote the decision problem whether a folded state of a polyomino $P$ is valid. For $Q \subseteq \mathbb R^3$, let \textsc{Folding}$(Q)$ denote the problem of deciding whether a folded state corresponding to a folding blueprint $\beta\colon P \to Q$ is valid. Finally, let \textsc{Unlink} denote the problem of deciding whether a collection of polygonal circles embedded in $\mathbb R^3$ is an unlink.

\begin{thm}
    \label{thm:folding-eq-unlink}
    \textnormal{\textsc{Folding}} and \textnormal{\textsc{Folding}}$(\Ccal)$ are polynomially equivalent to \textnormal{\textsc{Unlink}}.
\end{thm}
\begin{proof}[Proof sketch]
    \textsc{Folding}$(\Ccal)$ trivially can be reduced to \textsc{Folding}, and  \textsc{Folding} can be reduced to \textsc{Unlink} by Theorem \ref{thm:foldedstatevalid}. Hence we only need to find a reduction from \textsc{Unlink} to \textsc{Folding}$(\Ccal)$. We provide a proof sketch and leave the (easy but tedious) details to the reader.

    Every link can be represented by a diagram, that is, a projection to the $x$--$y$-plane in which for every pair of crossing lines we record which one lies `above' and which one lies `below'. Moreover (by applying a suitable ambient isotopy) we may assume that all lines in this diagram lie in an $\epsilon$-neighbourhood of the $x$-axis, each crossing lies in an $\epsilon$-neighbourhood of $x = n$ for some $n \in \mathbb N$, and no two crossings lie in a neighbourhood of the same $x$. 

    Let $L$ be an arbitrary link, represented by a link diagram as above. We can create a polyomino $P$ together with a folded state $F$ from this diagram as follows, see Figure~\ref{fig:polyominofromdiagram}. 
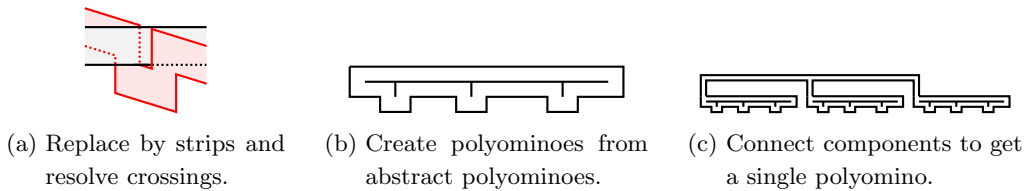
\begin{figure}
\centering

	\begin{subfigure}[t]{.25\textwidth}
	    \centering
	\begin{tikzpicture}[xscale=.8,yscale=.5]
	        \path[fill=gray!10] (-1,1)--(1,1)--(1,.5)--(0.1,1)--(0.1,0)--(-1,0)--cycle;
	        \path[fill=red!10] (-1,1.5)--(-0.1,1.05)--(-0.1,1)--(-1,1)--cycle;
	        \path[fill=red!10] (-0.1,0)--(0.1,-0.1)--(0.1,0.95)--(1,0.5)--(1,-.5)--(.5,-.25)--(.5,-1.25)--(-.5,-.75)--(-.5,0)--cycle;
	        \draw[thick,densely dotted,red,opacity=.7] (-1,.5)--(-.5,.25)--(-.5,0);
	        \draw[thick,densely dotted,red,opacity=.7] (-0.1,1)--(-0.1,0);
	        \draw[thick,densely dotted,opacity=.7] (0.1,0)--(1,0);
	        \draw[thick,red] (-1,1.5)--(-0.1,1.05)--(-0.1,1);
	        \draw[thick,red] (-.5,0)--(-.5,-.75)--(.5,-1.25)--(.5,-.25)--(1,-.5);
	        \draw[thick] (-1,1)--(1,1);
	        \draw[thick] (-1,0)--(0.1,0);
	        \draw[thick,red] (-0.1,0)--(0.1,-0.1)--(0.1,0.95)--(1,0.5);
	        \draw[thick] (-.2,0)--(.05,0);
	\end{tikzpicture}
	\caption{Replace by strips and resolve crossings.}
	\label{subfig:resolvecrossing}
	\end{subfigure}
	\quad
	\begin{subfigure}[t]{.29\textwidth}
	    \centering
	        \begin{tikzpicture}[scale=.2]
	            \draw[thick] (0,0)--(16,0);
	            \draw[thick] (2,0)--(2,-1);
	            \draw[thick] (7,0)--(7,-1);
	            \draw[thick] (13,0)--(13,-1);
	            \draw[thick] (-1,1)--(-1,-1)--(1,-1)--(1,-2)--(3,-2)--(3,-1)--(6,-1)--(6,-2)--(8,-2)--(8,-1)--(12,-1)--(12,-2)--(14,-2)--(14,-1)--(17,-1)--(17,1)--(-1,1);
	        \end{tikzpicture}
	    \caption{Create polyominoes from abstract polyominoes.}
	    \label{subfig:createpolyomino}
	\end{subfigure}
	\quad
	\begin{subfigure}[t]{.3\textwidth}
        \centering
            \begin{tikzpicture}[scale=.07]
                \draw[thick] (0,0)--(16,0);
            \draw[thick] (2,0)--(2,-1);
            \draw[thick] (7,0)--(7,-1);
            \draw[thick] (13,0)--(13,-1);
            \draw[thick] (-1,1)--(-1,-1)--(1,-1)--(1,-2)--(3,-2)--(3,-1)--(6,-1)--(6,-2)--(8,-2)--(8,-1)--(12,-1)--(12,-2)--(14,-2)--(14,-1)--(17,-1)--(17,1)--(0,1);
            \begin{scope}[xshift=20cm]
            \draw[thick] (0,0)--(16,0);
            \draw[thick] (2,0)--(2,-1);
            \draw[thick] (7,0)--(7,-1);
            \draw[thick] (13,0)--(13,-1);
            \draw[thick] (-1,1)--(-1,-1)--(1,-1)--(1,-2)--(3,-2)--(3,-1)--(6,-1)--(6,-2)--(8,-2)--(8,-1)--(12,-1)--(12,-2)--(14,-2)--(14,-1)--(17,-1)--(17,1)--(0,1);
            \end{scope}
            \begin{scope}[xshift=40cm]
                \draw[thick] (0,0)--(16,0);
                \draw[thick] (2,0)--(2,-1);
                \draw[thick] (7,0)--(7,-1);
                \draw[thick] (13,0)--(13,-1);
                \draw[thick] (-1,1)--(-1,-1)--(1,-1)--(1,-2)--(3,-2)--(3,-1)--(6,-1)--(6,-2)--(8,-2)--(8,-1)--(12,-1)--(12,-2)--(14,-2)--(14,-1)--(17,-1)--(17,1)--(0,1);
            \end{scope}

            \draw[thick] (-1,1)--(-1,5)--(40,5)--(40,1);
            \draw[thick] (0,1)--(0,4)--(19,4)--(19,1);
            \draw[thick] (20,1)--(20,4)--(39,4)--(39,1);
            \end{tikzpicture}
        \caption{Connect components to get a single polyomino.}
        \label{subfig:connectpolyomino}
        \end{subfigure}
    \caption{Creating a polyomino from a link diagram}
    \label{fig:polyominofromdiagram}
\end{figure}
First replace every line in the diagram by a surface strip of width $1$ contained in an $\epsilon$-neighbourhood of the $x$--$z$-plane, and resolve crossings by attaching gadgets to the strip as shown in Figure \ref{subfig:resolvecrossing}. This gives a collection of abstract polyominoes $P_1, \dots, P_k$ (one for each connected component of $L$) together with folded states $F_1, \dots, F_k$ of these abstract polyominoes. Each $P_i$ is homeomorphic to an annulus, and each $F_i$ corresponds to a folding blueprint $\beta_i \colon P_i \to Q$ where $Q$ is a strip in the $x$--$z$-plane. Moreover, $\bigcup_{i} F_i(P_i)$ is ambient isotopic to a side-by-side embedding of the abstract polyominoes if and only if $L$ is an unlink.

Next we obtain a polyomino $P_i'$ from the abstract polyomino $P_i$ by `cutting along a line in $z$-direction' and `reconnecting the sides of the cut by  a frame', see  Figure \ref{subfig:createpolyomino}. It is not hard to see that $P_i'$ admits a folded state $F_i'$ corresponding to a folding blueprint $\beta_i' \colon P_i' \to Q$ such that $\bigcup_{i} F_i'(P_i')$ is ambient isotopic to trivial side-by-side embeddings of the polyominoes $P_i'$ if and only if $L$ is an unlink. For instance, $\beta_i'$ may map all faces of the frame we added in this step to the same unit square in $Q$.

We obtain a polyomino $P$ by connecting the polyominoes $P_i'$ as shown in Figure~\ref{subfig:connectpolyomino}. Once again, it is straightforward to check that there is a folded state $F$ of $P$ which corresponds to a folding blueprint $\beta \colon P \to Q$ such that $F$ is valid if and only if $\bigcup_{i} F_i'(P_i')$ is ambient isotopic to trivial side-by-side embeddings of the polyominoes $P_i'$. Hence $F$ is a folding if and only if $L$ is an unlink.

Finally, observe that starting from $F$ and `folding $180\degree$ along all edges in $x$-direction' and then `folding $180\degree$ along all edges in $z$-direction' yields a folded state $F'$ corresponding to a folding blueprint mapping $P$ to a unit square (and therefore to the unit cube).
\end{proof}

\section{Consistent mappings to the unit cube}
\label{sec:mapping}

Let us now consider polyominoes and the  unit cube as (combinatorial) polygonal complexes. A \emph{consistent mapping} from a polyomino $P$ to the unit cube $\Ccal$ is a homomorphism $\varphi$ of the corresponding polygonal complexes, i.e., it maps faces to faces, edges to edges and vertices to vertices and preserves the incidence relation. We note that consistent mappings can be seen as a combinatorial version of folding blueprints; indeed, they carry precisely the same information. Up to isomorphism there is exactly one surjective consistent mapping from the standard cube net depicted in Figure \ref{fig:cubenet} to the unit cube, which can be used to define a labelling on the unit cube $\Ccal$. This labelling gives us a one-to-one correspondence between consistent mappings from a polyomino $P$ to $\Ccal$ and certain labellings of the unit squares of $P$. 

\begin{figure}[ht]
    \centering
\begin{tikzpicture}[xscale=0.7,yscale=0.7]
\initcube\printcube
\rolldown\printcube
\rollup\rollup\printcube
\rolldown\rollright\printcube
\rollleft\rollleft\printcube
\rollleft\printcube
\begin{scope}[shift={(-.5,-.5)}]
\draw[thick] (-2,0)--(0,0)--(0,-1)--(1,-1)--(1,0)--(2,0)--(2,1)--(1,1)--(1,2)--(0,2)--(0,1)--(-2,1)--cycle;
\draw[dashed] (-1,1)--(-1,0);
\draw[dashed] (0,1)--(0,0);
\draw[dashed] (1,1)--(1,0);
\draw[dashed] (1,0)--(0,0);
\draw[dashed] (1,1)--(0,1);
\end{scope}
\end{tikzpicture}
    \caption{Standard cube net}
    \label{fig:cubenet}
\end{figure}
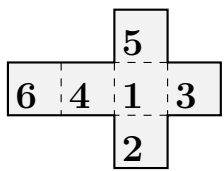

\begin{rmk}
Consistent mappings from a polyomino to $\Ccal$ are exactly the mappings yielded by the algorithm explained in \cite[section 4.3]{aich19}. An implementation in Python is freely available online, a link can be found in the mentioned paper.
\end{rmk}

Clearly every valid folding of a polyomino $P$ into a cube gives a consistent mapping in the obvious way. However, not the complete information about a specific folding can be seen from its respective consistent mapping. In particular it is not possible to distinguish $+180^\circ$ folds from $-180^\circ$ folds. This also means that if a face of $\Ccal$ is covered by several different unit squares of $P$, the consistent mapping does not provide the order in which the unit squares appear on this face. 

\begin{exa}
\label{exa:stampfolding}
    Consider a strip of some integer length $k$, that is a rectangular polyomino of size $1 \times k$.
    Then one possible consistent mapping $\varphi$ maps each unit square of $P$ onto the same face $\one$ of $\Ccal$, see Figure~\ref{fig:strip}. It is not hard to see that there are several valid foldings yielding this specific consistent mapping $\varphi$. Clearly every edge must be folded by either $+180^\circ$ or $-180^\circ$ and for every given sequence of $k-1$ angles in $\{+180^\circ, -180^\circ$\} there is a folding of $P$ such that the vertical edges are folded according to the sequence. Indeed, starting at one side of the strip we can consecutively fold each edge as determined by the sequence such that the next face lies on the outside of the `multi-layered square' containing the squares encountered up to this points. Thus we get at least $2^{k-1}$ different foldings for our strip. Counting the number of possible foldings of the strip of length $k$ into a square is known as the stamp folding problem. 
See Chapter 7 in Martin Gardners book~\cite{gardner1983} for the history of this long standing open problem.    
     Its solutions are currently known up to $k=45$ (see e.g. sequence A000136 in the online encyclopedia of integer sequences at https://oeis.org/A000136); for small $k$ they are given by the sequence
    \[
        1, 2, 6, 16, 50, 144, 462, 1392, 4536, 14060, 46310, 146376, 485914, 1557892, 5202690, \dots.
    \]

    \begin{figure}[ht]
        \centering
        \begin{tikzpicture}[xscale=0.7,yscale=0.7]
        \initcube\printcube
        \flipright\printcube
        \flipright\printcube
        \flipright\printcube
        \flipright\flipright\printcube
        \flipright\printcube
        \flipright\printcube
        \begin{scope}[shift={(-.5,-.5)}]
        \draw (0,1)--(0,0)--(8,-0)--(8,1)--cycle;
        \foreach \i in {1,...,7}{
            \draw[dashed] (\i,1)--(\i,0);}
        \node () at (4.5,0.5) {$\cdots$};
        \end{scope}
        \end{tikzpicture}
        \caption{Strip of size $1 \times k$ for even $k$}
        \label{fig:strip}
        \end{figure}
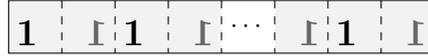
\end{exa}

This example suggests that if we want to uniquely identify a given folding, it is reasonable to also include information about the order in which the unit squares of the polyomino appear on the cube to a consistent mapping. We say that square $s$ of polyomino $P$ lies in layer $l$ of a folding of $P$ into $\Ccal$ if there are exactly $l-1$ squares of $P$ covering the same face as $s$ and lying below $s$. 
A \emph{layer map} of a polyomino $P$ for a given consistent mapping $\varphi$ is a map $\ell$ assigning to each unit square $s$ of $P$ an integer between 1 and the number of unit squares of $P$ mapped to the face $\varphi(s)$ of $\Ccal$ such that its restriction to $\varphi^{-1}(f)$ is a bijection for every face $f$ of $\Ccal$. The triple $(P,\varphi, \ell)$ is called a \emph{pseudo-folding}. It is called \emph{surjective} if $\varphi$ is surjective.

\subsection{Validity of pseudo-foldings}

Recall that by Definitions~\ref{def:folding} and \ref{def:foldinginto} a folding of a polyomino $P$ into $\Ccal$ is a realisation $F$ of a folding blueprint $\beta: P \to \Ccal$ which is equivalent to a trivial folding. Clearly any realisation of a folding blueprint $\beta: P \to \Ccal$ uniquely defines a pseudo-folding $(P,\varphi, \ell)$: The consistent mapping $\varphi$ is given by $\beta$, while the layer map $\ell$ is provided by the order of the images of the squares of $P$ under the embedding $F: P \to \mathbb R^3$ for each of the faces of the cube $\Ccal$. In this case we say that the folded state $F$ \emph{realises} the pseudo-folding $(P,\varphi, \ell)$, even if it is not equivalent to a trivial folding and thus no folding.

\begin{dfn}
    A pseudo-folding is \emph{valid}, if it corresponds to a folding of $P$ into $\Ccal$.
\end{dfn}

For the rest of this section we want to discuss obstacles preventing pseudo-foldings from being valid. Instead of giving rigorous definitions for these obstacles, we decided to present them by discussing simple examples where those occur. According to the above discussion, there are two things that can go wrong: A given pseudo-folding might not have any realisation or none of its realisations is equivalent to a trivial folding. 

Let us start with the first issue and present two examples of consistent mappings which cannot have any realisations. This is due to the fact that any possible layer map forces self-intersections and thus cannot be realised by any embedding $F$.

\begin{exa}\label{exa:non-foldable-unit-square-hole}
Let us consider the consistent mapping in Figure~\ref{fig:notfoldablehole}, which was implicitly already discussed in the proof of \cite[Lem. 9]{aich19}. In any valid folding into $\Ccal$, whenever two unit squares with the same label have neighbours with the same label, then the order of their respective layers must be retained, otherwise there would be self-intersections. In particular, if the $\one$ in the top left corner is folded on top of the $\one$ in the bottom right corner, then the top $\three$ is folded on top of the bottom $3$ and continuing this argument along the cycle $\one-\three-\six-\two-\one-\three-\six-\two-\one$ leads to a contradiction.

\begin{figure}[ht]
    \centering
\begin{tikzpicture}[xscale=0.7,yscale=0.7]
\initcube\printcube
\rolldown\printcube
\rolldown\printcube
\rollright\printcube
\rollright\printcube
\rollup\printcube
\rollright\printcube\rollleft
\rollup\printcube
\rollup\printcube\rolldown
\rollleft\printcube
\begin{scope}[shift={(-.5,-.5)}]
\draw[thick] (0,1)--(2,1)--(2,2)--(3,2)--(3,0)--(4,0)--(4,-1)--(3,-1)--(3,-2)--(0,-2)--cycle;
\draw[thick] (1,0)--(1,-1)--(2,-1)--(2,0)--cycle;
\draw[dashed] (0,0)--(1,0);
\draw[dashed] (2,0)--(3,0);
\draw[dashed] (0,-1)--(1,-1);
\draw[dashed] (2,-1)--(3,-1);

\draw[dashed] (1,1)--(1,0);
\draw[dashed] (2,1)--(2,0);
\draw[dashed] (1,-1)--(1,-2);
\draw[dashed] (2,-1)--(2,-2);

\draw[dashed] (2,1)--(3,1);
\draw[dashed] (3,0)--(3,-1);
\end{scope}
\end{tikzpicture}
    \caption{A surjective consistent mapping that cannot be realised by any folded state}
    \label{fig:notfoldablehole}
\end{figure}
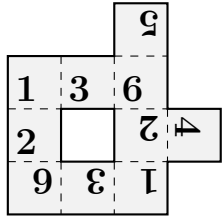
\end{exa}

\begin{exa}\label{exa:self-intersection-fan}
To show that there is no valid folding for the consistent mapping shown in Figure~\ref{fig:notfoldablefan}, consider the order of the five unit squares in the third (and fourth, fifth) column in which they are folded on top of each other. Note that any row ending with a T-shape cannot lie between two rows not having a `T' in the same direction in this order. This is the case because the T-part would lie between those two layers and thus self-intersection would occur. In particular row 1 and 5 cannot be placed between two other rows and the same holds for 2 and 4. Thus without loss of generality rows 1 and 5 are folded on top of all other rows, and row 2 and 4 are folded on bottom. It is now easy to check that none of these permutations give a valid folding. Starting with row 1, row 2 has to be folded downwards, then row 3 upwards, row 4 again downwards and thus ends up between 2 and 3. In particular row 5 cannot be folded upwards to end up next to row 1.

\begin{figure}[ht]
    \centering
\begin{tikzpicture}[xscale=0.7,yscale=0.7]
    \initcube\flipleft\rollleft\rollup \printcube
    \rolldown\printcube
    \rolldown\printcube\rollup
    \rollright\printcube
    \flipright\printcube
    \rollright\printcube
    \rollright\printcube
    \rollright\printcube \rollleft
    \flipdown\printcube

    \flipright\printcube
    \rollright\printcube
    \rollup\printcube\rolldown
    \rolldown\printcube\rollup
    \rollleft\flipleft

    \rollleft\printcube
    \rollleft\printcube
    \rollleft\printcube\rollright
    \flipdown \rollleft \printcube
    \rollright\printcube
    \rollright\printcube
    \rollright\printcube
    \rollright\printcube \rollleft
    \flipdown \printcube
    \flipright\printcube
    \rollright\printcube
    \rollright\printcube
    \rollup\printcube\rolldown
    \rolldown\printcube\rollup
    \rollleft\rollleft\flipleft

    \rollleft\printcube
    \rollleft\printcube
    \rollleft\printcube\rollright

    \flipdown\printcube
    \flipleft\printcube
    \rollleft\printcube
    \rollup\printcube\rolldown
    \rolldown\printcube\rollup
    \rollright\flipright

    \rollright\printcube
    \rollright\printcube
    \rollright\printcube

    \begin{scope}[shift={(-2.5,.5)}]
    \draw[thick] (0,1)--(1,1)--(1,0)--(7,0)--(7,-2)--(8,-2)--(8,-5)--(7,-5)--(7,-4)--(6,-4)--(6,-5)--(1,-5)--(1,-6)--(0,-6)--(0,-3)--(1,-3)--(1,-2)--(0,-2)--cycle;
    \draw[thick] (2,-1)--(1,-1)--(1,-2)--(2,-2);
    \draw[thick] (2,-3)--(1,-3)--(1,-4)--(2,-4);
    \draw[thick] (6,0)--(6,-1)--(5,-1);
    \draw[thick] (5,-2)--(6,-2)--(6,-3);
    \draw[thick] (5,-3)--(7,-3)--(7,-2);
    \draw[thick] (5,-4)--(6,-4);

    \foreach \i in {0,-1,-4,-5}
        \draw[dashed] (0,\i)--(1,\i);
    \draw[dashed] (1,0)--(1,-1);
    \draw[dashed] (1,-4)--(1,-5);
    \foreach \i in {2,3,4,5}
        \draw[dashed] (\i,0)--(\i,-5);
    \foreach \i in {-1,-2,-3,-4}
        \draw[dashed] (2,\i)--(5,\i);
    \draw[dashed] (6,-1)--(6,-2);
    \draw[dashed] (6,-3)--(6,-4);
    \draw[dashed] (6,-1)--(7,-1);
    \draw[dashed] (6,-2)--(7,-2);
    \draw[dashed] (7,-3)--(7,-4);
    \draw[dashed] (7,-3)--(8,-3);
    \draw[dashed] (7,-4)--(8,-4);
    \end{scope}
\end{tikzpicture}
    \caption{A surjective consistent mapping of a simply connected polyomino that cannot be realised by a folded state}
    \label{fig:notfoldablefan}
\end{figure}
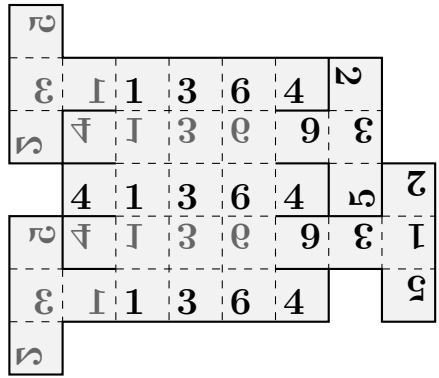
\end{exa}

Self-intersections always occur at pairs of edges of our polyomino mapped to the same edge of the cube by a pseudo-folding. They are caused by layer conflicts of the two faces incident to each of the edges. Each of the edges is folded by either $90\degree$ or $180\degree$. Thus there are three types of intersections, shown in Figure~\ref{fig:intersections}. Firstly, if both edges are folded by $90\degree$, the relative order of their faces must be preserved. Secondly, a $90\degree$ folded edge cannot lie between the two faces of a $180\degree$ fold. Finally, the layers of two $180\degree$ folded edges should be well nested. While the consistent mapping of Example~\ref{exa:non-foldable-unit-square-hole} forces a  self-intersection at two edges folded by $90\degree$, any pseudo-folding obtained from the consistent mapping in Example~\ref{exa:self-intersection-fan} forces an intersection of a $90\degree$ folded edge with a $180\degree$ folded edge.

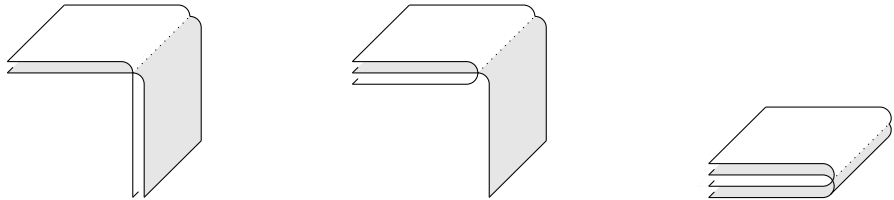
\begin{figure}[ht]
    \centering
\begin{subfigure}{.3\textwidth}
    \centering
    \begin{tikzpicture}[scale=1.5]
    \fill[gray!20] (1.6,1.6) to[out=0,in=90] (1.7,1.5) -- (1.7,0.5) -- (1.2,0) -- (1.2,1) to[out=90,in=0]  (1.1,1.1) ;
    \fill[gray!20] (0,1.1) -- (0.1,1.2) -- (1,1.2) to[out=0,in=90] (1.1,1.1);
    \draw (0,1.2) -- (0.5,1.7);
    \draw (1.2,0) -- (1.7,0.5);
    \draw (0,1.1) -- (0.05,1.15);
    \draw (1.1,0) -- (1.15,0.05);
    \draw[dotted] (1.1,1.1) -- (1.6,1.6);
    \draw (0,1.1) -- (1.1,1.1) to[out=0,in=90] (1.2,1) -- (1.2,0);
    \draw (0,1.2) -- (1,1.2) to[out=0,in=90] (1.1,1.1) -- (1.1,0);
    \begin{scope}[xshift=0.5cm,yshift=0.5cm]
    \draw (1.1,1.1) to[out=0,in=90] (1.2,1) -- (1.2,0);
    \draw (0,1.2) -- (1,1.2) to[out=0,in=90] (1.1,1.1);
    \end{scope}
    \end{tikzpicture}
\end{subfigure}
\begin{subfigure}{.3\textwidth}
    \centering
    \begin{tikzpicture}[scale=1.5]
    \fill[gray!20] (1.6,1.6) to[out=0,in=90] (1.7,1.5) -- (1.7,0.5) -- (1.2,0) -- (1.2,1) to[out=90,in=0]  (1.1,1.1) ;
    \fill[gray!20] (0,1.1) -- (0.1,1.2) -- (1,1.2) to[out=0,in=90] (1.1,1.1);
    \draw (0,1.2) -- (0.5,1.7); 
    \draw (1.2,0) -- (1.7,0.5);
    \draw (0,1.1) -- (0.05,1.15);
    \draw (0,1) -- (0.05,1.05);
    \draw[dotted] (1.1,1.1) -- (1.6,1.6);
    \draw (0,1.1) -- (1.1,1.1) to[out=0,in=90] (1.2,1) -- (1.2,0);
    \draw (0,1.2) -- (1,1.2) to[out=0,in=90] (1.1,1.1) to[out=-90,in=0] (1,1) -- (0,1);
    \begin{scope}[xshift=0.5cm,yshift=0.5cm]
    \draw (1.1,1.1) to[out=0,in=90] (1.2,1) -- (1.2,0);
    \draw (0,1.2) -- (1,1.2) to[out=0,in=90] (1.1,1.1);
    \end{scope}
    \end{tikzpicture}
\end{subfigure}
\begin{subfigure}{.3\textwidth}
    \centering
    \begin{tikzpicture}[scale=1.5]
    \fill[gray!20] (0.1,1.3) -- (1,1.3) to[out=0,in=90] (1.1,1.2) to[out=-90,in=0] (1,1.1) -- (-0.1,1.1);
    \fill[gray!20] (0.2,1.2) -- (1,1.2) to[out=0,in=90] (1.1,1.1) to[out=-90,in=0] (1,1.0) -- (0,1.0);
    \begin{scope}
    \clip (0.1,1.3) -- (1,1.3) to[out=0,in=90] (1.1,1.2) to[out=-90,in=0] (1,1.1) -- (-0.1,1.1)--cycle;
    \fill[white] (-0.1,1.2) -- (1,1.2) to[out=0,in=90] (1.1,1.1) to[out=-90,in=0] (1,1.0) -- (-0.1,1.0);
    \end{scope}
    \fill[gray!20] (1.08,1.15) to[out=-45,in=90] (1.1,1.1) to[out=-90,in=45] (1.07,1.03) -- (1.57,1.53) to[out=45,in=-90] (1.6,1.6) to[out=90,in=-45] (1.58,1.65); 
    \draw (0,1.3) -- (0.5,1.8); 
    \draw (0,1.2) -- (0.05,1.25);
    \draw (0,1.1) -- (0.05,1.15);
    \draw (1.07,1.03) -- (1.57,1.53);
    \draw (0,1) -- (0.05,1.05);
    \draw[dotted] (1.08,1.15) -- (1.58,1.65);
    \draw (0,1.3) -- (1,1.3) to[out=0,in=90] (1.1,1.2) to[out=-90,in=0] (1,1.1) -- (0,1.1);
    \draw (0,1.2) -- (1,1.2) to[out=0,in=90] (1.1,1.1) to[out=-90,in=0] (1,1) -- (0,1);
    \begin{scope}[xshift=0.5cm,yshift=0.5cm]
    \draw (0,1.3) -- (1,1.3) to[out=0,in=90](1.1,1.2) to[out=-90,in=45] (1.08,1.15);
    \draw (1.08,1.15) to[out=-45,in=90] (1.1,1.1) to[out=-90,in=45] (1.07,1.03);
    \end{scope}
    \end{tikzpicture}
\end{subfigure}
    \caption{Different types of self-intersections}
    \label{fig:intersections}
\end{figure}

But even if a given pseudo-folding has realisations, it might still not be valid. This happens if any ($\Leftrightarrow$ each) of its realisation is not equivalent to a trivial folding. Intuitively this means that the realisation $F$ could not have been obtained from the flat polyomino without any self-intersections occurring during the folding process. From Theorem~\ref{thm:foldedstatevalid} we know that this happens exactly if the boundary of $F(P)$ is not the unlink. Obviously this can only happen if the polyomino $P$ contains at least one hole. 

We provide two structurally different examples of two realisable pseudo-foldings which are not valid, starting with an example of a pseudo-folding containing a twist.

\begin{exa}\label{exa:twisted}
    Consider the pseudo-folding $(P,\varphi,\ell)$ given in Figure~\ref{fig:notfoldablewithouttwist}, where layer numbers $\ell$ are drawn as roman numerals in each square. 
    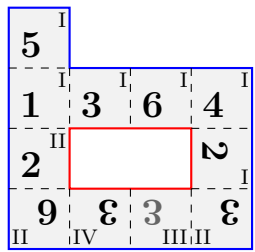
\begin{figure}[ht]
        \centering
    \begin{tikzpicture}[xscale=0.8,yscale=0.8]
        \initcube\printcube
        \rollup\printcube
        \rolldown
        \rolldown\printcube
        \rolldown\printcube
        \rollright\printcube
        \flipright\printcube
        \flipright\printcube
        \rollup\printcube
        \rollup\printcube
        \rollleft\printcube
        \rollleft\printcube
        \begin{scope}[shift={(-.5,-.5)}]
        \draw[blue,thick] (0,1)--(0,2)--(1,2)--(1,1)--(4,1)--(4,-2)--(0,-2)--cycle;
        \draw[red,thick] (1,0)--(3,0)--(3,-1)--(1,-1)--cycle;
        \foreach \i in {1,...,3}{
            \draw[dashed] (\i,1) -- (\i,0);
            \draw[dashed] (\i,-1) -- (\i,-2);
            \draw[dashed] (0,\i-2) -- (1,\i-2);
        }
        \node[anchor=east,inner sep=1pt] at (1,0.8) {\RNum{1}};
        \node[anchor=east,inner sep=1pt] at (1,1.8) {\RNum{1}};
        \node[anchor=east,inner sep=1pt] at (2,0.8) {\RNum{1}};
        \node[anchor=east,inner sep=1pt] at (3,0.8) {\RNum{1}};
        \node[anchor=east,inner sep=1pt] at (4,0.8) {\RNum{1}};
        \node[anchor=east,inner sep=1pt] at (4,-0.8) {\RNum{1}};
        \node[anchor=west,inner sep=1pt] at (3,-1.8) {\RNum{2}};
        \node[anchor=east,inner sep=1pt] at (3,-1.8) {\RNum{3}};
        \node[anchor=west,inner sep=1pt] at (1,-1.8) {\RNum{4}};
        \node[anchor=west,inner sep=1pt] at (0,-1.8) {\RNum{2}};
        \node[anchor=east,inner sep=1pt] at (1,-0.2) {\RNum{2}};
        \draw[dashed] (3,0) -- (4,0) (3,-1) -- (4,-1); 
        \end{scope}
    \end{tikzpicture}
        \caption{A pseudo-folding containing a twist}
        \label{fig:notfoldablewithouttwist}
    \end{figure}

    One can easily check that there are no self-intersection. However, it is still not possible to realise this pseudo-folding. In fact, we could `cheat' by cutting along one edge of the polyomino, twisting one of the two resulting ends by $360\degree$ and glue them back together to obtain the requested result. In other words, the pseudo-folding is not ambient isotopic to the flat polyomino, but to a `doubly twisted Möbius polyomino'.
    A good way to see this is by looking at the linking diagram of the two boundaries, which is drawn on the left of Figure~\ref{fig:linkingdiagramtwist}. It is not hard to see that this diagram is equivalent (via Reidemeister moves) to the two linked rings shown on the right.

    \begin{figure}[ht]
        \centering
    \begin{tikzpicture}[x=1.9cm,y=1.9cm]
        \draw (-2,0)--(0,0)--(0,-1)--(1,-1)--(1,0)--(2,0)--(2,1)--(1,1)--(1,2)--(0,2)--(0,1)--(-2,1)--cycle;
        \draw[dashed] (-1,1)--(-1,0);
        \draw[dashed] (0,1)--(0,0);
        \draw[dashed] (1,1)--(1,0);
        \draw[dashed] (1,0)--(0,0);
        \draw[dashed] (1,1)--(0,1);
        \node at (0.5,0.5) {\Large $\one$};
        \node at (0.5,1.5) {\Large $\five$};
        \node at (1.5,0.5) {\Large $\three$};
        \node at (0.5,-0.5) {\Large $\two$};
        \node at (-0.5,0.5) {\Large $\four$};
        \node at (-1.5,0.5) {\Large $\six$};
        \draw[blue, thick] (0.1,-1.1) -- (0.1,1.9) -- (0.9,1.9) -- (0.9,0.9) -- (2.1,0.9) -- (2.1,2.1) -- (-2.1,2.1) -- (-2.1,0.9) -- (-0.1,0.9) -- (-0.1,-0.1) -- (0.05,-0.1) (0.15,-0.1) -- (0.85,-0.1) (0.95,-0.1) -- (1.1,-0.1) -- (1.1,0.8) -- (1.9, 0.8) -- (1.9, 0.7) -- (1.2, 0.7) -- (1.2, 0.6) -- (2.2, 0.6) -- (2.2, 2.2) -- (-2.2, 2.2) -- (-2.2, 0.8) -- (-1.1,0.8) -- (-1.1,-1.1) -- (0.1,-1.1);
        \draw[red, thick] (0.9,0.3) -- (1.05,0.3) (1.15,0.3) -- (2.2,0.3) -- (2.2,-1.4) -- (-2.2,-1.4) -- (-2.2,0.2) -- (-1.15,0.2) (-1.05,0.2) -- (-0.9,0.2) -- (-0.9,-0.9) -- (0.05,-0.9) (0.15,-0.9) -- (0.85,-0.9) (0.95,-0.9) -- (1.9,-0.9) -- (1.9,0.1) -- (1.2,0.1) -- (1.2,0.2) -- (2.1,0.2) -- (2.1,-1.3) -- (-2.1,-1.3) -- (-2.1,0.1) -- (-1.9,0.1) -- (-1.9,-1.2) -- (0.9,-1.2) -- (0.9,0.3);
        \draw[->] (2.7,0.5) -- (3.2,0.5);
        \begin{scope}[xshift = 8.2cm, yshift = 1cm]
        \begin{knot}[
                flip crossing=1,
                clip width=8,
            ]
            \strand[red,thick] (0,-0.5) circle[radius=1.3cm];
            \strand[blue,thick] (0,0.5) circle[radius=1.3cm];
        \end{knot}
        \end{scope}
    \end{tikzpicture}
        \caption{Linking diagram of the two boundaries.}
        \label{fig:linkingdiagramtwist}
    \end{figure}
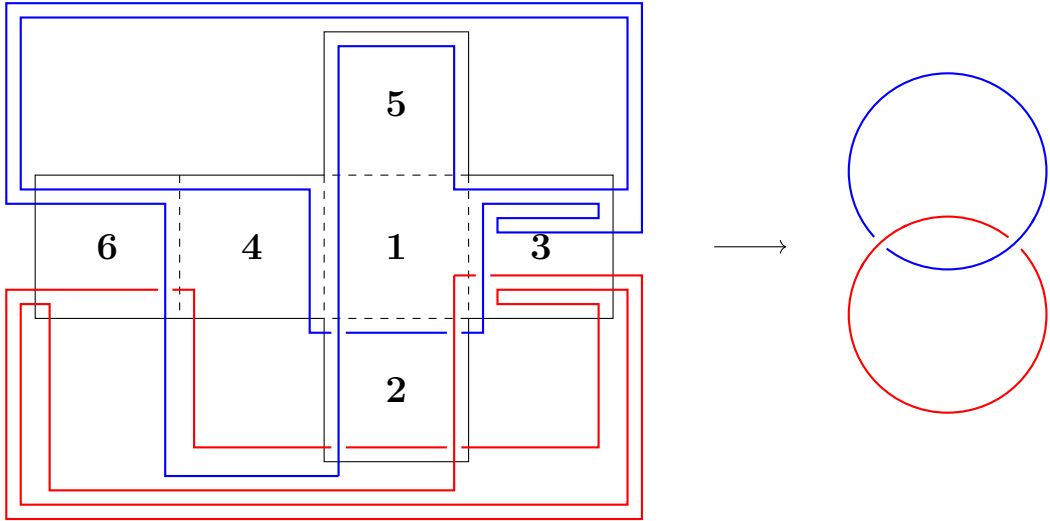
\end{exa}

Twistedness only occurs by multiples of $360\degree$, as realisations of a given consistent mapping must always be orientable surfaces. In particular, non-orientable surfaces such as Möbius strips cannot appear. 

For the polyomino $P$ and the consistent mapping $\varphi$ of Example~\ref{exa:twisted} there are in total 16 different layer maps not leading to any self-intersections. Out of these 16 maps, 4 provide valid pseudo-foldings, 8 include a $360\degree$ twist, and 4 include a $720\degree$ twist. In particular in some cases a realisable pseudo-folding can be made valid by replacing only the layer map. The following example shows that this is not always the case.

\begin{exa}
    Consider the polyomino $P$ and the consistent mapping $\varphi$ depicted in Figure~\ref{fig:realisable-not-valid}. It is not hard to find a layer map $\ell$ such that the pseudo-folding $(P,\varphi,\ell)$ is realisable, one example can be seen in the figure. However, a brute force case-by-case analysis shows that none of the 60 layer maps which do not lead to self-intersections provides a valid pseudo-folding.
    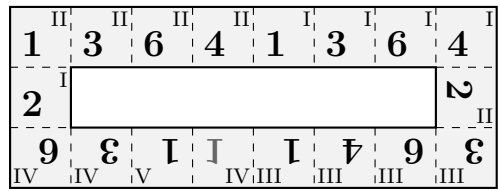
\begin{figure}[ht]
        \centering
    \begin{tikzpicture}[xscale=0.8,yscale=0.8]
        \initcube\printcube
        \rolldown\printcube
        \rolldown\printcube
        \rollright\printcube
        \rollright\printcube
        \flipright\printcube
        \flipright\printcube
        \rollright\printcube
        \rollright\printcube
        \rollright\printcube
        \rollup\printcube
        \rollup\printcube
        \rollleft\printcube
        \rollleft\printcube
        \rollleft\printcube
        \rollleft\printcube
        \rollleft\printcube
        \rollleft\printcube
        \begin{scope}[shift={(-.5,-.5)}]
        \draw[thick] (0,1)--(8,1)--(8,-2)--(0,-2)--cycle;
        \draw[thick] (1,0)--(7,0)--(7,-1)--(1,-1)--cycle;
        \foreach \i in {1,...,7}{
            \draw[dashed] (\i,1) -- (\i,0);
            \draw[dashed] (\i,-1) -- (\i,-2);
        }
        \foreach \i in {0,-1}{
            \draw[dashed] (0,\i) -- (1,\i) (7,\i) -- (8,\i);  
        }
        \node[anchor=east,inner sep=1pt] at (1,0.8) {\RNum{2}};
        \node[anchor=east,inner sep=1pt] at (2,0.8) {\RNum{2}};
        \node[anchor=east,inner sep=1pt] at (3,0.8) {\RNum{2}};
        \node[anchor=east,inner sep=1pt] at (4,0.8) {\RNum{2}};
        \node[anchor=east,inner sep=1pt] at (5,0.8) {\RNum{1}};
        \node[anchor=east,inner sep=1pt] at (6,0.8) {\RNum{1}};
        \node[anchor=east,inner sep=1pt] at (7,0.8) {\RNum{1}};
        \node[anchor=east,inner sep=1pt] at (8,0.8) {\RNum{1}};
        \node[anchor=east,inner sep=1pt] at (8,-0.8) {\RNum{2}};
        \node[anchor=west,inner sep=1pt] at (7,-1.8) {\RNum{3}};
        \node[anchor=west,inner sep=1pt] at (6,-1.8) {\RNum{3}};
        \node[anchor=west,inner sep=1pt] at (5,-1.8) {\RNum{3}};
        \node[anchor=west,inner sep=1pt] at (4,-1.8) {\RNum{3}};
        \node[anchor=east,inner sep=1pt] at (4,-1.8) {\RNum{4}};
        \node[anchor=west,inner sep=1pt] at (2,-1.8) {\RNum{5}};
        \node[anchor=west,inner sep=1pt] at (1,-1.8) {\RNum{4}};
        \node[anchor=west,inner sep=1pt] at (0,-1.8) {\RNum{4}};
        \node[anchor=east,inner sep=1pt] at (1,-0.2) {\RNum{1}};
        \end{scope}
    \end{tikzpicture}
        \caption{A consistent mapping that cannot be realised without a twist.}
        \label{fig:realisable-not-valid}
    \end{figure}
\end{exa}

Finally, we provide an example of a pseudo-folding with a single hole, where the frame of the hole is non-trivially knotted.

\begin{exa}
    Consider the pseudo-folding $(P,\varphi,\ell)$ given in Figure~\ref{fig:notfoldablewithoutknot}. 
    \begin{figure}[ht]
        \centering
    \begin{tikzpicture}[xscale=0.8,yscale=0.8]
        \initcube\printcube
        \flipdown\printcube
        \rollright\printcube
        \flipup\printcube
        \rollright\printcube
        \flipdown\printcube
        \rollright\printcube
        \flipup\printcube
        \rollright\printcube
        \flipup\printcube
        \flipleft\printcube
        \rollleft\printcube
        \flipup\printcube
        \rollleft\printcube
        \flipdown\printcube
        \rollleft\printcube
        \rollleft\printcube
        \flipdown\printcube
        \begin{scope}[shift={(-.5,-.5)}]
        \foreach \i in {0,1,3,4}{\draw[dashed] (\i,1) -- (\i,2);}
        \foreach \i in {0,2,4}{\draw[dashed] (\i,0) -- (\i,1);}
        \draw[dashed] (0,0)--(4,0) (-1,1)--(0,1) (4,1)--(5,1) (1,2)--(3,2);
        \draw[dashed] (1,0)--(1,-1) (3,0)--(3,-1) (2,2)--(2,3);
        \draw[blue,thick] (0,0)--(0,-1)--(4,-1)--(4,0)--(5,0)--(5,2)--(3,2)--(3,3)--(1,3)--(1,2)--(-1,2)--(-1,0)--cycle;
        \draw[blue,thick] (2,0)--(2,-1);
        \draw[thick] (0,1)--(4,1) (1,0)--(1,1) (2,1)--(2,2) (3,0)--(3,1);
        \node[anchor=east,inner sep=1pt] at (1,0.8) {\RNum{2}};
        \node[anchor=east,inner sep=1pt] at (1,-0.8) {\RNum{3}};
        \node[anchor=east,inner sep=1pt] at (2,-0.8) {\RNum{4}};
        \node[anchor=east,inner sep=1pt] at (2,0.8) {\RNum{1}};
        \node[anchor=east,inner sep=1pt] at (3,-0.8) {\RNum{2}};
        \node[anchor=east,inner sep=1pt] at (3,0.8) {\RNum{3}};
        \node[anchor=east,inner sep=1pt] at (4,-0.8) {\RNum{3}};
        \node[anchor=east,inner sep=1pt] at (4,0.8) {\RNum{1}};
        \node[anchor=east,inner sep=1pt] at (5,0.8) {\RNum{6}};
        \node[anchor=east,inner sep=1pt] at (5,1.2) {\RNum{5}};
        \node[anchor=west,inner sep=1pt] at (3,1.2) {\RNum{1}};
        \node[anchor=west,inner sep=1pt] at (2,1.2) {\RNum{2}};
        \node[anchor=west,inner sep=1pt] at (2,2.8) {\RNum{3}};
        \node[anchor=west,inner sep=1pt] at (1,2.8) {\RNum{4}};
        \node[anchor=west,inner sep=1pt] at (1,1.2) {\RNum{1}};
        \node[anchor=west,inner sep=1pt] at (0,1.2) {\RNum{2}};
        \node[anchor=west,inner sep=1pt] at (-1,1.2) {\RNum{7}};
        \node[anchor=west,inner sep=1pt] at (-1,0.8) {\RNum{4}};
        \end{scope}
    \end{tikzpicture}
        \caption{A pseudo-folding containing a knot.}
        \label{fig:notfoldablewithoutknot}
    \end{figure}
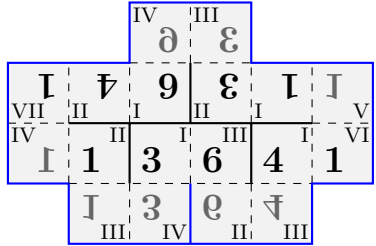
    Again, it is not possible to realise this pseudo-folding. Again, we can `cheat' by cutting one edge of the polyomino, thus ending up with a polyomino $P'$ without any holes. Then one can just realise the pseudo-folding for $P'$ and in the end glue the two sides of the cut edge back together. When trying to unfold the result, it will be clear relatively quickly, that the folded state is knotted. To make this clear, we look at the knot diagram of the outer boundary of $P$ depicted on the left of Figure~\ref{fig:linkingdiagramknot}. As our consistent mapping only maps onto the faces $\one$, $\three$, $\four$ and $\six$, it is reasonable to draw the frontal view of the cube, looking directly at the face $\two$ so that the layers of the four mentioned faces become visible. It is not hard to see that the knot of our boundary is the trefoil knot on the right of Figure~\ref{fig:linkingdiagramknot}, which is the simplest non-trivial knot.

    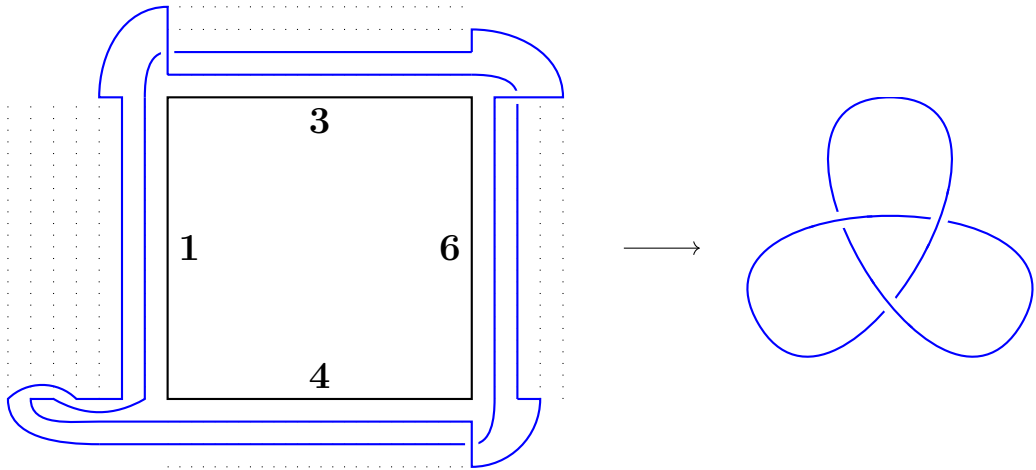
\begin{figure}[ht]
        \centering
    \begin{tikzpicture}[x=1cm,y=1cm]
        \draw[thick] (0,0)-- node[right] {\Large $\one$} (0,4)-- node[below] {\Large $\three$} (4,4)-- node[left] {\Large $\six$} (4,0)-- node[above] {\Large $\four$} cycle;
        \foreach \i in {-2.1,-1.8,-1.5,-1.2,-0.9,-0.6,-0.3,4.3,4.6,4.9,5.2}{
            \draw[loosely dotted] (\i,0) -- (\i,4);}
        \foreach \i in {-0.9,-0.6,-0.3,4.3,4.6,4.9,5.2}{
            \draw[loosely dotted] (0,\i) -- (4,\i);}
        \draw[blue,thick] (-0.3,4) to[out=90,in=180] (0,4.6) -- (4,4.6);
        \draw[white, line width=5pt] (0,4.3) -- (0,4.9);
        \draw[blue,thick] (4,4.3) to[out=0,in=90] (4.6,4) -- (4.6,0);
        \draw[white, line width=5pt] (4.3,4) -- (4.9,4);
        \draw[blue,thick] (4.3,0) to[out=-90,in=0] (4,-0.6) -- (-0.9,-0.6);
        \draw[white, line width=5pt] (4,-0.3) -- (4,-0.9);
        \draw[blue,thick] (0,4.3) -- (4,4.3) (4.6,0) -- (4.9,0) to [out=-90,in=0] (4,-0.9) -- (4,-0.3) -- (-0.9,-0.3) to[out=180,in=-90] (-1.8,0) -- (-1.5,0) to[out=-30,in=210] (-0.3,0) -- (-0.3,4) (4,4.6) -- (4,4.9) to[out=0,in=90] (5.2,4) -- (4.3,4) -- (4.3,0) (-0.9,-0.6) to[out=180,in=-90] (-2.1,0) to[out=45,in=135] (-1.2,0) -- (-0.6,0) -- (-0.6,4) -- (-0.9,4) to[out=90,in=180] (0,5.2) -- (0,4.3);
        \draw[->] (6,2) -- (7,2);
        \begin{scope}[xshift = 9.5cm, yshift = 2cm]
            \clip (2,-2)--(-2,-2)--(-2,2)--(2,2)--cycle;
        \begin{knot}[
            consider self intersections=true,
            flip crossing=2,
            clip width=8,
            ]
        \strand[blue,thick] (0,2) .. controls +(2.2,0) and +(120:-2.2) .. (210:2) .. controls +(120:2.2) and +(60:2.2) .. (-30:2) .. controls +(60:-2.2) and +(-2.2,0) .. (0,2);
        \end{knot}
        \end{scope}
    \end{tikzpicture}
        \caption{Linking diagram of the the outer boundary.}
        \label{fig:linkingdiagramknot}
    \end{figure}
\end{exa}

In the proof of Theorem~\ref{thm:folding-eq-unlink} we have provided a way to translate a given link diagram $L$ consisting of $k>0$ closed curves to a polyomino $P$ with $k$ holes and a folded state $F$ corresponding to a folding blueprint mapping $P$ to a unit square.  The main idea is to first realise every curve in $L$ as the inner boundary of a polyomino with a single hole and then connecting the resulting $k$ polyominoes to obtain a single polyomino $P$, see also Figure~\ref{fig:polyominofromdiagram} for a reminder. In fact, the inner boundaries (the boundaries of the $k$ holes) of the polyomino are knotted and linked in precisely the way given by $L$. From this we immediately obtain that every possible knot can be realised as a folding of some polyomino $P$ with a single hole into (and also onto) $\Ccal$.


We have seen four examples of obstacles preventing pseudo-foldings from being valid.  Obviously every polyomino can be trivially folded onto a single square and thus into $\Ccal$ using only $180\degree$ folds. In fact every example polyomino in this subsection admits a valid pseudo-folding whose consistent mapping onto $\Ccal$ is surjective and thus defines a folding onto $\Ccal$. In this context, the following important question remains open.

\begin{quest}
    Is there any polyomino which admits surjective pseudo-foldings, but none that are valid and define foldings?
\end{quest}

We finish this subsection with some algorithmic considerations. By Theorem \ref{thm:folding-eq-unlink}, deciding whether a given pseudo-folding is realisable is polynomially equivalent to \textsc{Unlink}. Since \textsc{Unlink} is known to be in NP \cite{HLP99}, this shows that the problem \textsc{CubeFolding} of deciding whether a given polyomino can be folded onto $\Ccal$ is in NP. In particular, there is an algorithm to solve this decision problem. 

However, even if there is an efficient algorithm for \textsc{Unlink}, the naive algorithm for \textsc{CubeFolding} has exponential runtime due to the fact that there are exponentially many pseudo-foldings. We conjecture that there is an efficient algorithm.

\begin{conj}
    \textsc{CubeFolding} can be solved in polynomial time. 
\end{conj}

\section{Rectangular polyominoes with one simple hole do not fold}
\label{sec:one_hole}

{\scshape Aichholzer et al.} showed in \cite[Thm. 1]{aich19} that the presence of a non-simple hole in any polyomino guarantees foldability onto $\Ccal$. Additionally, in \cite[Thm. 11]{aich19} and \cite[Thm. 15]{aich19} the authors showed that rectangular polyominoes with only a single square hole or a slit of size 1 and no other holes cannot fold onto $\Ccal$. The goal of this section is to extend this second statement also to the three other classes of simple holes shown in Figure~\ref{fig:simpleholes} to obtain the following main result of this section. 

\begin{thm}\label{thm:one-simple-hole}
    A rectangular polyomino $P$ with a simple hole $h$ and no other holes does not fold onto $\Ccal$.
\end{thm}

We will proof the statement separately first for I-slits of size 2 in Proposition~\ref{pro:I-slit-size-2} and then for U-slits of size 3 and L-slits of size 3 in Proposition~\ref{pro:U-slit}. But first we start with some preparations.

\begin{lem}[{\cite[Lem. 7]{aich19}}] \label{lem:extension}
In any folding of a polyomino $P$ into $\Ccal$, the four faces around a vertex not in the boundary of $P$ cannot cover more than 2 faces of $\Ccal$. In particular either all four faces cover the same face of $\Ccal$ or there are two adjacent faces of $P$ covering one face of $\Ccal$ and the other two faces cover a second face of $\Ccal$.
\end{lem}

\begin{cor} \label{cor:extension}
Let $P$ be a polyomino and $P'$ be a rectangular sub-polyomino of $P$ whose interior does not contain any boundary of $P$. Then any consistent mapping $\varphi$ of $P$ to $\Ccal$ must be constant on each row of $P'$ or constant on each column of $P'$. 
\end{cor}

Let $P$ be a rectangular polyomino. The \emph{support} of a set $H$ of holes of $P$ is the smallest rectangular sub-polyomino $P'$ of $P$ such that each $h \in H$ is a valid hole of $P'$, i.e., $h$ does not lie on the boundary of $P'$.

Let $\varphi$ be a consistent mapping from $P$ to $\Ccal$. We say that $\varphi$ is \emph{trivial} for a simple hole $h$ of $P$ if $\varphi$ is a valid consistent mapping for the polyomino obtained from $P$ by removing $h$. In the case where $h$ is a unit square hole $\varphi$ is extended to the additional square as given by its boundary edges. Clearly it is sufficient to look at the support of $h$ when checking whether $\varphi$ is trivial for $h$.

\begin{pro}\label{pro:I-slit-size-2}
A rectangular polyomino $P$ with an I-slit of size 2 and no other holes does not fold onto $\Ccal$.
\end{pro}
\begin{proof}
We will show that no surjective consistent mapping $\varphi$ from $P$ onto the unit cube $\Ccal$ can exist.

Note that $\varphi$ is trivial for the I-slit $h$, if and only if it maps both `sides' of the slit to the same edges of $\Ccal$. By definition any trivial mapping is also a consistent mapping after removing $h$ and it is clear that such a mapping never covers all faces of $\Ccal$. Denote by $P'$ the support of $h$, which is a $4 \times 2$-polyomino. Using the algorithm from \cite[section 4.3]{aich19}, we deduce that up to isomorphism and up to symmetry of $P'$ there are 6 possible patterns a non-trivial $\varphi$ can take on $P'$; all of them are shown in Figure~\ref{fig:mapsIslit}.

\begin{figure}[ht]
    \centering
\begin{subfigure}{.15\textwidth}
    \centering
    \begin{tikzpicture}[xscale=0.7,yscale=0.7]
        \initcube\printcube
        \rolldown\printcube
        \flipdown\printcube
        \rolldown\printcube
        \rollright\printcube
        \rollup\printcube
        \flipup\printcube
        \rollup\printcube
        \begin{scope}[shift={(-.5,-.5)}]
        \draw[thick] (0,1)--(0,-3)--(2,-3)--(2,1)--cycle;
        \draw[thick] (1,0)--(1,-2);
        \draw[dashed] (0,0)--(2,0);
        \draw[dashed] (0,-1)--(2,-1);
        \draw[dashed] (0,-2)--(2,-2);
        \draw[dashed] (1,1)--(1,0);
        \draw[dashed] (1,-2)--(1,-3);
        \end{scope}
    \end{tikzpicture}
    \caption{\label{case:sl2f}}
\end{subfigure}
\begin{subfigure}{.15\textwidth}
    \centering
    \begin{tikzpicture}[xscale=0.7,yscale=0.7]
        \initcube\printcube
        \flipdown\printcube
        \flipdown\printcube
        \flipdown\printcube
        \rollright\printcube
        \rollup\printcube
        \flipup\printcube
        \rollup\printcube
        \begin{scope}[shift={(-.5,-.5)}]
        \draw[thick] (0,1)--(0,-3)--(2,-3)--(2,1)--cycle;
        \draw[thick] (1,0)--(1,-2);
        \draw[dashed] (0,0)--(2,0);
        \draw[dashed] (0,-1)--(2,-1);
        \draw[dashed] (0,-2)--(2,-2);
        \draw[dashed] (1,1)--(1,0);
        \draw[dashed] (1,-2)--(1,-3);
        \end{scope}
    \end{tikzpicture}
    \caption{\label{case:sl2e}}
\end{subfigure}
\begin{subfigure}{.15\textwidth}
    \centering
    \begin{tikzpicture}[xscale=0.7,yscale=0.7]
        \initcube\printcube
        \flipdown\printcube
        \rolldown\printcube
        \rolldown\printcube
        \flipright\printcube
        \flipup\printcube
        \rollup\printcube
        \rollup\printcube
        \begin{scope}[shift={(-.5,-.5)}]
        \draw[thick] (0,1)--(0,-3)--(2,-3)--(2,1)--cycle;
        \draw[thick] (1,0)--(1,-2);
        \draw[dashed] (0,0)--(2,0);
        \draw[dashed] (0,-1)--(2,-1);
        \draw[dashed] (0,-2)--(2,-2);
        \draw[dashed] (1,1)--(1,0);
        \draw[dashed] (1,-2)--(1,-3);
        \end{scope}
    \end{tikzpicture}
    \caption{\label{case:sl2d}}
\end{subfigure}
\begin{subfigure}{.15\textwidth}
    \centering
    \begin{tikzpicture}[xscale=0.7,yscale=0.7]
        \initcube\printcube
        \flipdown\printcube
        \rolldown\printcube
        \flipdown\printcube
        \flipright\printcube
        \rollup\printcube
        \rollup\printcube
        \rollup\printcube
        \begin{scope}[shift={(-.5,-.5)}]
        \draw[thick] (0,1)--(0,-3)--(2,-3)--(2,1)--cycle;
        \draw[thick] (1,0)--(1,-2);
        \draw[dashed] (0,0)--(2,0);
        \draw[dashed] (0,-1)--(2,-1);
        \draw[dashed] (0,-2)--(2,-2);
        \draw[dashed] (1,1)--(1,0);
        \draw[dashed] (1,-2)--(1,-3);
        \end{scope}
    \end{tikzpicture}
    \caption{\label{case:sl2c}}
\end{subfigure}
\begin{subfigure}{.15\textwidth}
    \centering
    \begin{tikzpicture}[xscale=0.7,yscale=0.7]
        \initcube\printcube
        \flipdown\printcube
        \flipdown\printcube
        \rolldown\printcube
        \flipright\printcube
        \flipup\printcube
        \flipup\printcube
        \rollup\printcube
        \begin{scope}[shift={(-.5,-.5)}]
        \draw[thick] (0,1)--(0,-3)--(2,-3)--(2,1)--cycle;
        \draw[thick] (1,0)--(1,-2);
        \draw[dashed] (0,0)--(2,0);
        \draw[dashed] (0,-1)--(2,-1);
        \draw[dashed] (0,-2)--(2,-2);
        \draw[dashed] (1,1)--(1,0);
        \draw[dashed] (1,-2)--(1,-3);
        \end{scope}
    \end{tikzpicture}
    \caption{\label{case:sl2b}}
\end{subfigure}
\begin{subfigure}{.15\textwidth}
    \centering
    \begin{tikzpicture}[xscale=0.7,yscale=0.7]
        \initcube\printcube
        \flipdown\printcube
        \flipdown\printcube
        \flipdown\printcube
        \flipright\printcube
        \rollup\printcube
        \flipup\printcube
        \rollup\printcube
        \begin{scope}[shift={(-.5,-.5)}]
        \draw[thick] (0,1)--(0,-3)--(2,-3)--(2,1)--cycle;
        \draw[thick] (1,0)--(1,-2);
        \draw[dashed] (0,0)--(2,0);
        \draw[dashed] (0,-1)--(2,-1);
        \draw[dashed] (0,-2)--(2,-2);
        \draw[dashed] (1,1)--(1,0);
        \draw[dashed] (1,-2)--(1,-3);
        \end{scope}
    \end{tikzpicture}
    \caption{\label{case:sl2a}}
\end{subfigure}
\caption{All non-trivial consistent mappings from the support of an I-slit of length $2$ to $\Ccal$.}
\label{fig:mapsIslit}
\end{figure}
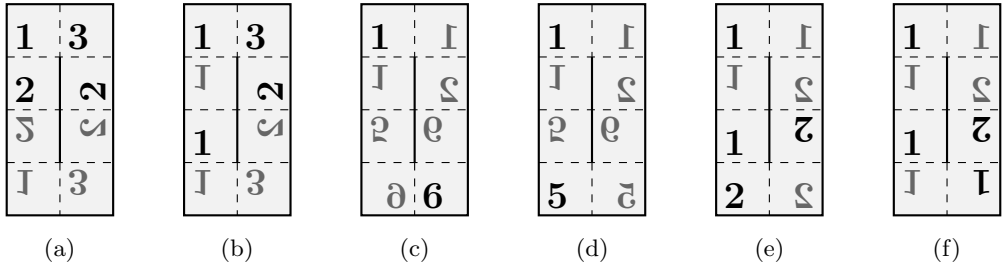

We use Corollary~\ref{cor:extension} to show that none of these cases can extend to a surjective consistent mapping from $P$ to $\Ccal$. First note that it is enough to show that the restriction of $\varphi$ to all unit squares of $P$ sharing a row or a column with $P'$ cannot map surjective onto $\Ccal$. Then by Corollary~\ref{cor:extension} the image of $P$ under $\varphi$ contains the same faces as this restriction and thus not all faces of $\Ccal$.

We make the following observations:
\begin{itemize}
    \item If at least two unit squares in the first row of $P'$ cover different faces of $\Ccal$, any unit square of $P$ lying above $P'$ maps to the same face of $\Ccal$ as the unit square below it and the same holds for other sides of the boundary of $P'$. Therefore $\varphi$ only maps to faces $\one$, $\two$ and $\three$ of $\Ccal$ in case \subref{case:sl2f}.
    \item In case \subref{case:sl2e} only the faces $\one$, $\three$, $\four$ and $\six$ of $\Ccal$ can be covered by unit squares to the left of $P'$, so $\varphi$ does not map to face $\five$.
    \item In cases \subref{case:sl2d}, \subref{case:sl2c} and \subref{case:sl2b} only faces $\one$, $\two$, $\five$ and $\six$ of $\Ccal$ can be covered by unit squares above and below $P'$, so $\varphi$ does not map to face $\four$.
    \item In case \subref{case:sl2a} either the column of $P$ containing the first column of $P'$ or the row of $P$ containing the first row of $P'$ is mapped onto face $\one$ of $\Ccal$. Thus either all unit squares of $P$ above and below $P'$ are mapped to face $\one$ and $\varphi$ does not map to face $\five$ or all unit squares to the left of $P'$ are mapped to face $\one$ and $\varphi$ does not map to face $\four$.
\end{itemize}
Summing up, in none of the cases $\varphi$ can be surjective, so a rectangular polyomino with a single straight slit of size 2 does not fold onto $\Ccal$.
\end{proof}

\begin{pro}\label{pro:U-slit}
    A rectangular polyomino $P$ with an U-slit of size 3 or an L-slit of size 2 and no other holes does not fold onto $\Ccal$.
\end{pro}
\begin{proof}
Let us start with the case where $P$ contains a single U-slit $h$ of size 3 and no other holes. Observe that $\varphi$ trivial, if both sides of the slit map to the same edge of $\Ccal$. Additionally, if $\varphi$ maps the unit square inside of the slits to the same face as the only unit square connected to it, $P$ reduces to the case where $h$ is a unit square hole and thus does not fold onto $\Ccal$. Excluding these two types of consistent mappings, up to isomorphism and symmetry, only seven possible consistent mappings $\varphi$ from the support $P'$ of $h$ to $\Ccal$ remain. They can be seen in Figure~\ref{fig:mapsUslit}.

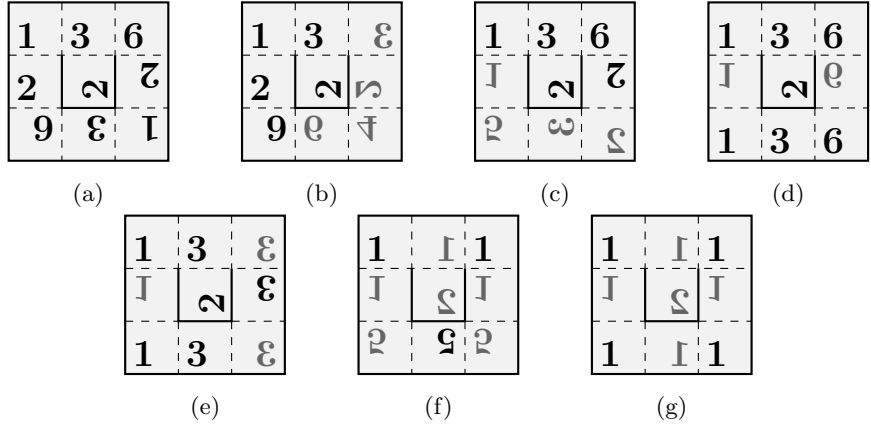
\begin{figure}[ht]
    \centering
\begin{subfigure}{.2\textwidth}
    \centering
    \begin{tikzpicture}[xscale=0.7,yscale=0.7]
        \initcube\printcube
        \rolldown\printcube
        \rolldown\printcube
        \rollright\printcube
        \rollright\printcube
        \rollup\printcube
        \rollup\printcube
        \rollleft\printcube
        \rolldown\printcube
        \begin{scope}[shift={(-.5,-.5)}]
        \draw[thick] (0,1)--(3,1)--(3,-2)--(0,-2)--cycle;
        \draw[thick] (1,0)--(1,-1)--(2,-1)--(2,0);
        \draw[dashed] (0,0)--(3,0);
        \draw[dashed] (0,-1)--(1,-1);
        \draw[dashed] (2,-1)--(3,-1);
        \draw[dashed] (1,1)--(1,0);
        \draw[dashed] (2,1)--(2,0);
        \draw[dashed] (1,-1)--(1,-2);
        \draw[dashed] (2,-1)--(2,-2);
        \end{scope}
    \end{tikzpicture}
    \caption{\label{case:uslg}}
\end{subfigure}
\begin{subfigure}{.2\textwidth}
    \centering
    \begin{tikzpicture}[xscale=0.7,yscale=0.7]
        \initcube\printcube
        \rolldown\printcube
        \rolldown\printcube
        \flipright\printcube
        \rollright\printcube
        \rollup\printcube
        \rollup\printcube
        \flipleft\printcube
        \rolldown\printcube
        \begin{scope}[shift={(-.5,-.5)}]
        \draw[thick] (0,1)--(3,1)--(3,-2)--(0,-2)--cycle;
        \draw[thick] (1,0)--(1,-1)--(2,-1)--(2,0);
        \draw[dashed] (0,0)--(3,0);
        \draw[dashed] (0,-1)--(1,-1);
        \draw[dashed] (2,-1)--(3,-1);
        \draw[dashed] (1,1)--(1,0);
        \draw[dashed] (2,1)--(2,0);
        \draw[dashed] (1,-1)--(1,-2);
        \draw[dashed] (2,-1)--(2,-2);
        \end{scope}
    \end{tikzpicture}
    \caption{\label{case:uslf}}
\end{subfigure}
\begin{subfigure}{.2\textwidth}
    \centering
    \begin{tikzpicture}[xscale=0.7,yscale=0.7]
        \initcube\printcube
        \flipdown\printcube
        \rolldown\printcube
        \rollright\printcube
        \rollright\printcube
        \flipup\printcube
        \rollup\printcube
        \rollleft\printcube
        \rolldown\printcube
        \begin{scope}[shift={(-.5,-.5)}]
        \draw[thick] (0,1)--(3,1)--(3,-2)--(0,-2)--cycle;
        \draw[thick] (1,0)--(1,-1)--(2,-1)--(2,0);
        \draw[dashed] (0,0)--(3,0);
        \draw[dashed] (0,-1)--(1,-1);
        \draw[dashed] (2,-1)--(3,-1);
        \draw[dashed] (1,1)--(1,0);
        \draw[dashed] (2,1)--(2,0);
        \draw[dashed] (1,-1)--(1,-2);
        \draw[dashed] (2,-1)--(2,-2);
        \end{scope}
    \end{tikzpicture}
    \caption{\label{case:usle}}
\end{subfigure}
\begin{subfigure}{.2\textwidth}
    \centering
    \begin{tikzpicture}[xscale=0.7,yscale=0.7]
        \initcube\printcube
        \flipdown\printcube
        \flipdown\printcube
        \rollright\printcube
        \rollright\printcube
        \flipup\printcube
        \flipup\printcube
        \rollleft\printcube
        \rolldown\printcube
        \begin{scope}[shift={(-.5,-.5)}]
        \draw[thick] (0,1)--(3,1)--(3,-2)--(0,-2)--cycle;
        \draw[thick] (1,0)--(1,-1)--(2,-1)--(2,0);
        \draw[dashed] (0,0)--(3,0);
        \draw[dashed] (0,-1)--(1,-1);
        \draw[dashed] (2,-1)--(3,-1);
        \draw[dashed] (1,1)--(1,0);
        \draw[dashed] (2,1)--(2,0);
        \draw[dashed] (1,-1)--(1,-2);
        \draw[dashed] (2,-1)--(2,-2);
        \end{scope}
    \end{tikzpicture}
    \caption{\label{case:usld}}
\end{subfigure}
\begin{subfigure}{.2\textwidth}
    \centering
    \begin{tikzpicture}[xscale=0.7,yscale=0.7]
        \initcube\printcube
        \flipdown\printcube
        \flipdown\printcube
        \rollright\printcube
        \flipright\printcube
        \flipup\printcube
        \flipup\printcube
        \flipleft\printcube
        \rolldown\printcube
        \begin{scope}[shift={(-.5,-.5)}]
        \draw[thick] (0,1)--(3,1)--(3,-2)--(0,-2)--cycle;
        \draw[thick] (1,0)--(1,-1)--(2,-1)--(2,0);
        \draw[dashed] (0,0)--(3,0);
        \draw[dashed] (0,-1)--(1,-1);
        \draw[dashed] (2,-1)--(3,-1);
        \draw[dashed] (1,1)--(1,0);
        \draw[dashed] (2,1)--(2,0);
        \draw[dashed] (1,-1)--(1,-2);
        \draw[dashed] (2,-1)--(2,-2);
        \end{scope}
    \end{tikzpicture}
    \caption{\label{case:uslc}}
\end{subfigure}
\begin{subfigure}{.2\textwidth}
    \centering
    \begin{tikzpicture}[xscale=0.7,yscale=0.7]
        \initcube\printcube
        \flipdown\printcube
        \rolldown\printcube
        \flipright\printcube
        \flipright\printcube
        \rollup\printcube
        \flipup\printcube
        \flipleft\printcube
        \rolldown\printcube
        \begin{scope}[shift={(-.5,-.5)}]
        \draw[thick] (0,1)--(3,1)--(3,-2)--(0,-2)--cycle;
        \draw[thick] (1,0)--(1,-1)--(2,-1)--(2,0);
        \draw[dashed] (0,0)--(3,0);
        \draw[dashed] (0,-1)--(1,-1);
        \draw[dashed] (2,-1)--(3,-1);
        \draw[dashed] (1,1)--(1,0);
        \draw[dashed] (2,1)--(2,0);
        \draw[dashed] (1,-1)--(1,-2);
        \draw[dashed] (2,-1)--(2,-2);
        \end{scope}
    \end{tikzpicture}
    \caption{\label{case:uslb}}
\end{subfigure}
\begin{subfigure}{.2\textwidth}
    \centering
    \begin{tikzpicture}[xscale=0.7,yscale=0.7]
        \initcube\printcube
        \flipdown\printcube
        \flipdown\printcube
        \flipright\printcube
        \flipright\printcube
        \flipup\printcube
        \flipup\printcube
        \flipleft\printcube
        \rolldown\printcube
        \begin{scope}[shift={(-.5,-.5)}]
        \draw[thick] (0,1)--(3,1)--(3,-2)--(0,-2)--cycle;
        \draw[thick] (1,0)--(1,-1)--(2,-1)--(2,0);
        \draw[dashed] (0,0)--(3,0);
        \draw[dashed] (0,-1)--(1,-1);
        \draw[dashed] (2,-1)--(3,-1);
        \draw[dashed] (1,1)--(1,0);
        \draw[dashed] (2,1)--(2,0);
        \draw[dashed] (1,-1)--(1,-2);
        \draw[dashed] (2,-1)--(2,-2);
        \end{scope}
    \end{tikzpicture}
    \caption{\label{case:usla}}
\end{subfigure}
\caption{All non-trivial consistent mappings of the support of a U-slit to $\Ccal$ such that the face `inside' the slit does not map to the same face as the face `above' the slit.}
\label{fig:mapsUslit}
\end{figure}

Similar as in the proof of the previous proposition, we make the following observations.
\begin{itemize}
    \item If at least two unit squares in the first row of $P'$ cover different faces of $\Ccal$, any unit square of $P$ lying above $P'$ maps to the same face of $\Ccal$ as the unit square below it and the same holds for other sides of the boundary of $P'$. This handles cases \subref{case:uslg}, \subref{case:uslf}, and \subref{case:usle}, where exactly one face of $\Ccal$ does not occur.
    \item In cases \subref{case:usld} and \subref{case:uslc} only the faces $\one$, $\three$, $\four$ and $\six$ of $\Ccal$ can be covered by unit squares to the left and right of $P'$, so $\varphi$ does not map any unit square to face $\five$.
    \item In case \subref{case:uslb} only faces $\one$, $\two$, $\five$ and $\six$ of $\Ccal$ can be covered by unit squares above and below $P'$, so $\varphi$ does not map to face $\four$.
    \item In case \subref{case:usla} either the column of $P$ containing the first column of $P'$ or the row of $P$ containing the first row of $P'$ is mapped to face $\one$ of $\Ccal$. Thus either all unit squares of $P$ above and below $P'$ are mapped to face $\one$ and $\varphi$ does not map to face $\five$ or all unit squares to the left and right of $P'$ are mapped to face $\one$ and $\varphi$ does not map to face $\four$.
\end{itemize}
As before $\varphi$ cannot be surjective in any of these cases, thus $P$ does not fold onto $\Ccal$.

Finally note that if any $P$ containing an L-slit of size 2 and no other holes folds onto $\Ccal$, the polyomino with a U-slit of size 3 obtained by cutting an additional edge has to fold in the same way. We conclude that no such polyomino $P$ can fold onto $\Ccal$. 
\end{proof}

\section{Folding polyominoes with two unit square holes}
\label{sec:two_holes}

We have seen in Theorem~\ref{thm:one-simple-hole} that a single hole in a rectangular polyomino $P$ is not sufficient for $P$ to become foldable. But what about two or more holes? In this section we fully characterise which rectangular polyominoes containing only unit square holes fold onto the cube. In particular we will show that certain pairs of unit square holes are necessary and sufficient for foldability.

\begin{dfn}
Let $P$ be a rectangular polyomino. We say that two unit square holes of $P$ \emph{cooperate} if the holes lie in the same or in adjacent rows (resp. columns) and the number of columns (resp. rows) between them is odd. 
\end{dfn}

Observe first that by \cite[Thm. 5]{aich19} two cooperating unit square holes guarantee foldability. The main idea is to use $180^\circ$ folds to reduce the corresponding polyomino to one of the two situations depicted in Figure~\ref{fig:coop-unitsquareholes}. It is left to the reader to check that the consistent mappings shown in the figure provide a valid folding.

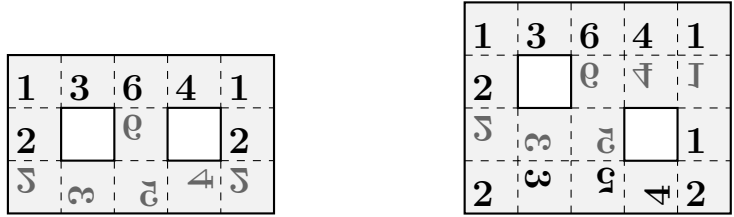
\begin{figure}[ht]
\centering
\begin{subfigure}{.4\textwidth}
\centering
\begin{tikzpicture}[xscale=0.7,yscale=0.7]
\initcube\printcube
\rolldown\printcube
\flipdown\printcube
\rollright\printcube
\rollright\printcube
\rollright\printcube
\rollright\printcube
\flipup\printcube
\rollup\printcube
\rollleft\printcube
\rollleft\printcube
\rollleft\printcube
\rollright\flipdown\printcube

\begin{scope}[shift={(-.5,-.5)}]
\draw[thick] (0,1)--(0,-2)--(5,-2)--(5,1)--cycle;
\foreach \i in {0,-1}{
    \draw[dashed] (0,\i)--(1,\i);
    \draw[dashed] (2,\i)--(3,\i);
    \draw[dashed] (4,\i)--(5,\i);}
\foreach \i in {1,...,4}{
    \draw[dashed] (\i,1)--(\i,0);
    \draw[dashed] (\i,-1)--(\i,-2);}
\draw[thick] (1,0)--(1,-1)--(2,-1)--(2,0)--cycle;
\draw[thick] (3,0)--(3,-1)--(4,-1)--(4,0)--cycle;
\end{scope}
\end{tikzpicture}
\end{subfigure}
\begin{subfigure}{.4\textwidth}
\centering
\begin{tikzpicture}[xscale=0.7,yscale=0.7]
\initcube\printcube
\rolldown\printcube
\flipdown\printcube
\flipdown\printcube
\rollright\printcube
\flipup\printcube
\rollright\printcube
\flipdown\printcube
\rollright\printcube
\rollright\printcube
\rollup\printcube
\flipup\printcube
\flipup\printcube
\rollleft\printcube
\flipdown\printcube
\rollleft\printcube
\flipup\printcube
\rollleft\printcube

\begin{scope}[shift={(-.5,-.5)}]
\draw[thick] (0,1)--(0,-3)--(5,-3)--(5,1)--cycle;
\foreach \i in {0,-1,-2}{
    \draw[dashed] (0,\i)--(1,\i);
    \draw[dashed] (2,\i)--(3,\i);
    \draw[dashed] (4,\i)--(5,\i);}
\foreach \i in {1,...,4}{
    \draw[dashed] (\i,1)--(\i,0);
    \draw[dashed] (\i,-3)--(\i,-2);}
\draw[dashed] (1,-1) -- (1,-2) -- (2,-2) -- (2,-1);
\draw[dashed] (3,-1) -- (3,0) -- (4,0) -- (4,-1);
\draw[thick] (1,0)--(1,-1)--(2,-1)--(2,0)--cycle;
\draw[thick] (3,-2)--(3,-1)--(4,-1)--(4,-2)--cycle;
\end{scope}
\end{tikzpicture}
\end{subfigure}
\caption{Consistent mappings for cooperating unit square holes}
\label{fig:coop-unitsquareholes}
\end{figure}

Let $h$ be a unit square hole of a polyomino $P$ and let $P'$ be the support of $h$. For the upcoming Theorem it is essential to know all consistent mappings from $P'$ to $\Ccal$ which are non-trivial for $h$. Following the discussion in \cite[Sec. 4.1]{aich19}, there are two essentially different mappings: Either the four boundary edges of $h$ are mapped to the same edge of $\Ccal$ or they are mapped to two edges of $\Ccal$ forming an L-shape.

It is shown in \cite[Lem. 9]{aich19} that any mapping which maps the boundary of $h$ to a single edge of $\Ccal$ does not yield a valid folding. The strategy in the proof is to show that in this case up to isomorphism one always ends up with the consistent mapping $\varphi$ shown in Figure~\ref{fig:squareholemap-1}. But the discussion in Example~\ref{exa:non-foldable-unit-square-hole} shows that $\varphi$ does not provide a valid folding without self-intersections. In the upcoming theorem we want to exclude consistent mappings which can never yield valid foldings. Thus we call a consistent mapping $\varphi$ on $P$ \emph{\good}, if for each unit square hole $h$ of $P$, the restriction of $\varphi$ to the support of $h$ is not isomorphic to the consistent mapping in Figure~\ref{fig:squareholemap-1}.

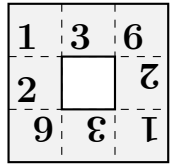
\begin{figure}[ht]
\centering
\begin{tikzpicture}[xscale=0.7,yscale=0.7]
\initcube\printcube
\rolldown\printcube
\rolldown\printcube
\rollright\printcube
\rollright\printcube
\rollup\printcube
\rollup\printcube
\rollleft\printcube
\begin{scope}[shift={(-.5,-.5)}]
\draw[thick] (0,1)--(0,-2)--(3,-2)--(3,1)--cycle;
\foreach \i in {0,-1}{
    \draw[dashed] (0,\i)--(3,\i);}
\foreach \i in {1,...,2}{
    \draw[dashed] (\i,1)--(\i,-2);}
\draw[thick] (1,0)--(1,-1)--(2,-1)--(2,0)--cycle;
\end{scope}
\end{tikzpicture}
\caption{Non-trivial consistent mapping of a unit square hole onto $\Ccal$ of non-foldable type}
\label{fig:squareholemap-1}
\end{figure}

In \cite[Thm. 10]{aich19} the authors show that up to isomorphisms of the cube and up to rotation and reflection of $P'$, there exists only a single {\good} consistent mapping $\varphi$ from $P'$ onto $\Ccal$ which is non-trivial for $h$. This mapping arises from the second case and can be found in Figure~\ref{fig:squareholemap-2}. There is a unique face $f$ of $\Ccal$ whose boundary contains the two edges covered by the 4 boundary edges of $h$. Observe that $f$ is covered by two squares lying on different sides of the hole $h$ and that the face opposite to $f$ in $\Ccal$ cannot be covered by a square in $P'$. We say that the \emph{type of $h$ under $\varphi$} is $f$.

\begin{figure}[ht]
\centering
\begin{tikzpicture}[xscale=0.7,yscale=0.7]
\initcube\printcube
\rolldown\printcube
\flipdown\printcube
\rollright\printcube
\rollright\printcube
\rollup\printcube
\flipup\printcube
\rollleft\printcube
\begin{scope}[shift={(-.5,-.5)}]
\draw[thick] (0,1)--(0,-2)--(3,-2)--(3,1)--cycle;
\foreach \i in {0,-1}{
    \draw[dashed] (0,\i)--(3,\i);}
\foreach \i in {1,...,2}{
    \draw[dashed] (\i,1)--(\i,-2);}
\draw[thick] (1,0)--(1,-1)--(2,-1)--(2,0)--cycle;
\end{scope}
\end{tikzpicture}
\caption{Non-trivial consistent mapping $\varphi$ of the support of a unit square hole $h$ to $\Ccal$. The type of $h$ under $\varphi$ is $\three$.}
\label{fig:squareholemap-2}
\end{figure}
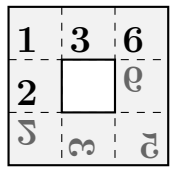

\begin{thm}
\label{thm:two_holes}
    Let $P$ be a rectangular polyomino with unit square holes and no other holes. Then $P$ folds onto $\Ccal$ if and only if two of the holes cooperate.
\end{thm}

\begin{proof}
We have just seen that $P$ folds if it contains two cooperating unit square holes.

To prove the other direction, we will show that for any polyomino $P$ with unit square holes and no other holes such that no two of them cooperate, there cannot exist a {\good} surjective consistent mapping. We use induction on the number of holes $k$.

The base cases $k=0$ and $k=1$ follow from Theorem~\ref{thm:one-simple-hole}. Consider now $k > 1$ and assume that the induction hypothesis holds for all polyominoes with less than $k$ unit square holes. Let $P$ be a polyomino with $k$ unit square holes and no other holes and assume for a contradiction that there is a {\good} surjective consistent mapping $\varphi$. Then $\varphi$ is non-trivial for every hole $h$ of $P$, otherwise the natural extension of $\varphi$ to the polyomino obtained from $P$ by removing $h$ would be a {\good} surjective consistent mapping from a polyomino with $k-1$ holes to $\Ccal$, contradicting our assumption.

Let us define a graph $G$ as follows. The vertex set is the set of holes $H$ of $P$ and two holes $h_1, h_2 \in H$ are connected by an edge if they lie in the same or adjacent rows or columns of $P$ and the support of $\{h_1,h_2\}$ in $P$ does not contain any other holes in $H$. 
\begin{clm} \label{clm:graph-connected}
    The graph $G$ is connected.
\end{clm}
To prove this, suppose for a contradiction that $G$ is not connected. Let $h_1$ and $h_2$ be two holes lying in different components of $G$ such that the support $P'$ of $\{h_1,h_2\}$ is of minimal size, that is, the sum of its rows and columns is minimal. Then $P'$ cannot contain any other hole. Indeed, such a hole $h$ cannot lie in the support of $h_1$ or $h_2$ and thus both the support of $\{h,h_1\}$ and $\{h,h_2\}$ have smaller size than $P'$. This contradicts the minimality of $P'$.

By construction of $G$ the holes $h_1$ and $h_2$ cannot lie in the same or adjacent rows or columns of $P$, otherwise they would either cooperate or be connected by an edge in $G$. 
Assume without loss of generality that $h_1$ lies in row 2 and column 2 and $h_2$ lies in row $i \geq 4$ and column $j \geq 4$ of $P'$. The map $\varphi$ is non-trivial for both $h_1$ and $h_2$, thus there are two adjacent unit squares in the first three rows of column 3 mapped to different faces of $\Ccal$. Using Corollary~\ref{cor:extension} we see that $\varphi$ must be constant on each row of the sub-polyomino consisting of rows $1$ to $i-1$ and columns $3$ to $j+1$, see Figure~\ref{fig:twononcooperativeholes}. But this contradicts our assumption that $\varphi$ is non-trivial for $h_2$ and thus proves Claim~\ref{clm:graph-connected}. 

\begin{figure}[ht]
\centering
\begin{tikzpicture}[xscale=0.7,yscale=0.7]
\begin{scope}[shift={(-.5,-.5)}]
\fill[gray!10] (0,1)--(0,-6)--(7,-6)--(7,1);
\end{scope}
\initcube\printcube
\rolldown\printcube
\flipdown\printcube
\rollright\printcube
\rollright\printcube
\rollup\printcube
\flipup\printcube
\rollleft\printcube
\rollright\flipright\flipright\printcube
\flipright\printcube
\flipright\printcube
\flipdown\printcube
\flipleft\printcube
\flipleft\printcube
\rolldown\printcube
\flipright\printcube
\flipright\printcube
\flipdown\flipdown\printcube
\flipleft\printcube
\flipleft\printcube
\flipleft\flipleft\printcube
\begin{scope}[shift={(-.5,-.5)}]
\foreach \i in {0,...,-5}{
    \draw[dashed] (0,\i)--(3,\i) (7,\i)--(4,\i);}
\foreach \i in {1,...,6}{
    \draw[dashed] (\i,1)--(\i,-2) (\i,-6)--(\i,-3);}
\draw[loosely dotted, thick] (3.3,-2.3) -- (3.7,-2.7);
\draw[thick,fill=white] (1,0)--(1,-1)--(2,-1)--(2,0)--cycle;
\draw[thick,fill=white] (5,-4)--(5,-5)--(6,-5)--(6,-4)--cycle;
\draw[thick] (0,1)--(0,-6)--(7,-6)--(7,1)--cycle;
\draw[thick, red] (2,1) -- (2,-4) -- (7,-4) -- (7,1) -- cycle;
\end{scope}
\end{tikzpicture}
\caption{The holes $h_1$ and $h_2$ cannot both be non-trivial.}
\label{fig:twononcooperativeholes}
\end{figure}
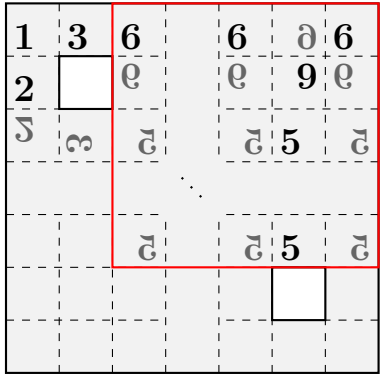

\begin{clm}\label{clm:type-coincides}
Whenever two holes $h_1,h_2 \in H$ are neighbours in $G$, their type coincides.
\end{clm}
To prove this claim, assume first that the two holes lie in the same row or column of $P$. This case is already treated in the proof of \cite[Thm. 12]{aich19}, we will give a brief summary. Assume without loss of generality that $h_1$ lies in column $2$ and $h_2$ in column $j$ of the support $P'$. The map $\varphi$ is non-trivial for both $h_1$ and $h_2$, thus there must be two adjacent unit squares in column 3 mapping to different faces of $\Ccal$. Thus Corollary~\ref{cor:extension} yields that $\varphi$ is constant on each row of the sub-polyomino consisting of columns $3$ to $j-1$. As $h_1$ and $h_2$ do not cooperate by assumption, the number of columns between them is even. We conclude that the pattern around $h_2$ is the reflection of the pattern around $h_1$, so in particular they are of the same type. Up to isomorphism and symmetry there are are only two possible situations, see Figure~\ref{fig:ush-same}.

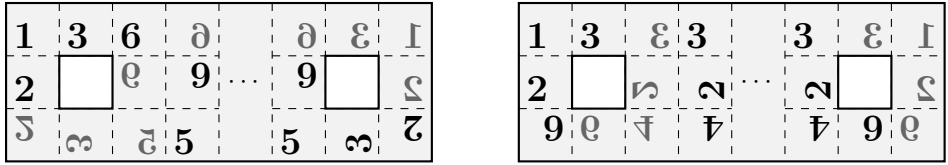
\begin{figure}[ht]
\centering
\begin{subfigure}{.45\textwidth}
\centering
\begin{tikzpicture}[xscale=0.7,yscale=0.7]
\initcube\printcube
\rolldown\printcube
\flipdown\printcube
\rollright\printcube
\rollright\printcube
\rollup\printcube
\flipup\printcube
\rollleft\printcube
\rollright
\flipright\printcube
\flipdown\printcube
\rolldown\printcube
\rollup
\flipup
\flipright
\flipright\printcube
\flipdown\printcube
\rolldown\printcube
\rollright\printcube
\rollright\printcube
\flipup\printcube
\rollup\printcube
\rollleft\printcube
\begin{scope}[shift={(-.5,-.5)}]
\fill[gray!10] (4,1)--(4,-2)--(5,-2)--(5,1);
\draw[thick] (0,1)--(0,-2)--(8,-2)--(8,1)--cycle;
\foreach \i in {0,-1}{
    \draw[dashed] (0,\i)--(4,\i);
    \draw[dashed] (5,\i)--(8,\i);}
\foreach \i in {1,...,7}{
    \draw[dashed] (\i,1)--(\i,-2);}
\draw[thick] (1,0)--(1,-1)--(2,-1)--(2,0)--cycle;
\draw[thick] (7,-1)--(7,0)--(6,0)--(6,-1)--cycle;
\node () at (4.5,-0.5) {$\cdots$};
\end{scope}
\end{tikzpicture}
\end{subfigure}
\begin{subfigure}{.45\textwidth}
\centering
\begin{tikzpicture}[xscale=0.7,yscale=0.7]
\initcube\printcube
\rolldown\printcube
\rolldown\printcube
\flipright\printcube
\rollright\printcube
\rollup\printcube
\rollup\printcube
\flipleft\printcube
\flipright
\flipright\printcube
\rolldown\printcube
\rolldown\printcube
\rollup
\rollup
\flipright
\flipright\printcube
\rolldown\printcube
\rolldown\printcube
\rollright\printcube
\flipright\printcube
\rollup\printcube
\rollup\printcube
\rollleft\printcube
\begin{scope}[shift={(-.5,-.5)}]
\fill[gray!10] (4,1)--(4,-2)--(5,-2)--(5,1);
\draw[thick] (0,1)--(0,-2)--(8,-2)--(8,1)--cycle;
\foreach \i in {0,-1}{
    \draw[dashed] (0,\i)--(4,\i);
    \draw[dashed] (5,\i)--(8,\i);}
\foreach \i in {1,...,7}{
    \draw[dashed] (\i,1)--(\i,-2);}
\draw[thick] (1,0)--(1,-1)--(2,-1)--(2,0)--cycle;
\draw[thick] (7,-1)--(7,0)--(6,0)--(6,-1)--cycle;
\node () at (4.5,-0.5) {$\cdots$};
\end{scope}
\end{tikzpicture}
\end{subfigure}
\caption{Possible patterns for holes in the same row}
\label{fig:ush-same}
\end{figure}

Next we consider holes lying in adjacent rows or columns of $P$. Assume without loss of generality that $h_1$ lies in row $2$ and column $2$ and $h_2$ in row $3$ in column $j$ of the support $P'$. By assumption $h_1$ and $h_2$ do not cooperate, so $j$ must be odd and at least $5$. Additionally $\varphi$ is non-trivial for $h_1$ and $h_2$. An application of Corollary~\ref{cor:extension} to the sub-polyomino consisting of rows $1$ and $2$ and columns $3$ to $j+1$ yields that $\varphi$ is constant on its columns and thus in particular the squares in row $1$ and $2$ of column $3$ are mapped to the same face of $\Ccal$. Analogously we obtain that the squares in rows 3 and 4 of column $j-1$ are mapped to the same face of $\Ccal$. Finally, applying Corollary~\ref{cor:extension} to the rectangular sub-polyomino consisting of columns $3$ to $j-1$ of $P'$ provides that the situation is up to isomorphism as shown in Figure~\ref{fig:ush-adjacent}. In particular the type of $h_1$ and $h_2$ under $\varphi$ coincides. This concludes the proof of Claim~\ref{clm:type-coincides}.

\begin{figure}[ht]
\centering
\begin{tikzpicture}[xscale=0.7,yscale=0.7]
\initcube\printcube
\rolldown\printcube
\flipdown\printcube
\flipdown\printcube
\rollright\printcube
\flipup\printcube
\rollright\printcube
\flipdown\printcube
\flipup
\rollup\printcube
\flipup\printcube
\rollleft\printcube
\rollright
\flipright\printcube
\flipdown\printcube
\rolldown\printcube
\flipdown\printcube
\flipup
\rollup
\flipup
\flipright
\flipright\printcube
\flipdown\printcube
\rolldown\printcube
\flipdown\printcube
\rollright\printcube
\rollright\printcube
\rollup\printcube
\flipup\printcube
\flipup\printcube
\rollleft\printcube
\flipdown\printcube
\begin{scope}[shift={(-.5,-.5)}]
\fill[gray!10] (4,1)--(4,-3)--(5,-3)--(5,1);
\draw[thick] (0,1)--(0,-3)--(8,-3)--(8,1)--cycle;
\foreach \i in {0,-1,-2}{
    \draw[dashed] (0,\i)--(4,\i);
    \draw[dashed] (5,\i)--(8,\i);}
\foreach \i in {1,...,7}{
    \draw[dashed] (\i,1)--(\i,-3);}
\draw[thick] (1,0)--(1,-1)--(2,-1)--(2,0)--cycle;
\draw[thick] (7,-2)--(7,-1)--(6,-1)--(6,-2)--cycle;
\node () at (4.5,-1) {$\cdots$};
\end{scope}
\end{tikzpicture}
\caption{Possible pattern for holes in adjacent rows}
\label{fig:ush-adjacent}
\end{figure}

Connectedness of $G$ implies that the type of all unit square holes of the polyomino $P$ coincides. Let $f$ be this type. Then for any hole $h$ no square of the support of $h$ is mapped to the face $f'$ opposite to $f$ in $\Ccal$. Furthermore, the restriction of $\varphi$ to a rectangular sub-polyomino $P''$ of $P$ containing only a single hole $h$ cannot be surjective by the induction basis. Using that the support of $h$ already covers all faces but $f'$, we conclude that no square of $P''$ is mapped to $f'$ by $\varphi$. Finally, it is not hard to see that any square of the polyomino $P$ is contained in some rectangular sub-polyomino containing precisely one hole of $P$, so that $\varphi$ cannot be surjective on $P$.
\end{proof}

\section{Folding tree-shaped polyominoes}
\label{sec:tree-shaped}

A polyomino $P$ is called \emph{tree-shaped} if its dual $D(P)$ is a tree. Note that this both forbids holes and also any rectangular sub-polyominoes of size at least $2 \times 2$. A tree-shaped polyomino is called \emph{path-shaped}, if additionally every square is connected to at most two other squares. The goal of this section is to fully characterise all tree-shaped polyominoes which can be folded onto the cube. The \emph{bounding box} of a polyomino $P$ is the smallest rectangle containing $P$. Note that by definition $P$ touches each side of the rectangle. The size of the bounding box of $P$ is called the \emph{bounding size} of $P$.

Our main goal is to extend results of {\scshape Aichholzer et al.} in \cite{aich18}, who provided a characterisation of foldable tree-shaped polyominoes with a bounding box of height at most $3$. However, both their statement and proof appear to be incomplete; for instance, their claimed results would imply that the polyomino sketched in Figure \ref{fig:forb-attach-b} below would not be foldable. Hence we decided to additionally come up with our own formulation and proof of their statement. To do this, we use similar methods.

We start by providing some useful language. Let $P$ and $Q$ be two tree-shaped polyominoes. We say that $Q$ is a \emph{minor} of $P$ if $P$ contains a sub-polyomino $P'$ which can be reduced to $Q$ by $180\degree$ folds. Note that because $P$ is tree-shaped, containing a foldable sub-polyomino implies that $P$ is foldable. In particular, whenever $P$ contains a foldable minor, $P$ must be foldable.

A \emph{leaf square} of $P$ is a square of $P$ which is only connected at one side to the rest of $P$. In other words a leaf-square is a square whose dual is a leaf in $D(P)$. A \emph{leaf-fold} is the $180\degree$ fold of a leaf square of $P$.  Note that a leaf-fold transforms $P$ into another tree-shaped polyomino $P'$ which has one square less.

In the upcoming proofs, we often show foldability of a polyomino by proving that it contains a cube net as a minor. Up to symmetry, there are eleven different cube nets; all of them are depicted in Figure~\ref{fig:cubenets}.  Due to the number of squares in each row, the first six cube nets are said to be of type 1-4-1, the following three are of type 2-3-1, then there is one of type 2-2-2 and the last one is of type 3-3.

\begin{figure}[ht]
    \tikzstyle{every node}=[circle, draw, fill,
                        inner sep=0pt, minimum width=4pt]
\centering
\begin{subfigure}[b]{.16\textwidth}
    \centering
    \begin{tikzpicture}[xscale=0.5,yscale=0.5]
        \fill[gray!10] (0,-1)--(1,-1)--(1,0)--(4,0)--(4,1)--(1,1)--(1,2)--(0,2);
        \foreach \i in {1,2,3}{\draw[dashed] (\i,0)--(\i,1);}
        \draw[dashed] (0,1)--(1,1) (0,0)--(1,0);
        \draw[thick] (0.5,-0.5) node {} -- (0.5,0.5) node {} -- (0.5,1.5) node {} (0.5,0.5) --  (1.5,0.5) node {} --  (2.5,0.5) node {} --  (3.5,0.5) node {};
        \draw[thick] (0,-1)--(1,-1)--(1,0)--(4,0)--(4,1)--(1,1)--(1,2)--(0,2)--cycle;
    \end{tikzpicture}
    \caption{}
    \label{cubenet-a}
\end{subfigure}
\begin{subfigure}[b]{.16\textwidth}
    \centering
    \begin{tikzpicture}[xscale=0.5,yscale=0.5]
        \fill[gray!10] (0,0)--(1,0)--(1,-1)--(2,-1)--(2,0)--(4,0)--(4,1)--(1,1)--(1,2)--(0,2);
        \foreach \i in {1,2,3}{\draw[dashed] (\i,0)--(\i,1);}
        \draw[dashed] (0,1)--(1,1) (1,0)--(2,0);
        \draw[thick] (0.5,1.5) node {} -- (0.5,0.5) node {} -- (1.5,0.5) node {}  -- (1.5,-0.5) node {} (1.5,0.5) --  (2.5,0.5) node {} --  (3.5,0.5) node {};
        \draw[thick] (0,0)--(1,0)--(1,-1)--(2,-1)--(2,0)--(4,0)--(4,1)--(1,1)--(1,2)--(0,2)--cycle;
    \end{tikzpicture}
    \caption{}
    \label{cubenet-b}
\end{subfigure}
\begin{subfigure}[b]{.16\textwidth}
    \centering
    \begin{tikzpicture}[xscale=0.5,yscale=0.5]
        \fill[gray!10] (0,0)--(2,0)--(2,-1)--(3,-1)--(3,0)--(4,0)--(4,1)--(1,1)--(1,2)--(0,2);
        \foreach \i in {1,2,3}{\draw[dashed] (\i,0)--(\i,1);}
        \draw[dashed] (0,1)--(1,1) (2,0)--(3,0);
        \draw[thick] (0.5,1.5) node {} -- (0.5,0.5) node {} -- (1.5,0.5) node {} -- (2.5,0.5) node {}  -- (2.5,-0.5) node {} (2.5,0.5) --  (3.5,0.5) node {};
        \draw[thick] (0,0)--(2,0)--(2,-1)--(3,-1)--(3,0)--(4,0)--(4,1)--(1,1)--(1,2)--(0,2)--cycle;
    \end{tikzpicture}
    \caption{}
    \label{cubenet-c}
\end{subfigure}
\begin{subfigure}[b]{.16\textwidth}
    \centering
    \begin{tikzpicture}[xscale=0.5,yscale=0.5]
        \fill[gray!10] (0,0)--(3,0)--(3,-1)--(4,-1)--(4,1)--(1,1)--(1,2)--(0,2);
        \foreach \i in {1,2,3}{\draw[dashed] (\i,0)--(\i,1);}
        \draw[dashed] (0,1)--(1,1) (3,0)--(4,0);
        \draw[thick] (0.5,1.5) node {} -- (0.5,0.5) node {} -- (1.5,0.5) node {} -- (2.5,0.5) node {} --  (3.5,0.5) node {} -- (3.5,-0.5) node {};
        \draw[thick] (0,0)--(3,0)--(3,-1)--(4,-1)--(4,1)--(1,1)--(1,2)--(0,2)--cycle;
    \end{tikzpicture}
    \caption{}
    \label{cubenet-d}
\end{subfigure}
\begin{subfigure}[b]{.16\textwidth}
    \centering
    \begin{tikzpicture}[xscale=0.5,yscale=0.5]
        \fill[gray!10] (0,0)--(1,0)--(1,-1)--(2,-1)--(2,0)--(4,0)--(4,1)--(2,1)--(2,2)--(1,2)--(1,1)--(0,1);
        \foreach \i in {1,2,3}{\draw[dashed] (\i,0)--(\i,1);}
        \draw[dashed] (1,1)--(2,1) (1,0)--(2,0);
        \draw[thick] (0.5,0.5) node {} -- (1.5,0.5) node {} -- (2.5,0.5) node {} --  (3.5,0.5) node {}  (1.5,0.5) --  (1.5,-0.5) node {} (1.5,0.5) --  (1.5,1.5) node {};
        \draw[thick] (0,0)--(1,0)--(1,-1)--(2,-1)--(2,0)--(4,0)--(4,1)--(2,1)--(2,2)--(1,2)--(1,1)--(0,1)--cycle;
    \end{tikzpicture}
    \caption{}
    \label{cubenet-e}
\end{subfigure}
\begin{subfigure}[b]{.16\textwidth}
    \centering
    \begin{tikzpicture}[xscale=0.5,yscale=0.5]
        \fill[gray!10] (0,0)--(2,0)--(2,-1)--(3,-1)--(3,0)--(4,0)--(4,1)--(2,1)--(2,2)--(1,2)--(1,1)--(0,1);
        \foreach \i in {1,2,3}{\draw[dashed] (\i,0)--(\i,1);}
        \draw[dashed] (1,1)--(2,1) (2,0)--(3,0);
        \draw[thick] (0.5,0.5) node {} -- (1.5,0.5) node {} -- (2.5,0.5) node {} --  (3.5,0.5) node {}  (2.5,0.5) --  (2.5,-0.5) node {} (1.5,0.5) --  (1.5,1.5) node {};
        \draw[thick] (0,0)--(2,0)--(2,-1)--(3,-1)--(3,0)--(4,0)--(4,1)--(2,1)--(2,2)--(1,2)--(1,1)--(0,1)--cycle;
    \end{tikzpicture}
    \caption{}
    \label{cubenet-f}
\end{subfigure}
\par\bigskip
\begin{subfigure}[b]{.16\textwidth}
    \centering
    \begin{tikzpicture}[xscale=0.5,yscale=0.5]
        \fill[gray!10] (0,-1)--(1,-1)--(1,0)--(3,0)--(3,1)--(4,1)--(4,2)--(2,2)--(2,1)--(0,1);
        \draw[dashed] (0,0)--(1,0) (1,1)--(1,0) (2,0)--(2,1) (2,1)--(3,1) (3,1)--(3,2);
        \draw[thick] (0.5,0.5) node {} -- (1.5,0.5) node {} --  (2.5,0.5) node {} --  (2.5,1.5) node {} -- (3.5,1.5) node {} (0.5,0.5) -- (0.5,-0.5) node {};
        \draw[thick] (0,-1)--(1,-1)--(1,0)--(3,0)--(3,1)--(4,1)--(4,2)--(2,2)--(2,1)--(0,1)--cycle;
    \end{tikzpicture}
    \caption{}
    \label{cubenet-g}
\end{subfigure}
\begin{subfigure}[b]{.16\textwidth}
    \centering
    \begin{tikzpicture}[xscale=0.5,yscale=0.5]
        \fill[gray!10] (0,0)--(1,0)--(1,-1)--(2,-1)--(2,0)--(3,0)--(3,1)--(4,1)--(4,2)--(2,2)--(2,1)--(0,1);
        \draw[dashed] (1,0)--(2,0) (1,1)--(1,0) (2,0)--(2,1) (2,1)--(3,1) (3,1)--(3,2);
        \draw[thick] (0.5,0.5) node {} -- (1.5,0.5) node {} --  (2.5,0.5) node {} --  (2.5,1.5) node {} -- (3.5,1.5) node {} (1.5,0.5) -- (1.5,-0.5) node {};
        \draw[thick] (0,0)--(1,0)--(1,-1)--(2,-1)--(2,0)--(3,0)--(3,1)--(4,1)--(4,2)--(2,2)--(2,1)--(0,1)--cycle;
    \end{tikzpicture}
    \caption{}
    \label{cubenet-h}
\end{subfigure}
\begin{subfigure}[b]{.16\textwidth}
    \centering
    \begin{tikzpicture}[xscale=0.5,yscale=0.5]
        \fill[gray!10] (0,0)--(2,0)--(2,-1)--(3,-1)--(3,1)--(4,1)--(4,2)--(2,2)--(2,1)--(0,1);
        \draw[dashed] (2,0)--(3,0) (1,1)--(1,0) (2,0)--(2,1) (2,1)--(3,1) (3,1)--(3,2);
        \draw[thick] (0.5,0.5) node {} -- (1.5,0.5) node {} --  (2.5,0.5) node {} --  (2.5,1.5) node {} -- (3.5,1.5) node {} (2.5,0.5) -- (2.5,-0.5) node {};
        \draw[thick] (0,0)--(2,0)--(2,-1)--(3,-1)--(3,1)--(4,1)--(4,2)--(2,2)--(2,1)--(0,1)--cycle;
    \end{tikzpicture}
    \caption{}
    \label{cubenet-i}
\end{subfigure}
\quad\quad
\begin{subfigure}[b]{.16\textwidth}
    \centering
    \begin{tikzpicture}[xscale=0.5,yscale=0.5]
        \fill[gray!10] (0,0)--(2,0)--(2,1)--(3,1)--(3,2)--(4,2)--(4,3)--(2,3)--(2,2)--(1,2)--(1,1)--(0,1);
        \foreach \i in {1,2}{\draw[dashed] (\i,\i-1)--(\i,\i) (\i,\i) -- (\i+1,\i);}
        \draw[dashed] (3,2) -- (3,3);
        \draw[thick] (0.5,0.5) node {} -- (1.5,0.5) node {} --  (1.5,1.5) node {} --  (2.5,1.5) node {} -- (2.5,2.5) node {} -- (3.5,2.5) node {};
        \draw[thick] (0,0)--(2,0)--(2,1)--(3,1)--(3,2)--(4,2)--(4,3)--(2,3)--(2,2)--(1,2)--(1,1)--(0,1)--cycle;
    \end{tikzpicture}
    \caption{}
    \label{cubenet-j}
\end{subfigure}
\quad\quad
\begin{subfigure}[b]{.195\textwidth}
    \centering
    \begin{tikzpicture}[xscale=0.5,yscale=0.5]
        \fill[gray!10] (0,0)--(3,0)--(3,1)--(5,1)--(5,2)--(2,2)--(2,1)--(0,1);
        \foreach \i in {1,2}{\draw[dashed] (\i,0)--(\i,1) (\i+2,1) -- (\i+2,2);}
        \draw[dashed] (2,1) -- (3,1);
        \draw[thick] (0.5,0.5) node {} -- (1.5,0.5) node {} --  (2.5,0.5) node {} --  (2.5,1.5) node {} -- (3.5,1.5) node {} -- (4.5,1.5) node {};
        \draw[thick] (0,0)--(3,0)--(3,1)--(5,1)--(5,2)--(2,2)--(2,1)--(0,1)--cycle;
    \end{tikzpicture}
    \caption{}
    \label{cubenet-k}
\end{subfigure}
\caption{The eleven different cube nets.}
\label{fig:cubenets}
\end{figure}
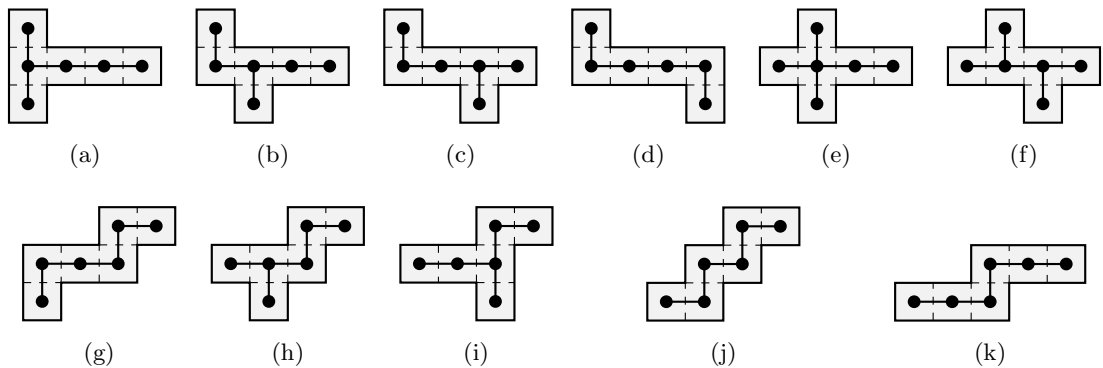

A first trivial observation is that any folding of a polyomino of bounding
size $1 \times n$ can cover at most 4 faces of the cube $\Ccal$.

\subsection[Folding tree-shaped polyomino of bounding size 2 x n]{Folding tree-shaped polyomino of bounding size \protect{$2 \times n$}}

Let us start by defining an infinite class of tree-shaped polyominoes of bounding size $2 \times n$. Class A includes all polyominoes consisting of a $1 \times n$ strip, possibly with several attachments containing only squares in columns of distance at most one to the column where it is attached to the strip. Note that each of these valid attachments is contained in the attachment shown in Figure~\ref{fig:attachment-a}. In this figure, the square at the center of the bottom row is part of the strip, which possibly continues to the left and right. In particular the squares of the attachment directly left and right of this square can only be present if the attachment is attached to the square at the respective end of the strip.

\begin{figure}[ht]
    \tikzstyle{every node}=[circle, draw, fill,
                        inner sep=0pt, minimum width=4pt]
\centering
\begin{tikzpicture}[xscale=0.7,yscale=0.7]
    \fill[gray!10] (0,1)--(1,1)--(1,2)--(0,2);
    \fill[gray!10] (0,0)--(0,2)--(-1,2)--(-1,0) (1,0)--(1,2)--(2,2)--(2,0);
    \foreach \i in {0,1}{\draw[dashed] (\i,2)--(\i,1);}
    \draw[dashed] (-1,1)--(0,1) (0,1)--(1,1) (1,1) -- (2,1);
    \draw[thick, densely dotted] (-2,0.5) -- (-1.05,0.5) (0.05,0.5) -- (0.95,0.5) (2.05,0.5) -- (3,0.5);
    \draw[thick] (0.5,0.5) -- (0.5,1.5) node {} -- (-0.5,1.5) node {} -- (-0.5,0.5) node {} (0.5,1.5) --  (1.5,1.5) node {} --  (1.5,0.5) node {};
    \draw[thick] (0,1)--(0,0)--(-1,0)--(-1,2)--(2,2)--(2,0)--(1,0)--(1,1);
    \foreach \i in {-1.5,0.5,2.5}{\node at (\i,0.5) [diamond,minimum width=6pt,minimum height=6pt,gray!80] {};}	
\end{tikzpicture}
\caption{Valid Class A attachments are part of this attachment.}
\label{fig:attachment-a}
\end{figure}
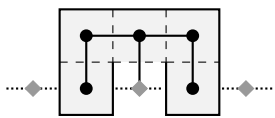

Observe that tree-shaped polyominoes of bounding size $2 \times n$ for $n \leq 3$ cannot fold onto $\Ccal$. Indeed, any foldable polyomino has to contain at least 6 squares and none of the 11 different nets of the cube of Figure~\ref{fig:cubenets} fits into a $2 \times 3$ rectangle. Thus the following theorem fully classifies the set of foldable tree-shaped polyominoes of bounding size $2 \times n$.

\begin{thm}
\label{thm:2_times_n}
Let $P$ be a tree-shaped polyomino of bounding size $2 \times n$ for some $n \geq 4$. Then $P$ folds onto $\Ccal$ if and only if it is not contained in Class A.
\end{thm}
\begin{proof}
We start by showing that a polyomino $P$ in Class A does not fold onto $\Ccal$. Clearly in any folding of $P$ into $\Ccal$, the squares contained in the strip can cover at most 4 faces of~$\Ccal$. In particular all vertices at the same (say top) side of the strip are mapped onto the corners of the same face $f$ of $\Ccal$. Observe that each square of the possible attachments shares one of these vertices with the strip. This implies that none of the squares of the attachments can cover the face of $\Ccal$ opposite to $f$.

To prove the converse direction, let $P$ be a polyomino of bounding size $2 \times n$ for some $n \geq 4$ which does not belong to Class A. In other words, it is not possible to find a strip in $P$ such that $P$ consists of this strip and only valid attachments of Class A as depicted in Figure~\ref{fig:attachment-a}. We conclude that there must be two adjacent squares lying in the same column $i$ of the bounding box such that after removing the edge connecting these squares, both components are not valid attachments. Valid attachments are characterised by the fact that they do not visit any column at distance at least 2 from their root square. In particular we can assume that each of our two components visits at least one of the columns $i-2$ and $i+2$. In the case where one component visits column $i-2$ and the other component visits column $i+2$, it is not hard to see that $P$ contains a cube net of type 3-3 as a minor and thus must be foldable. Otherwise assume that both components visit the same column at distance $2$ from $i$, say $i+2$. Then $P$ contains two strips of length 3 on top of each other, which are connected only in column $i$. But by our assumption, the bounding box has length 4, so at least one of the strips is connected to one more square and thus must have length at least 4. But in this case  we can fold this strip onto four faces of the cube and cover the other two faces with squares of the second strip in column $i$ and $i+2$.
\end{proof}

\subsection[Folding tree-shaped polyomino of bounding size 3 x n]{Folding tree-shaped polyomino of bounding size \protect{$3 \times n$}}

Again, we start by defining an infinite class of polyominoes of bounding size $3 \times n$.
Class B includes all polyominoes consisting of a $1 \times n$ strip with exactly one attachment of height two satisfying the following condition. Assume that the strip lies in the lower row of the bounding box (the other case is a horizontal reflection). Then all squares of the attachment lie in the two rows above the strip and in columns of distance at most one to the column where it is attached to the strip, see Figure~\ref{fig:attachment-b}. Moreover, if the square in the center of the top row is connected to its left and right neighbour, it must also be connected to its bottom neighbour. In particular any attachment containing the attachment in Figure~\ref{fig:non-attachment-b} is forbidden because the upper central square is not connected to its bottom neighbour.

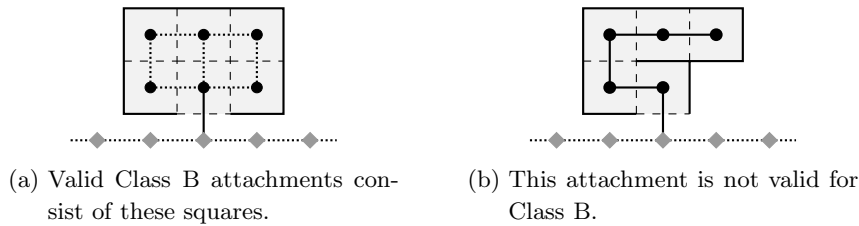
\begin{figure}[ht]
    \tikzstyle{every node}=[circle, draw, fill,
                        inner sep=0pt, minimum width=4pt]
\centering
\begin{subfigure}[b]{.35\textwidth}
    \centering
    \begin{tikzpicture}[xscale=0.7,yscale=0.7]
        \fill[gray!10] (-1,1)--(2,1)--(2,3)--(-1,3);
        \foreach \i in {0,1}{
        	\draw[dashed] (\i,3)--(\i,1);
        	\draw[thick, densely dotted] (-0.5,\i+1.5)--(0.5,\i+1.5) (0.5,\i+1.5)--(1.5,\i+1.5);
        	}
        \draw[dashed] (-1,2)--(0,2) (0,2)--(1,2) (1,2) -- (2,2) (0,1)--(1,1);
        \draw[thick, densely dotted] (-2,0.5) -- (3,0.5);
        \foreach \i in {0,1,2}{
        	\draw[thick, densely dotted] (\i-0.5,1.5)--(\i-0.5,2.5);
        	\foreach \j in {0,1}{\node at (\i-0.5,\j+1.5) {};}}
        \draw[thick] (0.5,0.5) -- (0.5,1.5);
        \draw[thick] (0,1)--(-1,1)--(-1,3)--(2,3)--(2,1)--(1,1);
        \foreach \i in {-1.5,-0.5,0.5,1.5,2.5}{\node at (\i,0.5) [diamond,minimum width=6pt,minimum height=6pt,gray!80] {};}
    \end{tikzpicture}
    \caption{Valid Class B attachments consist of these squares.}
    \label{fig:attachment-b}
\end{subfigure}
\quad\quad
\begin{subfigure}[b]{.35\textwidth}
    \centering
    \begin{tikzpicture}[xscale=0.7,yscale=0.7]
        \fill[gray!10] (-1,1)--(1,1)--(1,2)--(2,2)--(2,3)--(-1,3);
        \draw[dashed] (-1,2)--(0,2) (0,1)--(1,1) (0,1) -- (0,2) -- (0,3) (1,2) -- (1,3);
        \draw[thick, densely dotted] (-2,0.5) -- (3,0.5);
        \draw[thick] (0.5,0.5) -- (0.5,1.5) node {} --(-0.5,1.5) node {} -- (-0.5,2.5) node {} -- (0.5,2.5) node {} -- (1.5,2.5) node {};
        \draw[thick] (0,1)--(-1,1)--(-1,3)--(2,3)--(2,2)--(1,2)--(1,1) (1,2)--(0,2);
        \foreach \i in {-1.5,-0.5,0.5,1.5,2.5}{\node at (\i,0.5) [diamond,minimum width=6pt,minimum height=6pt,gray!80] {};}
    \end{tikzpicture}
    \caption{This attachment is not valid for Class B.}
    \label{fig:non-attachment-b}
\end{subfigure}
\caption{Valid and non-valid attachments for Class B.}
\label{fig:valid-attachments}
\end{figure}

The following theorem characterises all foldable tree-shaped polyominoes of bounding size $3 \times n$ for $n \leq 5$, thus leaving only finitely many polyominoes to be treated separately.

\begin{thm}
\label{thm:3_times_n}
Let $P$ be a tree-shaped polyomino of bounding size $3 \times n$ for some $n \geq 5$. Then $P$ folds onto $\Ccal$ if and only if it is not contained in Class B.
\end{thm}
\begin{proof}
To prove the first direction, take any folding of a Class B polyomino $P$ into $\Ccal$. In this folding the squares of the strip of $P$ can cover at most 4 faces of $\Ccal$, leaving two opposite faces $f$ and $f'$ of $\Ccal$ unoccupied. Let $s$ be the square of the attachment adjacent to the strip and observe that by the definition of valid attachments each square of the attachment contains one of the two vertices of $s$ not contained in the strip. In the folding into $\Ccal$, these two vertices are either both contained in $f$ or in $f'$. We conclude that the respective second face cannot be covered by the attachment.

For the converse direction, assume that $P$ does not fold onto $\Ccal$. By definition of the bounding box, $P$ contains a path-shaped sub-polyomino starting in the first column and ending in the last column of the bounding box. Assume that this path-shaped polyomino contains two adjacent squares in column $i$ for some $3 \leq i \leq n-2$. Then clearly it contains a cube net of type 3-3 and thus folds onto $\Ccal$. We conclude that $P$ contains a strip of length at least $n-2$ which visits column 2 and column $n-1$. Let $S$ be a strip in $P$ of maximal length satisfying these conditions. Observe that each attachment attached to $S$ in column $i$ can visit only columns $i-1$ to $i+1$. Indeed, if this is not the case we can apply $180\degree$ folds to horizontal edges to reduce $P$ to a polyomino of bounding size $2 \times n$ having an attachment not contained in these three columns, which is foldable by the previous theorem.


Assume first that $S$ lies in the central row of the bounding box. If $S$ has length at least 4, we fold all vertical edges of $P$ not contained in the strip by $180\degree$ and the resulting structure contains a cube net of type 1-4-1. If there is an attachment visiting both the top and bottom row, we can apply $180\degree$ folds to each horizontal edge of all other attachments and end up in the previous case, where $S$ has length 4. We are left with the case where $S$ has length 3 and we have at least one attachment at each side of the strip. In this case column 1 must be visited by an attachment, which by the earlier observation must be attached in column 2. But then $P$ contains a cube net of type 2-3-1 as a minor.

We conclude that $S$ lies in an extremal row, say the bottom row and has one attachment of height 2 reaching the top row. It cannot have a second attachment, otherwise we could apply a $180\degree$ fold to the attachment of height 2 and obtain a polyomino with the same bounding size as $P$ and $S$ lying in the central row, which is foldable by the previous case. Recall that if the attachment is attached in column $i$, it is contained in columns $i-1$ to $i+1$. Furthermore, it is not hard to see that no square of the attachment lies in the bottom row of the bounding box. Indeed, such a square has to be connected to its top neighbour. Figure~\ref{fig:forb-attach-a} shows one possible polyomino satisfying these conditions. A $180\degree$ fold of the horizontal edge connecting the strip and the attachment (drawn red in the figure) yields a polyomino with the same bounding size as $P$ having $S$ in the central row, which is foldable by the earlier discussion.

Finally, assume for a contradiction that the square $s'$ in the upper row of column $i$ is connected to both its left and right neighbour, but not to its bottom neighbour. Then the shortest path-shaped sub-polyomino of $P$ connecting $s'$ with the strip must either visit column $i-1$ or column $i+1$, the first case is depicted in Figure~\ref{fig:forb-attach-b}. But then a $180\degree$ fold applied to the vertical edge at the respective side of $s'$ yields an attachment which is not contained in columns $i-1$ to $i+1$ any longer, thus by the previous discussion this attachment guarantees foldability, contradicting our assumption.

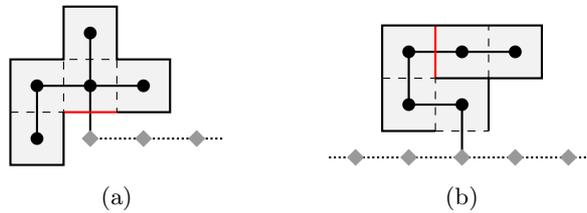
\begin{figure}[ht]
    \tikzstyle{every node}=[circle, draw, fill,
        inner sep=0pt, minimum width=4pt]
    \centering
    \begin{subfigure}{.30\textwidth}
        \centering
        \begin{tikzpicture}[xscale=0.7,yscale=0.7]
            \fill[gray!10] (-1,0)--(0,0)--(0,1)--(2,1)--(2,2)--(1,2)--(1,3)--(0,3)--(0,2)--(-1,2);
            \draw[dashed] (-1,1)--(0,1)--(0,2)--(1,2)--(1,1);
            \draw[thick, densely dotted] (0.5,0.5) -- (3,0.5);
            \draw[thick] (1.5,1.5) node {} -- (0.5,1.5) node {} -- (-0.5,1.5) node {} -- (-0.5,0.5) node {};
            \draw[thick] (0.5,0.5) -- (0.5,1.5) -- (0.5,2.5) node {};
            \draw[red,thick] (0,1) -- (1,1);
            \draw[thick] (1,1)--(2,1)--(2,2)--(1,2)--(1,3)--(0,3)--(0,2)--(-1,2)--(-1,0)--(0,0)--(0,1);
            \foreach \i in {0.5,1.5,2.5}{\node at (\i,0.5) [diamond,minimum width=6pt,minimum height=6pt,gray!80] {};}
        \end{tikzpicture}
        \caption{}
        \label{fig:forb-attach-a}
    \end{subfigure}
    \begin{subfigure}{.30\textwidth}
        \centering
        \begin{tikzpicture}[xscale=0.7,yscale=0.7]
            \fill[gray!10] (-1,1)--(1,1)--(1,2)--(2,2)--(2,3)--(-1,3);
            \draw[dashed] (-1,2)--(0,2) (0,1)--(1,1) (0,1) -- (0,2) (1,2) -- (1,3);
            \draw[thick, densely dotted] (-2,0.5) -- (3,0.5);
            \draw[thick] (0.5,0.5) -- (0.5,1.5) node {} --(-0.5,1.5) node {} -- (-0.5,2.5) node {} -- (0.5,2.5) node {} -- (1.5,2.5) node {};
            \draw[red,thick] (0,2) -- (0,3);
            \draw[thick] (0,1)--(-1,1)--(-1,3)--(2,3)--(2,2)--(1,2)--(1,1) (1,2)--(0,2);
            \foreach \i in {-1.5,-0.5,0.5,1.5,2.5}{\node at (\i,0.5) [diamond,minimum width=6pt,minimum height=6pt,gray!80] {};}
        \end{tikzpicture}
        \caption{}
        \label{fig:forb-attach-b}
    \end{subfigure}
    \caption{Attachments guaranteeing foldability}
    \label{fig:forbidden-attachments-height-two}
    \end{figure}
\end{proof}

To provide a rough idea, there are (up to symmetry) 124 tree-shaped polyominoes of bounding size $3 \times 3$ and 3942 tree-shaped polyominoes of bounding size $3 \times 4$. Unfortunately, there seems to be no `nice' characterisation for those. To keep things as simple as possible, we provide in Figures \ref{fig:3times3-foldable} and \ref{fig:3times4-foldable} all minimal foldable polyominoes (up to symmetry) of the respective bounding sizes. Here minimal foldable means that removing any leaf square either reduces the bounding size or destroys foldability. These sets were determined by computer. 

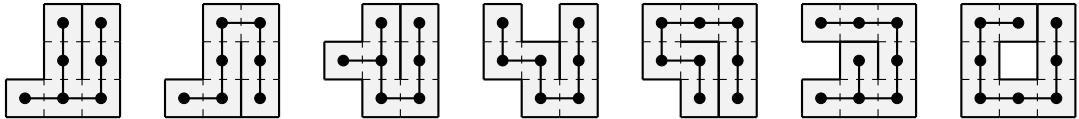
\begin{figure}[ht]
    \tikzstyle{every node}=[circle, draw, fill, inner sep=0pt, minimum width=4pt]
    \centering
    \loaddata{draw_new_3x3.txt}
    \foreach \cubenet in \loadeddata {
        \begin{subfigure}{0.125\textwidth}
        \centering
        \begin{tikzpicture}[xscale=0.5,yscale=0.5]
            \foreach \a/\b/\c/\d/\e/\f in \cubenet {
                \fill[gray!10] (\a,\b) -- (\a+1,\b) -- (\a+1,\b+1) -- (\a,\b+1);
                }
            \foreach \a/\b/\c/\d/\e/\f in \cubenet {
                \ifthenelse{\equal{\c}{0}}
                    {
                        \draw[thick] (\a,\b+1)--(\a+1,\b+1);
                    }
                    {
                        \draw[thick] (\a+0.5,\b+0.5)--(\a+0.5,\b+1.5);
                        \draw[dashed] (\a,\b+1)--(\a+1,\b+1); 
                    }
                \ifthenelse{\equal{\d}{0}}
                    {
                        \draw[thick] (\a,\b)--(\a,\b+1);
                    }
                    {
                        \draw[dashed] (\a,\b)--(\a,\b+1);
                        \draw[thick] (\a+0.5,\b+0.5)--(\a-0.5,\b+0.5); 
                    }
                \ifthenelse{\equal{\e}{0}}
                    {
                        \draw[thick] (\a+1,\b)--(\a,\b);
                    }{}
                \ifthenelse{\equal{\f}{0}}
                    {
                        \draw[thick] (\a+1,\b+1)--(\a+1,\b);
                    }
                    {}
                \draw (\a+0.5,\b+0.5) node{};
            }
        \end{tikzpicture}
        \caption*{}
        \end{subfigure}
    }
    \vspace{-0.5cm}
    \caption{The 7 minimal foldable tree-shaped polyominoes of bounding size $3 \times 3$.}
    \label{fig:3times3-foldable}
\end{figure}

\begin{figure}[p]
    \tikzstyle{every node}=[circle, draw, fill, inner sep=0pt, minimum width=4pt]
    \centering
    \loaddata{draw_new_3x4.txt}
    \foreach \cubenet in \loadeddata {
        \begin{subfigure}{0.17\textwidth}
        \centering
        \begin{tikzpicture}[xscale=0.5,yscale=0.5]
            \foreach \a/\b/\c/\d/\e/\f in \cubenet {
                \fill[gray!10] (\a,\b) -- (\a+1,\b) -- (\a+1,\b+1) -- (\a,\b+1);
                }
            \foreach \a/\b/\c/\d/\e/\f in \cubenet {
                \ifthenelse{\equal{\c}{0}}
                    {
                        \draw[thick] (\a,\b+1)--(\a+1,\b+1);
                    }
                    {
                        \draw[thick] (\a+0.5,\b+0.5)--(\a+0.5,\b+1.5);
                        \draw[dashed] (\a,\b+1)--(\a+1,\b+1); 
                    }
                \ifthenelse{\equal{\d}{0}}
                    {
                        \draw[thick] (\a,\b)--(\a,\b+1);
                    }
                    {
                        \draw[dashed] (\a,\b)--(\a,\b+1);
                        \draw[thick] (\a+0.5,\b+0.5)--(\a-0.5,\b+0.5); 
                    }
                \ifthenelse{\equal{\e}{0}}
                    {
                        \draw[thick] (\a+1,\b)--(\a,\b);
                    }{}
                \ifthenelse{\equal{\f}{0}}
                    {
                        \draw[thick] (\a+1,\b+1)--(\a+1,\b);
                    }
                    {}
                \draw (\a+0.5,\b+0.5) node{};
            }
        \end{tikzpicture}
        \vspace{0.6cm}
        \end{subfigure}
    }
    \vspace{-0.5cm}
    \caption{The 45 minimal foldable tree-shaped polyominoes of bounding size $3 \times 4$.}
    \label{fig:3times4-foldable}
\end{figure}
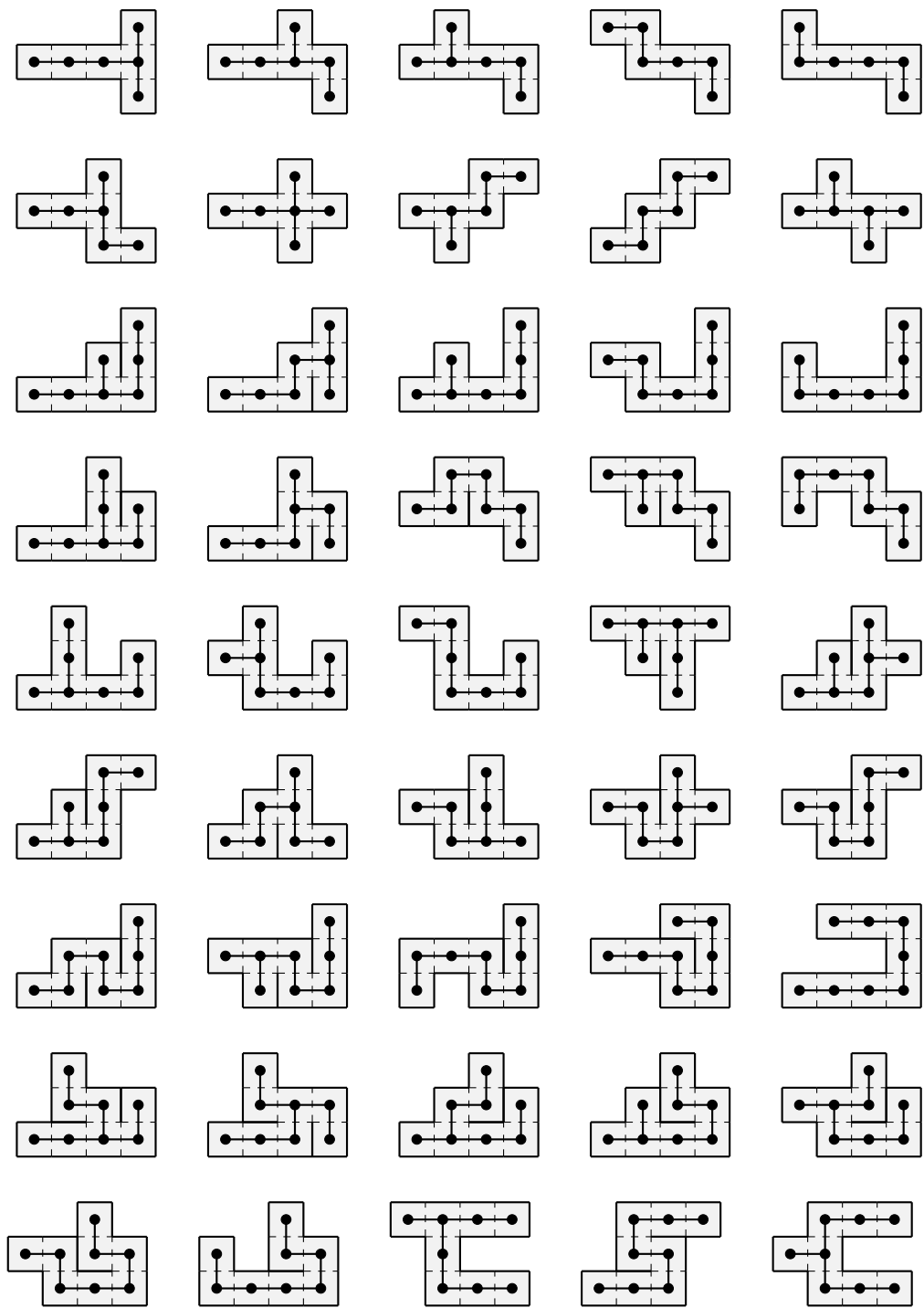

\subsection[Folding tree-shaped polyomino of bounding size at least 4 x 4]{Folding tree-shaped polyomino of bounding size at least $4 \times 4$}

The goal of this section is to show that the polyomino $P_W$ depicted in Figure~\ref{fig:4times4-unfoldable} is the only tree-shaped polyomino of bounding size at least $4 \times 4$ which does not fold onto $\Ccal$.

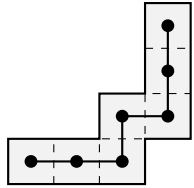
\begin{figure}[ht]
    \tikzstyle{every node}=[circle, draw, fill,
        inner sep=0pt, minimum width=4pt]
    \centering
       \begin{tikzpicture}[xscale=0.6,yscale=0.6]
        \draw[thick,fill=gray!10] (0,0)--(3,0)--(3,1)--(4,1)--(4,4)--(3,4)--(3,2)--(2,2)--(2,1)--(0,1)--cycle;
        \draw[dashed] (1,0)--(1,1) (2,0)--(2,1)--(3,1)--(3,2)--(4,2) (3,3)--(4,3);
        \draw[thick] (0.5,0.5) node {} -- (1.5,0.5) node {} -- (2.5,0.5) node {} --  (2.5,1.5) node {} --  (3.5,1.5) node {} --  (3.5,2.5) node {} --  (3.5,3.5) node {};
    \end{tikzpicture}
    \caption{Polyomino of bounding size $4 \times 4$ which does not fold onto $\Ccal$.}
    \label{fig:4times4-unfoldable}
\end{figure}

\begin{lem}
    $P_W$ does not fold onto $\Ccal$.
\end{lem}
\begin{proof}
Consider the square in the center of the polyomino $P_W$. If one of the two edges connecting it to the rest of $P_W$ is folded by $180\degree$, then the resulting shape has the form of the letter `L'. It consists of 6 squares and is not a cube net, so it cannot be folded onto $\Ccal$. Thus we may assume that both boundary edges of the central square are folded by $90\degree$. But then no folding of $P_W$ onto $\Ccal$ can cover the face directly opposite of the face covered by the central square. 
\end{proof}

\begin{lem}\label{lem:PW-with-leaf}
    Let $P$ be the polyomino $P_W$ with one additional leaf-square added. Then $P$ folds onto $\Ccal$.
\end{lem}
\begin{proof}
    By symmetry of $P_W$, we may assume without loss of generality that the additional leaf shares a vertical edge with $P_W$. All possible locations for this additional leaf are shown as red squares in Figure~\ref{fig:PW-leaf-fold}, where some sub-figures contain multiple locations as they are treated in the same way. It is not hard to check that each of the polyominoes is foldable. After applying $180\degree$ folds along the coloured edges of each polyomino, we see that each of the resulting shapes contains a $4 \times 1$ sub-polyomino, either with leaves attached on both sides or with a $3 \times 1$ sub-polyomino attached on one of its sides. In both cases folding everything by $90\degree$ yields the cube $\Ccal$.

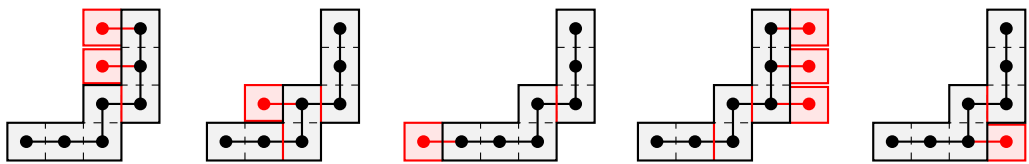
\begin{figure}[ht]
    \tikzstyle{every node}=[circle, draw, fill,
                        inner sep=0pt, minimum width=4pt]
\centering
\begin{subfigure}{.17\textwidth}
    \centering
    \begin{tikzpicture}[xscale=0.5,yscale=0.5]
        \fill[gray!10] (0,0)--(3,0)--(3,1)--(4,1)--(4,4)--(3,4)--(3,2)--(2,2)--(2,1)--(0,1);
        \fill[red!10] (3,4) -- (2,4) -- (2,3.05) -- (3,3.05);
        \fill[red!10] (3,2.95) -- (2,2.95) -- (2,2.05) -- (3,2.05);
        \draw[dashed] (1,0)--(1,1) (2,0)--(2,1)--(3,1)--(3,2)--(4,2) (3,3)--(4,3);
        \foreach \i in {2.5,3.5}{
            \draw[thick,red] (3.5,\i) -- (2.5,\i) node {};}
        \draw[thick] (0.5,0.5) node {} -- (1.5,0.5) node {} -- (2.5,0.5) node {} --  (2.5,1.5) node {} --  (3.5,1.5) node {} --  (3.5,2.5) node {} --  (3.5,3.5) node {};
        \draw[red,thick] (3,4) -- (2,4) -- (2,3.05) -- (3,3.05);
        \draw[red,thick] (3,2.95) -- (2,2.95) -- (2,2.05) -- (3,2.05);
        \draw[red,thick] (3,1) -- (3,2);
        \draw[thick] (0,0)--(3,0)--(3,1)--(4,1)--(4,4)--(3,4)--(3,2)--(2,2)--(2,1)--(0,1)--cycle;
    \end{tikzpicture}
\end{subfigure}
\begin{subfigure}{.17\textwidth}
    \centering
    \begin{tikzpicture}[xscale=0.5,yscale=0.5]
        \fill[gray!10] (0,0)--(3,0)--(3,1)--(4,1)--(4,4)--(3,4)--(3,2)--(2,2)--(2,1)--(0,1);
        \fill[red!10] (2,2) -- (1,2) -- (1,1.05) -- (2,1.05);
        \draw[dashed] (1,0)--(1,1) (2,0)--(2,1)--(3,1)--(3,2)--(4,2) (3,3)--(4,3);
        \draw[thick,red] (2.5,1.5) -- (1.5,1.5) node {};
        \draw[thick] (0.5,0.5) node {} -- (1.5,0.5) node {} -- (2.5,0.5) node {} --  (2.5,1.5) node {} --  (3.5,1.5) node {} --  (3.5,2.5) node {} --  (3.5,3.5) node {};
        \draw[red,thick] (2,2) -- (1,2) -- (1,1.05) -- (2,1.05);
        \draw[red,thick] (3,1) -- (3,2) (2,0) -- (2,1);
        \draw[thick] (0,0)--(3,0)--(3,1)--(4,1)--(4,4)--(3,4)--(3,2)--(2,2)--(2,1)--(0,1)--cycle;
    \end{tikzpicture}
\end{subfigure}
\begin{subfigure}{.20\textwidth}
    \centering
    \begin{tikzpicture}[xscale=0.5,yscale=0.5]
        \fill[gray!10] (0,0)--(3,0)--(3,1)--(4,1)--(4,4)--(3,4)--(3,2)--(2,2)--(2,1)--(0,1);
        \fill[red!10] (0,1) -- (-1,1) -- (-1,0) -- (0,0);
        \draw[dashed] (1,0)--(1,1) (2,0)--(2,1)--(3,1)--(3,2)--(4,2) (3,3)--(4,3);
        \draw[thick,red] (0.5,0.5) -- (-0.5,0.5) node {};
        \draw[thick] (0.5,0.5) node {} -- (1.5,0.5) node {} -- (2.5,0.5) node {} --  (2.5,1.5) node {} --  (3.5,1.5) node {} --  (3.5,2.5) node {} --  (3.5,3.5) node {};
        \draw[red,thick] (0,1) -- (-1,1) -- (-1,0) -- (0,0);
        \draw[red,thick] (3,1) -- (3,2);
        \draw[thick] (0,0)--(3,0)--(3,1)--(4,1)--(4,4)--(3,4)--(3,2)--(2,2)--(2,1)--(0,1)--cycle;
    \end{tikzpicture}
\end{subfigure}
\begin{subfigure}{.20\textwidth}
    \centering
    \begin{tikzpicture}[xscale=0.5,yscale=0.5]
        \fill[gray!10] (0,0)--(3,0)--(3,1)--(4,1)--(4,4)--(3,4)--(3,2)--(2,2)--(2,1)--(0,1);
        \fill[red!10] (4,4) -- (5,4) -- (5,3.05) -- (4,3.05);
        \fill[red!10] (4,2.95) -- (5,2.95) -- (5,2.05) -- (4,2.05);
        \fill[red!10] (4,1.95) -- (5,1.95) -- (5,1)-- (4,1);
        \draw[dashed] (1,0)--(1,1) (2,0)--(2,1)--(3,1)--(3,2)--(4,2) (3,3)--(4,3);
        \foreach \i in {1.5,2.5,3.5}{
            \draw[thick,red] (3.5,\i) -- (4.5,\i) node {};}
        \draw[thick] (0.5,0.5) node {} -- (1.5,0.5) node {} -- (2.5,0.5) node {} --  (2.5,1.5) node {} --  (3.5,1.5) node {} --  (3.5,2.5) node {} --  (3.5,3.5) node {};
        \draw[red,thick] (4,4) -- (5,4) -- (5,3.05) -- (4,3.05);
        \draw[red,thick] (4,2.95) -- (5,2.95) -- (5,2.05) -- (4,2.05);
        \draw[red,thick] (4,1.95) -- (5,1.95) -- (5,1)-- (4,1);
        \draw[red,thick] (3,1) -- (3,2) (2,0) -- (2,1);
        \draw[thick] (0,0)--(3,0)--(3,1)--(4,1)--(4,4)--(3,4)--(3,2)--(2,2)--(2,1)--(0,1)--cycle;
    \end{tikzpicture}
\end{subfigure}
\begin{subfigure}{.17\textwidth}
    \centering
    \begin{tikzpicture}[xscale=0.5,yscale=0.5]
        \fill[gray!10] (0,0)--(3,0)--(3,1)--(4,1)--(4,4)--(3,4)--(3,2)--(2,2)--(2,1)--(0,1);
        \fill[red!10] (3,0.95) -- (4,0.95) -- (4,0) -- (3,0);
        \draw[dashed] (1,0)--(1,1) (2,0)--(2,1)--(3,1)--(3,2)--(4,2) (3,3)--(4,3);
        \draw[thick,red] (2.5,0.5) -- (3.5,0.5) node {};
        \draw[thick] (0.5,0.5) node {} -- (1.5,0.5) node {} -- (2.5,0.5) node {} --  (2.5,1.5) node {} --  (3.5,1.5) node {} --  (3.5,2.5) node {} --  (3.5,3.5) node {};
        \draw[red,thick] (3,0.95) -- (4,0.95) -- (4,0) -- (3,0);
        \draw[red,thick] (3,1) -- (3,2);
        \draw[thick] (0,0)--(3,0)--(3,1)--(4,1)--(4,4)--(3,4)--(3,2)--(2,2)--(2,1)--(0,1)--cycle;
    \end{tikzpicture}
\end{subfigure}
\caption{Possible positions of the additional square and possible $180\degree$ folds to reduce the polyomino to a foldable shape.}
\label{fig:PW-leaf-fold}
\end{figure}
\end{proof}

\begin{thm}
\label{thm:4times4}
    Let $P$ be a tree-shaped polyomino of bounding size $4 \times 4$ which is not the polyomino $P_W$ depicted in Figure~\ref{fig:4times4-unfoldable}. Then $P$ folds onto $\Ccal$.
\end{thm}

\begin{proof}
Given a polyomino $P$ of bounding size $4 \times 4$, we start by applying leaf-folds in arbitrary order to reduce the polyomino as long as possible without reducing the bounding size. Clearly $P$ folds onto $\Ccal$ whenever the reduced shape folds onto $\Ccal$. Additionally, by Lemma~\ref{lem:PW-with-leaf} the polyomino $P$ folds onto $\Ccal$ if the reduced shape equals $P_W$ and we folded away at least one leaf during the process.

Thus it is enough to prove the statement of the theorem for \emph{reduced} polyominoes, that are polyominoes, where any leaf-fold reduces the bounding size.

Let $P$ be a reduced polyomino of bounding size $4 \times 4$ which does not fold onto $\Ccal$. We show that $P$ must be equal to $P_W$. A leaf square of $P$ is said to \emph{point towards} the top/left/bottom/right boundary, if it lies in the topmost row/leftmost column/bottommost row/rightmost column of the $4 \times 4$ bounding box and its neighbour does not lie in this row or column. The following two observations follow easily from the reducedness of $P$.
\begin{enumerate}[label=Observation \arabic*. , ref=\arabic*, wide=0pt]
\item Every leaf square of $P$ points towards a boundary. \label{itm:boundaryleaf}
\item If a leaf square points towards a boundary, then the respective row or column belonging to this boundary does not contain any other squares of $P$. \label{itm:boundaryleaf2}
\end{enumerate}

We organize our proof as a sequence of claims.

\begin{clm}
The polyomino $P$ does not contain a $4 \times 1$ strip as a sub-polyomino.
\end{clm}

Indeed, if it contains a $4 \times 1$ strip not lying in an extremal row or column of the bounding box, this $4 \times 1$ strip has two neighbours on different sides and thus contains a cube net of type 1-4-1. If it contains a $4 \times 1$ strip lying in an extremal row, say the bottom row, then there must be a sequence of squares connected to it and reaching the top row. We can now use $180\degree$ folds along vertical edges to reduce this sequence to a $3 \times 1$ strip attached at the $4 \times 1$ strip, which can be folded onto $\Ccal$ with only $90\degree$ folds.

\begin{clm}
    The polyomino $P$ must be path-shaped.
\end{clm}

Assume for a contradiction that $P$ is not path-shaped. Then $P$ contains a square having at least 3 neighbouring squares and thus has a sub-polyomino with a T-shape. This square will henceforth be called the \emph{central square}. Assume without loss of generality that the central square has a neighbour to the left, at the bottom and to the right and that the left square lies in the leftmost column of the bounding box of $P$, see Figure~\ref{fig:4times4-a}, otherwise rotate or reflect $P$. Removing the central square from $P$ splits $P$ into 3 connected components, which will be called left, bottom and right arm of $P$. Clearly each arm has to contain at least one leaf square of $P$, which must point towards one of the four boundaries. 

Assume first that the right arm contains a square of the rightmost column of the bounding box. Let $s$ denote the square of the right arm lying in the rightmost column which is closest to the central square with respect to the graph metric in the dual of $P$. If $s$ lies above the row containing the central square, then it contains a 2-3-1 cube net as a minor. Otherwise $s$ lies below the row containing the central square. Note that in this case, the bottom arm has to contain a leaf square pointing toward the bottom boundary, thus we can apply $180\degree$ folds along horizontal edges of the right arm to reduce $P$ to a polyomino containing a $4 \times 1$ strip without reducing the bounding size, which is foldable. We conclude that the right arm cannot visit the rightmost column and thus contains only a single leaf square pointing towards the top boundary. In particular the square to the right of the central square is connected to its upper neighbour.

We are left with the case where the bottom arm visits the rightmost boundary. Then the left arm cannot visit the bottom row, otherwise we could apply $180\degree$ folds along horizontal edges of the bottom arm to reduce $P$ to a polyomino containing a cube net of type 1-4-1, which makes $P$ foldable. Thus the bottom arm also has to visit the bottom row. If the central square lies in the second row from top, we apply $180\degree$ folds along vertical edges of the bottom arm to see that the cube net of type 2-3-1 shown in Figure~\ref{fig:4times4-d} is a minor of $P$. Otherwise the central square lies in the third row from the top and we apply $180\degree$ folds along vertical edges of the right arm to see that the cube net of type 1-4-1 shown in Figure~\ref{fig:4times4-e} is a minor of $P$. This proves the claim that $P$ is path-shaped.

\begin{clm}
    The polyomino $P$ contains a $3 \times 1$ strip as a sub-polyomino.
\end{clm}

Indeed, if $P$ does not contain a $3 \times 1$ strip, it is not hard to check that $P$ either contains a cube net of type 2-2-2 or the shape in Figure~\ref{fig:4times4-g}, both of which are foldable.

We finally end up in the case where $P$ is a path and contains a $3 \times 1$ strip. Assume again without loss of generality that the $3 \times 1$ strip lies horizontally in the bounding box of $P$ and that the leftmost square lies in the leftmost column. Furthermore, as the rightmost square cannot be a leaf square by Observation~\ref{itm:boundaryleaf} we may assume that it is connected to its upper neighbour. Again the square in the center of the $3 \times 1$ strip will be called central square, and the two components of $P$ after removing this square are called left and right arm of $P$. Clearly the left arm cannot visit the rightmost column, otherwise $P$ contains the foldable polyomino in Figure~\ref{fig:4times4-h} or its reflection as a minor. Thus the right arm has to visit the right column. 

Consider the square to the left of the central square. If it is not a leaf square it is either connected to its lower or upper neighbour. In the first case we can apply $180\degree$ folds along horizontal edges of the right arm to see that $P$ contains a cube net of type 2-3-1 as a minor. In the other case the unique leaf square of $P$ contained in the left arm must points towards the top boundary. By Observation~\ref{itm:boundaryleaf2} the central square cannot lie in the second row from the top, thus  $P$ can be reduced to the polyomino in Figure~\ref{fig:4times4-j} by only $180\degree$ folds, which with the additional fold shown in the figure reduces to a cube net of type 2-3-1.

We are left with the case where the square to the left of the central square is a leaf square. If the right arm touches both the top and the bottom boundary, we can use $180\degree$ folds along vertical edges to reduce it to a polyomino containing a vertical $4 \times 1$ strip, which is foldable. Thus the central square lies in the bottom row of the boundary box of $P$. By Observation~\ref{itm:boundaryleaf} the square at the top right from the central square cannot be a leaf square. If it is not connected to its right neighbour, $P$ contains one of the polyominoes in Figure~\ref{fig:4times4-k} and Figure~\ref{fig:4times4-l} as a minor, which both contain a cube net and thus are foldable. Otherwise $P$ either contains the polyomino in Figure~\ref{fig:4times4-m} as a minor, which is foldable because one additional $180\degree$ fold yields the polyomino in Figure~\ref{fig:4times4-l}, or it is equal to the polyomino $P_W$.

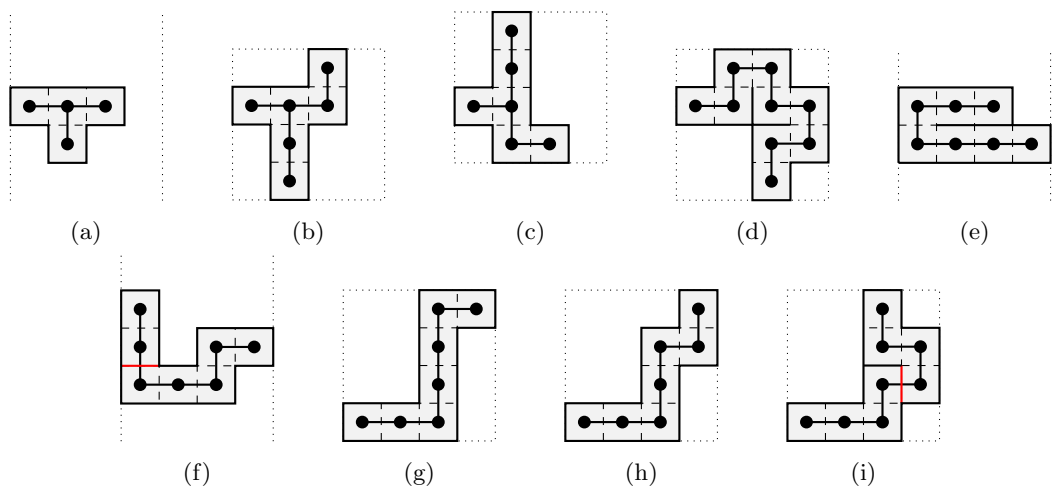
\begin{figure}[ht]
    \tikzstyle{every node}=[circle, draw, fill,
                        inner sep=0pt, minimum width=4pt]
\centering
\begin{subfigure}[b]{.19\textwidth}
    \centering
    \begin{tikzpicture}[xscale=0.5,yscale=0.5]
        \draw[dotted] (0,-2) -- (0,3);
        \draw[dotted] (4,-2) -- (4,3);
        \fill[gray!10] (0,0)--(1,0)--(1,-1)--(2,-1)--(2,0)--(3,0)--(3,1)--(0,1);
        \draw[dashed] (1,1)--(1,0)--(2,0)--(2,1);
        \draw[thick] (0.5,0.5) node {} -- (1.5,0.5) node {} -- (1.5,-0.5) node {} (1.5,0.5) --  (2.5,0.5) node {};
        \draw[thick] (0,0)--(1,0)--(1,-1)--(2,-1)--(2,0)--(3,0)--(3,1)--(0,1)--cycle;
    \end{tikzpicture}
    \caption{}
    \label{fig:4times4-a}
\end{subfigure}
\begin{subfigure}[b]{.19\textwidth}
    \centering
    \begin{tikzpicture}[xscale=0.5,yscale=0.5]
        \draw[dotted] (0,-2)--(0,2)--(4,2)--(4,-2)--cycle;
        \fill[gray!10] (0,0)--(1,0)--(1,-2)--(2,-2)--(2,0)--(3,0)--(3,2)--(2,2)--(2,1)--(0,1);
        \draw[dashed] (1,1)--(1,0)--(2,0)--(2,1)--(3,1) (1,-1)--(2,-1);
        \draw[thick] (0.5,0.5) node {} -- (1.5,0.5) node {} -- (1.5,-0.5) node {} -- (1.5,-1.5) node {} (1.5,0.5) --  (2.5,0.5) node {} --  (2.5,1.5) node {};
        \draw[thick] (0,0)--(1,0)--(1,-2)--(2,-2)--(2,0)--(3,0)--(3,2)--(2,2)--(2,1)--(0,1)--cycle;
    \end{tikzpicture}
    \caption{}
    \label{fig:4times4-d}
\end{subfigure}
\begin{subfigure}[b]{.19\textwidth}
    \centering
    \begin{tikzpicture}[xscale=0.5,yscale=0.5]
        \draw[white] (0,-2);
        \draw[dotted] (0,-1)--(0,3)--(4,3)--(4,-1)--cycle;
        \fill[gray!10] (0,0)--(1,0)--(1,-1)--(3,-1)--(3,0)--(2,0)--(2,3)--(1,3)--(1,1)--(0,1);
        \draw[dashed] (2,-1)--(2,0)--(1,0)--(1,1)--(2,1) (1,2)--(2,2);
        \draw[thick] (0.5,0.5) node {} -- (1.5,0.5) node {} -- (1.5,-0.5) node {} -- (2.5,-0.5) node {} (1.5,0.5) --  (1.5,1.5) node {} --  (1.5,2.5) node {};
        \draw[thick] (0,0)--(1,0)--(1,-1)--(3,-1)--(3,0)--(2,0)--(2,3)--(1,3)--(1,1)--(0,1)--cycle;
    \end{tikzpicture}
    \caption{}
    \label{fig:4times4-e}
\end{subfigure}
\begin{subfigure}[b]{.19\textwidth}
    \centering
    \begin{tikzpicture}[xscale=0.5,yscale=0.5]
        \draw[dotted] (0,-1)--(0,3)--(4,3)--(4,-1)--cycle;
        \fill[gray!10] (0,1)--(2,1)--(2,-1)--(3,-1)--(3,0)--(4,0)--(4,2)--(3,2)--(3,3)--(1,3)--(1,2)--(0,2);
        \draw[dashed] (1,1)--(1,2)--(3,2)--(3,0)--(2,0) (2,2)--(2,3) (3,1)--(4,1);
        \draw[thick] (0.5,1.5) node {} -- (1.5,1.5) node {} -- (1.5,2.5) node {} -- (2.5,2.5) node {} -- (2.5,1.5) node {} -- (3.5,1.5) node {} -- (3.5,0.5) node {} -- (2.5,0.5) node {} --  (2.5,-0.5) node {};
        \draw[thick] (0,1)--(2,1)--(2,-1)--(3,-1)--(3,0)--(4,0)--(4,2)--(3,2)--(3,3)--(1,3)--(1,2)--(0,2)--cycle;
        \draw[thick] (2,2)--(2,1)--(3,1);
    \end{tikzpicture}
    \caption{}
    \label{fig:4times4-g}
\end{subfigure}
\begin{subfigure}[b]{.19\textwidth}
    \centering
    \begin{tikzpicture}[xscale=0.5,yscale=0.5]
        \draw[dotted] (0,-2) -- (0,2);
        \draw[dotted] (4,-2) -- (4,2);
        \fill[gray!10] (0,-1)--(4,-1)--(4,0)--(3,0)--(3,1)--(0,1);
        \draw[dashed] (0,0)--(1,0) (1,-1)--(1,1) (2,-1)--(2,1) (3,-1)--(3,0);
        \draw[thick] (2.5,0.5) node {} -- (1.5,0.5) node {} -- (0.5,0.5) node {} -- (0.5,-0.5) node {} -- (1.5,-0.5) node {} -- (2.5,-0.5) node {} -- (3.5,-0.5) node {};
        \draw[thick] (0,-1)--(4,-1)--(4,0)--(1,0) (3,0)--(3,1)--(0,1)--(0,-1);
    \end{tikzpicture}
    \caption{}
    \label{fig:4times4-h}
\end{subfigure}
\begin{subfigure}[b]{.19\textwidth}
    \centering
    \begin{tikzpicture}[xscale=0.5,yscale=0.5]
        \draw[dotted] (0,-1) -- (0,4);
        \draw[dotted] (4,-1) -- (4,4);
        \fill[gray!10] (1,3)--(0,3)--(0,0)--(3,0)--(3,1)--(4,1)--(4,2)--(2,2)--(2,1)--(1,1);
        \draw[dashed] (0,2)--(1,2) (0,1)--(1,1)--(1,0) (2,0)--(2,1)--(3,1)--(3,2);
        \draw[thick] (0.5,2.5) node {} -- (0.5,1.5) node {} -- (0.5,0.5) node {} -- (1.5,0.5) node {} -- (2.5,0.5) node {} -- (2.5,1.5) node {} -- (3.5,1.5) node {};
        \draw[red,thick] (0,1)--(1,1);
        \draw[thick] (1,3)--(0,3)--(0,0)--(3,0)--(3,1)--(4,1)--(4,2)--(2,2)--(2,1)--(1,1)--cycle;
    \end{tikzpicture}
    \caption{}
    \label{fig:4times4-j}
\end{subfigure}
\begin{subfigure}[b]{.19\textwidth}
    \centering
    \begin{tikzpicture}[xscale=0.5,yscale=0.5]
        \draw[white] (0,3);
        \draw[dotted] (0,-2)--(0,2)--(4,2)--(4,-2)--cycle;
        \fill[gray!10] (0,-2)--(3,-2)--(3,1)--(4,1)--(4,2)--(2,2)--(2,-1)--(0,-1);
        \draw[dashed] (1,-2)--(1,-1) (2,-2)--(2,-1)--(3,-1) (2,0)--(3,0) (2,1)--(3,1)--(3,2);
        \draw[thick] (0.5,-1.5) node {} -- (1.5,-1.5) node {} -- (2.5,-1.5) node {} -- (2.5,-0.5) node {} -- (2.5,0.5) node {} -- (2.5,1.5) node {} -- (3.5,1.5) node {};
        \draw[thick] (0,-2)--(3,-2)--(3,1)--(4,1)--(4,2)--(2,2)--(2,-1)--(0,-1)--cycle;
    \end{tikzpicture}
    \caption{}
    \label{fig:4times4-k}
\end{subfigure}
\begin{subfigure}[b]{.19\textwidth}
    \centering
    \begin{tikzpicture}[xscale=0.5,yscale=0.5]
        \draw[dotted] (0,-2)--(0,2)--(4,2)--(4,-2)--cycle;
        \fill[gray!10] (0,-2)--(3,-2)--(3,0)--(4,0)--(4,2)--(3,2)--(3,1)--(2,1)--(2,-1)--(0,-1);
        \draw[dashed] (1,-2)--(1,-1) (2,-2)--(2,-1)--(3,-1) (2,0)--(3,0) (3,0)--(3,1)--(4,1);
        \draw[thick] (0.5,-1.5) node {} -- (1.5,-1.5) node {} -- (2.5,-1.5) node {} -- (2.5,-0.5) node {} -- (2.5,0.5) node {} -- (3.5,0.5) node {} -- (3.5,1.5) node {};
        \draw[thick] (0,-2)--(3,-2)--(3,0)--(4,0)--(4,2)--(3,2)--(3,1)--(2,1)--(2,-1)--(0,-1)--cycle;
    \end{tikzpicture}
    \caption{}
    \label{fig:4times4-l}
\end{subfigure}
\begin{subfigure}[b]{.19\textwidth}
    \centering
    \begin{tikzpicture}[xscale=0.5,yscale=0.5]
        \draw[dotted] (0,-2)--(0,2)--(4,2)--(4,-2)--cycle;
        \fill[gray!10] (0,-2)--(3,-2)--(3,-1)--(4,-1)--(4,1)--(3,1)--(3,2)--(2,2)--(2,-1)--(0,-1);
        \draw[dashed] (1,-2)--(1,-1) (2,-2)--(2,-1)--(3,-1)--(3,1)--(2,1) (3,0)--(4,0);
        \draw[thick] (0.5,-1.5) node {} -- (1.5,-1.5) node {} -- (2.5,-1.5) node {} -- (2.5,-0.5) node {} -- (3.5,-0.5) node {} -- (3.5,0.5) node {} -- (2.5,0.5) node {} -- (2.5,1.5) node {};
        \draw[red,thick] (3,-1)--(3,0);
        \draw[thick] (0,-2)--(3,-2)--(3,-1)--(4,-1)--(4,1)--(3,1)--(3,2)--(2,2)--(2,-1)--(0,-1)--cycle;
        \draw[thick] (2,0)--(3,0);
    \end{tikzpicture}
    \caption{}
    \label{fig:4times4-m}
\end{subfigure}
\caption{Polyominoes encountered in the proof of Theorem~\ref{thm:4times4}.}
\label{fig:4times4-cases}
\end{figure}
\end{proof}

\begin{thm}
\label{thm:4times5}
    Let $P$ be a tree-shaped polyomino of bounding size at least $4 \times 5$. Then $P$ folds onto $\Ccal$.
\end{thm}
\begin{proof}
For a given tree-shaped polyomino $P$ of bounding size at least $4 \times 5$ we consider the following 3-step approach.

Similar to the proof of the previous theorem, we start by applying leaf-folds in arbitrary order such that the resulting flat tree-shaped polyomino has bounding size at least $4 \times 4$. If no such leaf-fold is possible any more, we call the polyomino \emph{reduced}. Clearly any reduced polyomino $P'$ is still tree-shaped and has bounding size $4 \times k$ for some $k \geq 4$. If $P'$ is equal to the polyomino $P_W$ in Figure~\ref{fig:4times4-unfoldable}, we undo the last leaf-fold and see that the resulting polyomino and thus also the original polyomino $P$ is foldable by Lemma~\ref{lem:PW-with-leaf}. Thus it is enough to show that any reduced polyomino which is not equal to $P_W$ can be folded onto the cube. We use induction on $k \geq 4$.

For the induction basis assume first that the resulting tree-shaped polyomino $P'$ has bounding size $4 \times 4$. If this shape is not the polyomino $P_W$ depicted in Figure~\ref{fig:4times4-unfoldable} we are done by Theorem~\ref{thm:4times4} and can obtain a proper cube folding.

Assume now that the reduced tree-shaped polyomino $P'$ has bounding size $4 \times k$ for some $k>4$ (the case $k \times 4$ is symmetric, and bounding boxes where both dimensions are larger than 4 for sure have a leaf-fold.) We make a sequence of useful observations.
\begin{enumerate}[label=Observation \arabic*. , wide=0pt]
\item In the top (bottom) row of the bounding box there is only one leaf square and no other polyomino squares. Otherwise we could make a leaf-fold without decreasing the dimension of the bounding box below 4.
\item Except the two leaf squares on top and bottom, there is no other leaf square. Otherwise we could leaf-fold it without decreasing the dimension of the bounding box below 4.
\item Since there are only two leaf squares the polyomino $P'$ is path-shaped.
\item If there is no leaf square in the leftmost column, then the two squares in the left most column are in rows 2 and 3 and are connected in a C-shaped way together and to the right.
\item If there is a leaf square in the leftmost column, then this is connected in an L-shaped way to the right. Note that there are several possibilities. The L could be 2 or 3 rows high, having the leaf square on top or bottom, or there could even be both leaf squares in the leftmost columns with two L-shaped connections of height 2 each.
\end{enumerate}

In all cases we can again fold the leftmost column with a $180\degree$ fold into another path-shaped polyomino $P''$, reducing the bounding size to $4 \times (k-1)$. See Figure~\ref{fig:5times4-left-column} for the different possible foldings. Note that for the third polyomino in the figure, we first apply the $180\degree$ fold drawn in the figure to reduce it to the fourth polyomino in the figure. 

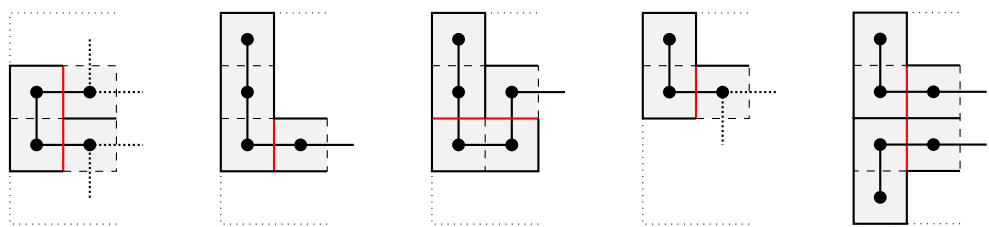
\begin{figure}[ht]
\tikzstyle{every node}=[circle, draw, fill,
    inner sep=0pt, minimum width=4pt]
\centering
\begin{subfigure}{.18\textwidth}
    \centering
    \begin{tikzpicture}[xscale=0.7,yscale=0.7]
        \draw[dotted] (2,0) -- (0,0) -- (0,4) -- (2,4);
        \fill[gray!10] (0,1)--(2,1)--(2,3)--(0,3);
        \draw[dashed] (0,2) -- (1,2);
        \draw[dashed] (1,1) -- (2,1) -- (2,3) -- (1,3);
        \draw[thick, densely dotted] (1.5,0.5) -- (1.5,1.5) -- (2.5,1.5);
        \draw[thick, densely dotted] (1.5,3.5) -- (1.5,2.5) -- (2.5,2.5);
        \draw[thick] (1.5,1.5) node {} -- (0.5,1.5) node {} -- (0.5,2.5) node {} -- (1.5,2.5) node {};
        \draw[red,thick] (1,1) -- (1,3);
        \draw[thick] (1,1) -- (0,1) -- (0,3) -- (1,3);
        \draw[thick] (1,2) -- (2,2);
    \end{tikzpicture}
\end{subfigure}
\begin{subfigure}{.18\textwidth}
    \centering
    \begin{tikzpicture}[xscale=0.7,yscale=0.7]
        \draw[dotted] (2,0) -- (0,0) -- (0,4) -- (2,4);
        \fill[gray!10] (0,1)--(2,1)--(2,2)--(1,2)--(1,4)--(0,4);
        \draw[dashed] (0,3) -- (1,3) (0,2) -- (1,2) (2,2) -- (2,1);
        \draw[thick] (0.5,3.5) node {} -- (0.5,2.5) node {} -- (0.5,1.5) node {} -- (1.5,1.5) node {} -- (2.5,1.5);
        \draw[red,thick] (1,1) -- (1,2);
        \draw[thick] (2,1) -- (0,1) -- (0,4) -- (1,4) -- (1,2) -- (2,2);
    \end{tikzpicture}
\end{subfigure}
\begin{subfigure}{.18\textwidth}
    \centering
    \begin{tikzpicture}[xscale=0.7,yscale=0.7]
        \draw[dotted] (2,0) -- (0,0) -- (0,4) -- (2,4);
        \fill[gray!10] (0,1)--(2,1)--(2,3)--(1,3)--(1,4)--(0,4);
        \draw[dashed] (0,3) -- (1,3) (1,2) -- (1,1) (2,2) -- (2,3);
        \draw[thick] (0.5,3.5) node {} -- (0.5,2.5) node {} -- (0.5,1.5) node {} -- (1.5,1.5) node {} -- (1.5,2.5) node {} -- (2.5,2.5);
        \draw[red,thick] (0,2) -- (2,2);
        \draw[thick] (2,2) -- (2,1) -- (0,1) -- (0,4) -- (1,4) -- (1,2) (1,3) -- (2,3);
    \end{tikzpicture}
\end{subfigure}
\begin{subfigure}{.18\textwidth}
    \centering
    \begin{tikzpicture}[xscale=0.7,yscale=0.7]
        \draw[dotted] (2,0) -- (0,0) -- (0,4) -- (2,4);
        \fill[gray!10] (0,2)--(2,2)--(2,3)--(1,3)--(1,4)--(0,4);
        \draw[dashed] (0,3) -- (1,3) (1,2) -- (2,2) -- (2,3);
        \draw[densely dotted,thick] (2.5,2.5) -- (1.5,2.5) -- (1.5,1.5);
        \draw[thick] (0.5,3.5) node {} -- (0.5,2.5) node {} -- (1.5,2.5) node {};
        \draw[red,thick] (1,2) -- (1,3);
        \draw[thick] (1,2) -- (0,2) -- (0,4) -- (1,4) -- (1,3) -- (2,3);
    \end{tikzpicture}
\end{subfigure}
\begin{subfigure}{.18\textwidth}
    \centering
    \begin{tikzpicture}[xscale=0.7,yscale=0.7]
        \draw[dotted] (2,0) -- (0,0) -- (0,4) -- (2,4);
        \fill[gray!10] (0,0)--(1,0)--(1,1)--(2,1)--(2,3)--(1,3)--(1,4)--(0,4);
        \draw[red,thick] (1,1) -- (1,3);
        \draw[thick] (2,1) -- (1,1) -- (1,0) -- (0,0) -- (0,2) --(2,2) (0,2) -- (0,4) -- (1,4) -- (1,3) -- (2,3);
        \draw[dashed] (0,3) -- (1,3) (2,3) -- (2,1) (1,1) --(0,1);
        \draw[thick] (0.5,3.5) node {} -- (0.5,2.5) node {} -- (1.5,2.5) node {} -- (2.5,2.5);
        \draw[thick] (0.5,0.5) node {} -- (0.5,1.5) node {} -- (1.5,1.5) node {} -- (2.5,1.5);
    \end{tikzpicture}
\end{subfigure}
\caption{Possible shapes of the two leftmost columns.}
\label{fig:5times4-left-column}
\end{figure}

It is not hard to check that in all cases $P''$ is reduced, thus we can apply the induction hypothesis and see that $P'$ and thus also $P$ can be folded onto the cube. 
\end{proof}

\section{Folding simply connected polyominoes}
\label{sec:simply_connected}
A polyomino is called simply connected if it does not contain any holes. Clearly this includes all tree-shaped polyominoes. However, it appears to be much harder to find a full characterisation for the set of all foldable simply connected polyominoes. Lemma~\ref{lem:extension} leads to a first naive approach. Clearly whenever a polyomino is simply connected but not tree-like, it has to contain at least one rectangular sub-polyomino of size at least $2 \times 2$. By the mentioned lemma each of these large rectangular sub-polyominoes has to be folded by $180\degree$ to either horizontal or vertical strips of width 1. The result then is in some sense tree-shaped. However, some squares might have multiple layers, which might allow more possible foldings than the respective normal tree-shaped polyomino with the same shape. The following example should make this clear.

\begin{exa}
    Consider the leftmost polyomino in Figure~\ref{fig:simply-connected-foldable}. It contains a rectangular sub-polyomino of size $4 \times 2$. By the previous discussion either all horizontal or all vertical edges must be folded by $180\degree$. It clearly does not make sense to fold the vertical edges as this would collapse the polyomino to a $1 \times 6$ strip, which obviously cannot be folded onto $\Ccal$. Horizontally folding along the red line yields the shape in the center of the figure. When considering this a normal single layered polyomino, it is not foldable. However we can use the two layers and just fold one of the attached arms along the red edge to obtain the shape at the right, which contains the standard cube net and thus is foldable.
    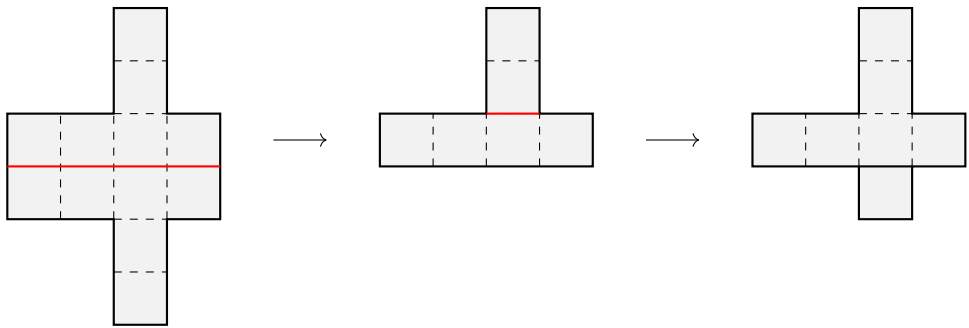
\begin{figure}[ht]
        \tikzstyle{every node}=[circle, draw, fill, inner sep=0pt, minimum width=4pt]
        \centering
        \begin{tikzpicture}[xscale=0.7,yscale=0.7]
            \begin{scope}[xshift=0,yshift=0]
                \draw[thick,fill=gray!10] (0,-1) -- (0,1) -- (2,1) -- (2,3) -- (3,3) -- (3,1) -- (4,1) -- (4,-1) -- (3,-1) -- (3,-3) -- (2,-3) -- (2,-1) -- cycle;
                \foreach \i in {-2,-1,1,2}{
                    \draw[dashed] (2,\i) -- (3,\i);
                }
                \foreach \i in {1,2,3}{
                    \draw[dashed] (\i,-1) -- (\i,1);
                }
                \draw[red,thick] (0,0) -- (4,0);
                \draw[->] (5,0.5) -- (6,0.5);
            \end{scope}
            \begin{scope}[xshift=7cm,yshift=0cm]
                \draw[thick,fill=gray!10] (0,0) -- (0,1) -- (2,1) -- (2,3) -- (3,3) -- (3,1) -- (4,1) -- (4,0) -- cycle;
                \foreach \i in {1,2}{
                    \draw[dashed] (2,\i) -- (3,\i);
                }
                \foreach \i in {1,2,3}{
                    \draw[dashed] (\i,0) -- (\i,1);
                }
                \draw[red,thick] (2,1) -- (3,1);
                \draw[->] (5,0.5) -- (6,0.5);
            \end{scope}
            \begin{scope}[xshift=14cm,yshift=-0cm]
                \draw[thick,fill=gray!10] (0,0) -- (0,1) -- (2,1) -- (2,3) -- (3,3) -- (3,1) -- (4,1) -- (4,0) -- (3,0) -- (3,-1) -- (2,-1) -- (2,0) -- cycle;
                \draw[thick] (2,0) -- (3,0);
                \foreach \i in {1,2}{
                    \draw[dashed] (2,\i) -- (3,\i);
                }
                \foreach \i in {1,2,3}{
                    \draw[dashed] (\i,0) -- (\i,1);
                }
            \end{scope}
        \end{tikzpicture}
        \caption{Multiple layers can make non-foldable shapes foldable.}
        \label{fig:simply-connected-foldable}
    \end{figure}
\end{exa}

We have seen in Section~\ref{sec:tree-shaped} that the bounding size of non-foldable tree-shaped polyominoes is bounded from above, in the sense that they either fit into a bounding box of size $3 \times n$ or of size $4 \times 4$. Clearly a similar upper bound for the bounding box cannot be given for simply connected polyominoes because a simply connected rectangular polyomino can never fold onto $\Ccal$. However, the upcoming result shows that in fact a non-foldable polyomino cannot differ a lot from a rectangular polyomino.

Let $P$ be a simply connected polyomino and let $B$ be its bounding box. For the upcoming arguments we orient both the boundary $\partial P$ of $P$ and the bounding box $B$ in clockwise direction. By definition $\partial P$ touches each of the four sides $B$. Let $v$ be a vertex of $\partial P$ which does not lie on $B$. We denote by $\pi_1(v)$ and $\pi_2(v)$ the unique paths obtained by starting at $v$ and following $\partial P$ in negative and positive direction until reaching $B$, respectively. Let $c_1$ and $c_2$ be the not necessarily different closest corners of $B$ in positive direction from the endpoint of $\pi_1(v)$ and in negative direction from the endpoint of $\pi_2(v)$, respectively. Then we say that $c_1$ and $c_2$ are \emph{valid corners} for $v$. 

This definition should become clear when looking at the example polyomino in Figure~\ref{fig:simplyconnected}. Note that $p_1(v)$ and $p_2(v)$ both end at the left boundary of the boundary box, thus the top left and the bottom left corner are valid for $v$. On the other hand, $p_1(w)$ ends at the right boundary, while $p_2(w)$ ends at the upper boundary. Thus by definition only the corner between these two boundaries, that is the top right corner, is valid for $w$.

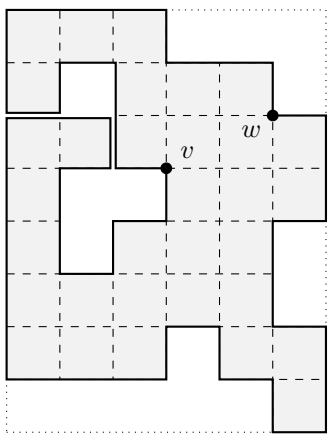
\begin{figure}[ht]
    \centering
\begin{tikzpicture}[xscale=0.7,yscale=0.7]
	\draw[dotted] (0,0) -- (0,8) -- (6,8) -- (6,0) -- cycle;
	\draw[thick,fill=gray!10] (0,1)--(0,5.95)--(1.95,5.95)--(1.95,5)--(1,5)--(1,3)--(2,3)--(2,4)--(3,4)--(3,5)--(2.05,5)--(2.05,7)--(1,7)--(1,6.05)--(0,6.05)--(0,8)--(3,8)--(3,7)--(5,7)--(5,6)--(6,6)--(6,4)--(5,4)--(5,2)--(6,2)--(6,0)--(5,0)--(5,1)--(4,1)--(4,2)--(3,2)--(3,1)--cycle;
	\node[circle, draw, fill,inner sep=0pt, minimum width=4pt] at (3,5) {};
	\node at (3.4,5.3) {$v$};
	\node[circle, draw, fill,inner sep=0pt, minimum width=4pt] at (5,6) {};
	\node at (4.6,5.7) {$w$};
    \clip (0,1)--(0,5.95)--(1.95,5.95)--(1.95,5)--(1,5)--(1,3)--(2,3)--(2,4)--(3,4)--(3,5)--(2.05,5)--(2.05,7)--(1,7)--(1,6.05)--(0,6.05)--(0,8)--(3,8)--(3,7)--(5,7)--(5,6)--(6,6)--(6,4)--(5,4)--(5,2)--(6,2)--(6,0)--(5,0)--(5,1)--(4,1)--(4,2)--(3,2)--(3,1)--cycle;
    \foreach \i in {1,...,5}{
        \draw[dashed] (\i,0)--(\i,8);
    }
    \foreach \i in {1,...,7}{
        \draw[dashed] (0,\i)--(6,\i);
    }
\end{tikzpicture}
    \caption{Valid corners for two vertices on the boundary of $P$}
    \label{fig:simplyconnected}
\end{figure}

\begin{thm}\label{thm:simply-connected}
Let $P$ be a simply connected polyomino containing a boundary point whose horizontal and vertical distance to one of its valid corners is at least 3. Then $P$ folds onto $\Ccal$.
\end{thm}
\begin{proof}
Assume without loss of generality that the top left corner of the bounding box of $P$ is valid for the point $v$ on the boundary of $P$ and that $v$ has horizontal and vertical distance at least 3 to this corner. Let $s$ denote the unit square having $v$ at its top left corner, which does not have to be present in $P$. Finally let $\pi_1(v)$ and $\pi_2(v)$ be the unique sub-paths of $\partial P$ connecting $v$ with the next points where $P$ touches the bounding box as defined above.

We fold $P$ in several steps. To assist the reader, in Figure~\ref{fig:simply-connected-folding} these steps are applied to the example polyomino of Figure~\ref{fig:simplyconnected}. Start by rolling up all rows below the row containing $s$. Next roll to the left all columns to the right of the column containing $s$. Then $s$ is the bottom right square of the resulting shape $P'$. Additionally, while $\pi_1(v)$ and $\pi_2(v)$ might have been deformed by this transformation in the case they visit rows below  or columns to the right of $s$, they still connect the top left corner of $s$ with the respective boundaries. In particular, after removing $s$, the shape $P'$ has two connected components: the \emph{left component} containing the square to the left of $s$ and the \emph{upper component} containing the square above $s$. Note that because the top left corner is valid for $v$, the left component touches the left boundary and the upper component touches the top boundary of the bounding box of $P$. 

We roll down all squares of the left component lying in a row above $s$ and roll to the right all squares of the upper component lying in columns left of $s$. Clearly all squares of the resulting shape either lie in the same row or the same column as $s$, thus the shape has the form of the letter `L'. Furthermore by construction the left component still touches the left boundary and the upper component still touches the upper boundary, thus the two components must have length at least 3. We conclude that it folds onto $\Ccal$, and thus the same is true for $P$.

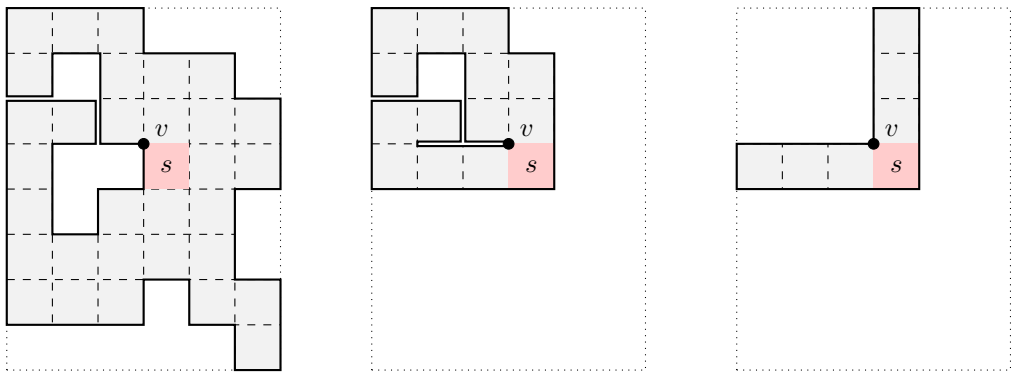
\begin{figure}[ht]
    \centering
\begin{tikzpicture}[xscale=0.6,yscale=0.6]
    \begin{scope}
        \draw[dotted] (0,0) -- (0,8) -- (6,8) -- (6,0) -- cycle;
        \begin{scope}
            \clip (0,1)--(0,5.95)--(1.95,5.95)--(1.95,5)--(1,5)--(1,3)--(2,3)--(2,4)--(3,4)--(3,5)--(2.05,5)--(2.05,7)--(1,7)--(1,6.05)--(0,6.05)--(0,8)--(3,8)--(3,7)--(5,7)--(5,6)--(6,6)--(6,4)--(5,4)--(5,2)--(6,2)--(6,0)--(5,0)--(5,1)--(4,1)--(4,2)--(3,2)--(3,1)--cycle;
            \fill[gray!10] (0,0) -- (0,8) -- (6,8) -- (6,0);
            \foreach \i in {1,...,5}{
                \draw[dashed] (\i,0)--(\i,8);
            }
            \foreach \i in {1,...,7}{
                \draw[dashed] (0,\i)--(6,\i);
            }
        \end{scope}
        \fill[red!20] (3,5)--(3,4)--(4,4)--(4,5);
        \node at (3.5,4.5) {$s$};
        \node[circle, draw, fill,inner sep=0pt, minimum width=4pt] at (3,5) {};
        \node at (3.4,5.3) {$v$};
        \draw[thick] (0,1)--(0,5.95)--(1.95,5.95)--(1.95,5)--(1,5)--(1,3)--(2,3)--(2,4)--(3,4)--(3,5)--(2.05,5)--(2.05,7)--(1,7)--(1,6.05)--(0,6.05)--(0,8)--(3,8)--(3,7)--(5,7)--(5,6)--(6,6)--(6,4)--(5,4)--(5,2)--(6,2)--(6,0)--(5,0)--(5,1)--(4,1)--(4,2)--(3,2)--(3,1)--cycle;
    \end{scope}
    \begin{scope}[xshift = 8cm]
        \draw[dotted] (0,0) -- (0,8) -- (6,8) -- (6,0) -- cycle;
        \begin{scope}
            \clip (0,4)--(0,5.95)--(1.95,5.95)--(1.95,5.05)--(1,5.05)--(1,4.95)--(3,4.95)--(3,5.05)--(2.05,5.05)--(2.05,7)--(1,7)--(1,6.05)--(0,6.05)--(0,8)--(3,8)--(3,7)--(4,7)--(4,4)--cycle;
            \fill[gray!10] (0,0) -- (0,8) -- (6,8) -- (6,0);
            \foreach \i in {1,...,5}{
                \draw[dashed] (\i,0)--(\i,8);
            }
            \foreach \i in {1,...,7}{
                \draw[dashed] (0,\i)--(6,\i);
            }
        \end{scope}
        \fill[red!20] (3,5)--(3,4)--(4,4)--(4,5);
        \node at (3.5,4.5) {$s$};
        \node[circle, draw, fill,inner sep=0pt, minimum width=4pt] at (3,5) {};
        \node at (3.4,5.3) {$v$};
        \draw[thick] (0,4)--(0,5.95)--(1.95,5.95)--(1.95,5.05)--(1,5.05)--(1,4.95)--(3,4.95)--(3,5.05)--(2.05,5.05)--(2.05,7)--(1,7)--(1,6.05)--(0,6.05)--(0,8)--(3,8)--(3,7)--(4,7)--(4,4)--cycle;
    \end{scope}
    \begin{scope}[xshift = 16cm]
        \draw[dotted] (0,0) -- (0,8) -- (6,8) -- (6,0) -- cycle;
        \begin{scope}
            \clip (0,4)--(0,5)--(3,5)--(3,8)--(4,8)--(4,4)--cycle;
            \fill[gray!10] (0,0) -- (0,8) -- (6,8) -- (6,0);
            \foreach \i in {1,...,5}{
                \draw[dashed] (\i,0)--(\i,8);
            }
            \foreach \i in {1,...,7}{
                \draw[dashed] (0,\i)--(6,\i);
            }
        \end{scope}
        \fill[red!20] (3,5)--(3,4)--(4,4)--(4,5);
        \node at (3.5,4.5) {$s$};
        \node[circle, draw, fill,inner sep=0pt, minimum width=4pt] at (3,5) {};
        \node at (3.4,5.3) {$v$};
        \draw[thick] (0,4)--(0,5)--(3,5)--(3,8)--(4,8)--(4,4)--cycle;
    \end{scope}
\end{tikzpicture}
    \caption{Folding strategy in the proof of Theorem~\ref{thm:simply-connected}}
    \label{fig:simply-connected-folding}
\end{figure}
\end{proof}

Observe that by definition every point of the boundary of a polyomino $P$ not lying inside the boundary box must have some valid corner. Thus we immediately see that whenever the boundary of $P$ contains a point at horizontal and vertical distance at least 3 from each of the four corners of the bounding box, $P$ folds onto $\Ccal$. We obtain the following corollary.

\begin{cor}
Let $P$ be a simply connected polyomino of bounding size $m \times n$ which does not fold onto $\Ccal$. Then it contains a rectangular sub-polyomino of size $(m-4) \times (n-4)$ at distance 2 from each boundary. 
\end{cor}

As mentioned before, this result has a close relation to the bounded size of non-foldable tree-shaped polyominoes. In fact, it directly implies that any non-foldable simply connected polyomino with a bounding box of size at least $6 \times 6$ contains a $2 \times 2$ rectangular sub-polyomino, or in other words, that every tree-shaped polyomino with a bounding box of size at least $6 \times 6$ folds onto $\Ccal$.

\section*{Conclusion and open problems}

For polyominoes with holes we showed in Section~\ref{sec:one_hole} that any rectangular polyomino with only one simple hole does not fold into a cube. Moreover, we provided a complete characterisation when a rectangular polyomino with two or more unit square holes (but no other holes) can be folded into a cube. In~\cite{aich19} several combinations of two simple holes were given that allow the polyomino to fold into a cube, and also examples for simple hole(s) were identified that cannot be folded into a unit cube. Restricting considerations to rectangular polyominoes we ask the following question.
\begin{opr}
For which combinations of (two or more) simple holes can a rectangular polyomino be folded into a cube?
\end{opr}

Concerning tree-shaped polyominoes we gave a full characterisation which of them can be folded onto the cube (Section~\ref{sec:tree-shaped}). This settles the case of tree-shaped polyominoes.
 
Finally, in Section~\ref{sec:simply_connected} we provided a sufficient condition when a simply-connected polyomino can be folded to a cube. Here a complete classification is still open.
\begin{opr}
Provide a complete characterisation when a simply-connected polyomino can be folded into a cube. 
\end{opr}


\bibliographystyle{abbrv}
\bibliography{latex}

\end{document}